\documentclass[12pt]{article}
\usepackage{epsfig}
\usepackage{amsfonts}
\usepackage{mathtools} 
\usepackage{amscd}
\usepackage{latexsym}
\usepackage{amsmath,amssymb,amsthm}
\usepackage{verbatim}
\usepackage{setspace}
\usepackage{color}
\usepackage{cite}
\usepackage{graphicx}
\usepackage{mathtools}
\usepackage[all]{xy}
\usepackage{tikz}

\usepackage[textheight=9in, textwidth=6.5in, letterpaper]{geometry}

\usepackage{color}   %May be necessary if you want to color links
\usepackage{hyperref}
\hypersetup{
    colorlinks=true,  %set true if you want colored links
    linktoc=all,     %set to all if you want both sections and subsections linked
    linkcolor=black,  %choose some color if you want links to stand out
    citecolor=black,
    filecolor=black,
    urlcolor=black,
}

\numberwithin{equation}{section}

\newtheorem{Theorem}{Theorem}[section]
\newtheorem*{Theorem*}{Theorem}
\newtheorem{Corollary}[Theorem]{Corollary}
\newtheorem{Lemma}[Theorem]{Lemma}
\newtheorem{Proposition}[Theorem]{Proposition}
\newtheorem{Conjecture}[Theorem]{Conjecture}
 { \theoremstyle{definition}
\newtheorem{Definition}[Theorem]{Definition}

\newtheorem{Example}[Theorem]{Example}
\newtheorem{Remark}[Theorem]{Remark} }

\def\p{\partial}
\def\cl{{\cal L}}

\def\<{\langle}
\def\>{\rangle}

\def\cO{\mathcal{O}}

\def\be{\begin{equation}}
\def\ee{\end{equation}}
\def\beq{\be\begin{array}{c}}
\def\eeq{\end{array}\ee}
\def\bes{\be\begin{split}}
\def\ees{\end{split} \ee}
\def\bs{\begin{split}}
\def\es{\end{split} }
\def\nn{\nonumber}

\def\b{{\beta}}
\def\a{{\alpha}}
\def\g{{ \gamma}}
\def\d{{\delta}}
\def\e{{\epsilon}}
\def\ve{{\varepsilon}}

   \makeatletter
  \let\over=\@@over \let\overwithdelims=\@@overwithdelims
  \let\atop=\@@atop \let\atopwithdelims=\@@atopwithdelims
  \let\above=\@@above \let\abovewithdelims=\@@abovewithdelims
\renewcommand\section{\@startsection {section}{1}{\z@}%
                                   {-3.5ex \@plus -1ex \@minus -.2ex}%nn
                                   {2.3ex \@plus.2ex}%
                                   {\normalfont\large\bfseries}}

\renewcommand\subsection{\@startsection{subsection}{2}{\z@}%
                                     {-3.25ex\@plus -1ex \@minus -.2ex}%
                                     {1.5ex \@plus .2ex}%
                                     {\normalfont\bfseries}}

\begin{document}
\begin{titlepage}
\unitlength = 1mm

\vskip 1cm
\begin{center}
{ \LARGE {\textsc{BV-refinement of the on-shell supersymmetry and localization}}}

\vspace{0.8cm}

Andrey Losev\\
\vspace{0.3cm}
{\it Shanghai Institute for Mathematics and Interdisciplinary Sciences, \\
Building 3, 62 Weicheng Road, Yangpu District, Shanghai, 200433, China\\
\vspace{0.3cm}
National Research University Higher School of Economics, Laboratory of Mirror Symmetry, NRU HSE, 6 Usacheva str., Moscow, Russia, 119048}

\vspace{1cm}
 Vyacheslav Lysov\\
 \vspace{0.3cm}
 {\it London Institute for Mathematical Sciences, \\
 Royal Institution,  21 Albemarle St, London W1S 4BS, UK}\\
  \vspace{0.3cm}
{\it Okinawa Institute of Science and Technology,\\
 1919-1 Tancha, Onna-son, Okinawa 904-0495, Japan}
\vspace{1cm}

\vspace{0.8cm}
\vspace{0.8cm}

\begin{abstract}
We generalize the BV formalism for the physical theories on supermanifolds with graded symmetry algebras realized off-shell and on-shell. An application of such generalization to supersymmetric theories allows us to formulate  the new classification which refines the usual  off-shell/on-shell   classification.  Our new classification is  based on the type of higher order antifield  terms in the corresponding BV action.   We provide  explicit examples for each class of the refined classification.   We show that the supersymmetric localization for the off-shell theories is a BV integral over particular Lagrangian submanifold.  We generalize the supersymmetric localization to the on-shell supersymmetric theories of quadratic type 
in our classification. The partition function for such theories becomes Gaussian integral for a particular choice of Lagrangian submanifold.
\end{abstract}

\end{center}

\end{titlepage}

\pagestyle{empty}
\pagestyle{plain}

\pagenumbering{arabic}

\tableofcontents

\section{Introduction}

Supersymmetric localization is a tool with  wide range of useful applications from Duistermaat-Heckman formula in math literature to supersymmetric partition functions in gauge theories. 

The key ingredients for localization are an integral over a supermanifold, an invariance of the classical  action  and measure under the parity-odd vector field transformations,  and   nilpotency of a vector field. In physics terminology the integral over supermanifold is a partition function for a physical system,  the odd vector field is a supersymmetry of the system, while the  nilpotecy means that the symmetry is realized off-shell.  Localization reduces an integral over the whole supermanifold to the integral over the neighborhood of the vector field's zeroes.

There are various means of generating the nilpotent symmetry for localization: working in off-shell description of theory, choosing a specific linear combination of supersymmetries in a theory with many supersymmetries or using a  nilpotency  up to an another symmetry and restricting to the invariant subsector of the theory.   Even with various existing methods of the nilpotent  symmetry construction,  there are some systems  with localization properties with  weaker nilpotency condition:  the symmetry vector field  squares to zero on the solutions to the equations of motion. In physics terminology such type of  symmetries are known as the on-shell  symmetries.

The most famous example of the on-shell supersymemtric system with a localization-like properties  is the  zero-dimensional theory of  superpotential. It is commonly used as an  introductory example \cite{horimirror,skinner2018supersymmetry} to supersymmetry and loop cancellation in the partition function. The partition function for such system, under  the smart change of variables, becomes a Gaussian integral and can be expressed in terms of critical points for the superpotential, which matches with the zeros of the parity-odd nilpotent vector field.

The key observation is that only a small portion of the on-shell supersymmetric systems exhibit the localization-like properties.  We suggest a refinement to the notion of the on-shell/off-shell (super)symmetry  based on the properties of the  Batalin-Vilkovisky (BV) \cite{Batalin:1981jr,Batalin:1983ggl}   description of the corresponding system.  The BV action is constructed order by order in antifields: the zeroth order is the classical action, the first order  describes the symmetry, while the higher orders describe the details of the on-shell realization of the symmetry.  Our refined  classification  distinguishes BV actions by the  highest order antifield terms: linear, quadratic, polynomial, formal series and  obstructed.  We present an explicit examples of zero-dimensional supersymmetric systems  for each class.

The  BV action is linear in antifields for the off-shell supersymmetric theories.   For many  known supersymmetric theories with on-shell supersymmetry and  good properties \cite{Baulieu:1990uv} the corresponding BV actions are  quadratic in antifields.  For a zero-dimensional theory of superpotential the BV action is also quadratic in antifields. We provide explicit zero-dimensional examples of the on-shell supersymmetric theories with polynomial and formal series  BV actions. We carefully investigate the possible obstructions to the order by order construction of  BV actions and provide examples when the BV action is obstructed at certain order in antifields. 

We recast the supersymmetric localization for the off-shell supersymmetric systems in the form of the BV integration with a particular choice of Lagrangian submanifold.  Moreover, the BV representation of localization is not restricted to the 
actions linear in antifields and can be applied to the arbitrary BV actions. In particular, we formulate new localization theorem for the  on-shell supersymmetric theories of quadratic type in refined classification.  The surprising localization of the superpotential theory in zero dimensions is an example for such theorem's application. The key point behind the localization is that the quadratic term contains only antifields for parity-odd (fermionic) variables.  Hence, a particular choice of Lagrangian submanifold can make antifields parametrically  large and the quadratic terms dominate the BV integral  turning the partition function into a Gaussian integral.

The structure of our paper is the following. We start with a brief review of  supergeometry, supersymmetry and localization in zero dimensions.  In section 3 we review the BV formalism, generalize it to the case of  theories with graded symmetries, discuss the obstructions and uniqueness of the order by order construction for the BV action. In the next section we present several concrete examples of BV actions for supersymmetric systems and  formulate  the BV-refinement for the notion of on-shell/off-shell supersymmetric theories.  In section 5 we discuss the supersymmetric localization on the language of BV integrals and  apply it to several examples.  In section 6 we recast the 1-dimensional integrals into supersymmetric partition functions and discuss the various aspects of the corresponding BV actions.

\section{$d=0$ Supersymmetry and localization}

In this section we briefly review  the supergeometry and supersymmetry in $d=0$. For more details please read \cite{deligne1999quantum,kostant2006graded,Witten:2012bg,manin2013gauge}.   In particular, we describe the  supersymmetry in zero dimensions, corresponding superspace, notions of on-shell/off-shell supersymmetry   and  localization. 

Our main focus will be the $\mathcal{N}=2$ supersymmetry in $d=0$. Although there is a simpler $\mathcal{N}=1$ supersymmetry its multiplets are too small to  differentiate  between the off-shell and on-shell realizations.

\subsection{Elements of supergeometry}

Supermanifold is defined in terms of its sheaf of functions. 

\begin{Definition} An $(n|m)$-dimensional supermanifold $\mathcal{M}$ is a sheaf $\mathcal{\cO}_{\mathcal{M}}$ on a smooth $n$-manifold $M$ (denoted as the body of $\mathcal{M}$) of graded-commutative algebras.  
The $\mathcal{\cO}_{\mathcal{M}}$ is locally isomorphic to $C^\infty(U) \otimes \wedge^\ast V^\ast$ for open $U\subset M$ and a fixed vector space $V$ of dimension $m$.
\end{Definition}

The function on a supermanifold of dimension $(n|m)$ in a  local chart  $U$ with  parity even coordinates $t^i,\;\; i=1..n$ on $U$ and parity odd coordinates $\theta^\a,\;\; \a=1..m$ is  a smooth function $f(t,\theta) $.
The  commutation relations for the coordinates 
\be
t^i t^j  = t^j t^i,\;\;\; t^i \theta^\a = \theta^\a t^i,\;\;\; \theta^\a \theta^\beta = - \theta^\b \theta^\a  
\ee
imply that all smooth functions of $\theta$ are finite polynomials.  The monomials in $\theta$ naturally carry a $\mathbb{Z}$-grading equal to the number of $\theta$'s in the corresponding monomial. The $\mathbb{Z}_2$ part of the grading for a monomial $g$ is called parity, denoted by   $|g|$ and is  used to keep track of the signs in various expressions.

\begin{Example} For a smooth $n$-manifold $M$ we can construct $(n|n)$-supermanifold $\mathcal{M} = T[1]M$ - the tangent bundle of $M$ with a shifted grading of the fibers. Let $x^i$ be local coordinates in open set  $U \subset M$.
The corresponding chart on $T[1]M$ has local coordinates $x^i$ of parity 0 and fiber coordinates $\psi^i$ of parity 1.   An element of $\mathcal{O}_{\mathcal{M}}$ locally  has the form $f(x, \psi)$. Globally we can identify functions on $T[1]M$
with differential forms on $M$,  i.e. $\cO_{\mathcal{M}} \simeq \Omega^\ast (M)$.  For a function $f(x ,\psi)$ we replace $\psi^i$ by the $dx^i$ to get  the corresponding form.
\end{Example}

\begin{Example} \label{ex_polyvectors_functions} For a smooth $n$-manifold $M$ we can construct $(n|n)$-supermanifold $\mathcal{M} = T^\ast[1]M$ - the cotangent bundle of $M$ with a shifted grading of the fibers. Let $x^i$ be local coordinates in open chart $U \subset M$.
The corresponding chart on $T^\ast[1]M$ has local coordinates $x^i$ of parity 0 and fiber coordinates $\theta_i$ of parity 1.   An element of $\mathcal{O}_{\mathcal{M}}$ locally  has the form $f(x, \theta)$. Globally we can identify functions on $T^\ast[1]M$
with polyvector fields on $M$,  i.e. $\cO_{\mathcal{M}} \simeq  \mathfrak{X} (M)$.  For a function $f(x ,\theta)$ we replace $\theta_i$ by the $\p_i$ to get  the corresponding polyvector field.
\end{Example}

\begin{Definition} A vector field $v \in \hbox{Vect}(\mathcal{M})$ of parity $|v| \in \{0,1\}$ is a derivation of $\mathcal{O}_{\mathcal{M}}$ of parity $|v|$, i.e. an $\mathbb{R}$-linear map 
$v : \cO_{\mathcal{M}} \to \cO_{\mathcal{M}}$, such that 
\begin{itemize}
\item it obeys the graded Leibniz rule 
\be
v(fg) = v(f) \cdot g + (-1)^{|v||f|} f \cdot v(g);
\ee
\item it is consistent with the  grading on functions 
\be
|v(f)| = |v| + |f|.
\ee
\end{itemize} 
\end{Definition}

Vector fields  $ \hbox{Vect}(\mathcal{M})$ on $\mathcal{M}$ form a superalgebra with Lie bracket in the form of the graded commutator, i.e. 
\be
[v,u\} = v\circ u - (-1)^{|v||u|} u\circ v.
\ee
In a local chart $U$ with coordinates $t^i, \theta^\a$ a vector field takes the form 
 \be\label{eq_local_vect_field_chart}
 v = \sum v^i (t, \theta) \frac{\p}{\p t^i} + \sum v^\a (t, \theta) \frac{\p}{\p \theta^\a}.
 \ee
Vector fields  $ \hbox{Vect}(\mathcal{M})$ also form a module over the ring of functions $\mathcal{O}_{\mathcal{M}}$.
The differential 1-forms form a dual $\mathcal{O}_{\mathcal{M}}$-module generated by $dt^j, d\theta^\b$, i.e.
\be
dt^j ( \p_{t^i})  = \delta^j_i,\;\;\; dt^j ( \p_{\theta^\a}) =0,\;\;\; d\theta^\a (\p_{t^i})  =0,\;\; d\theta^\a(\p_{\theta^\b}) = \delta^\a_\b.
\ee
The   generators  for superalgebra of differential forms  have the following supercommutative relations
\be
dt^i dt^j = -dt^j dt^i,\;\;\; d\theta^{\a} d\theta^{\beta} = d\theta^{\beta} d\theta^{\a},\;\;\; dt^i d\theta^{\a} = d\theta^{\a} dt^i.
\ee
More generally, the  differential forms $\Omega^\bullet (\mathcal{M})$ form a   supergalgebra  with the grading being the sum of parity and degree,   i.e.
\be\label{eq_diff_commute_sign}
\omega_1\cdot \omega_2 = (-1)^{(k_1+|\omega_1|)(k_2+|\omega_2|)} \omega_2\cdot \omega_1 ,\;\;\; \omega_i \in \Omega^{k_i}(\mathcal{M}).
\ee

\begin{Remark} There are several  conventions for the sign in   (\ref{eq_diff_commute_sign}) in existing literature on  BV formalism and supermanifolds. In particular, another commonly used convention  determines the sign in (\ref{eq_diff_commute_sign}) using the composition of signs for degree and parity.     Our conventions are used in \cite{manin2013gauge,Witten:2012bg,mnev2019quantum} while the alternative conventions  are used in \cite{kostant2006graded,khudaverdian2000semidensities, deligne1999quantum,henneaux1992quantization}. In the appendix to the section 1 of \cite{deligne1999quantum} authors provide a detailed analysis for the pros and cons 
of the both convention choices.
\end{Remark}
The super de Rham differential  is defined via
\be
d = dt^i \frac{\p}{\p t^i} + d\theta^{\a} \frac{\p}{\p \theta^\a}: \Omega^k(\mathcal{M}) \to \Omega^{k+1}(\mathcal{M}).
\ee
The Berezin integral for a single odd variable is defined via
\be
\int D\theta\;\; (f_0+f_1\theta) = f_1.
\ee
Berezin integral for two odd variables is defined as a pair  of single-variable  Berezin integrals 
\be
\int D\theta_2 D\theta_1\;\; f(\theta_1, \theta_2) = \int D\theta_2 \left(\int D\theta_1 \;\;f(\theta_1, \theta_2)\right).
\ee
In particular, for a generic function 
\be
\int D\theta_2 D\theta_1\; (f_0 +f_1 \theta_1+f_2 \theta_2+f_{12} \theta_1\theta_2) = \int D\theta_2\; (f_1 + f_{12} \theta_2) = f_{12}.
\ee
The higher dimensional version of Berezin integration is given by 
\be
\int D\theta_m\cdots D\theta_1 \; f(\theta_1,\ldots,\theta_n)  =  \frac{\p}{\p \theta^m} \cdots   \frac{\p}{\p \theta^2} \frac{\p}{\p \theta^1}  f. 
\ee
The  integration measure in the integral above is commonly referred to as the coordinate Berezinian 
\be
\mu = D\theta_m\cdots D\theta_1.
\ee
An integral of a function $f(t, \theta)$ over a supermanifold  $\mathcal{M}$ with body $M$ and  Berezinian $\mu  =\rho(t,\theta)d^nt D^m \theta$ is defined via 
\be
\int_{\mathcal{M}} \mu \; f = \int_{\mathcal{M}}\;  \rho(t,\theta) d^n t\; D^m \theta \; f(t,\theta) = \int_M d^n t\;  \frac{\p}{\p \theta^m} \cdots   \frac{\p}{\p \theta^2} \frac{\p}{\p \theta^1} ( f \rho).
\ee
Using Berezin  integration rules we can define a divergence of a vector field  on a supermanifold via 
\be\label{def_superman_divergence}
\int_{\mathcal{M}} \mu \; v( f) = - \int_{\mathcal{M}} \mu \; \hbox{div}_\mu v \cdot  f. 
\ee
  For  a vector field in components  (\ref{eq_local_vect_field_chart})  and Berezinian $ \mu=\rho(t,\theta)d^nt D^m \theta$ the  divergence is
 \be
 \hbox{div}_{\mu} v = \sum \frac{\p}{\p x^i} v^i - (-1)^{|v|} \sum \frac{\p}{\p \theta^\a} v^\a  + v (\ln \rho).
 \ee

\subsection{Superspace formalism}\label{sec_superspace_formalism}
Let us consider a supermanifold $\mathbb{R}^{0|2}$, the 2-dimensional odd Euclidean space,  with coordinates $\theta$ and $\bar{\theta}$ and a map
\be
\hat{x} : \mathbb{R}^{0|2} \to \mathbb{R} \;:\; (\theta, \bar{\theta}) \mapsto \hat{x}(\theta, \bar{\theta}).
\ee
Let us use  the notation  $x, F$ for even coordinates  and $\psi, \bar{\psi}$  for odd coordinates on the space of maps  $\hbox{Maps} (\mathbb{R}^{0|2} , \mathbb{R})  =  \mathbb{R}^{2|2}$, i.e. each map is written in the form 
\be
\hat{x}(\theta, \bar{\theta}) = x +\theta \bar{\psi} + \bar{\theta} \psi + F \theta \bar{\theta}.
\ee
The  function $\hat{x}(\theta, \bar{\theta})$ in physics literature is known as the  {\it superfield}. There are odd translations on  $\mathbb{R}^{0|2}$  
\be
\theta \to \theta +\e,\;\;\; \bar{\theta} \to \bar{\theta}+\bar{\e}
\ee
generated by the odd vector fields 
\be
\mathfrak{Q} = \frac{\p}{\p \theta},\;\; \bar{\mathfrak{Q}}  = \frac{\p}{\p \bar{\theta}}
\ee
which form  the   superalgebra 
\be\label{N=2_offshel_susy}
[\mathfrak{Q}, \bar{\mathfrak{Q}}\} =[\mathfrak{Q},\mathfrak{Q}\}= [\bar{\mathfrak{Q}},\bar{\mathfrak{Q}}\} = 0.
\ee
The superalgebra (\ref{N=2_offshel_susy}) in physics is known as  the {\it  $d=0$ $\mathcal{N}=2$   supersymmetry algebra}.   
\begin{Remark}
In our naming for the superalgebra we  used the terminology  from supersymmetric  sigma models. A sigma model is the (quantum) theory of   maps $\Sigma \to M$ and the supersymmetry acts on the space of maps.  The  $d=0$ refers to the parity-even dimension of the source  space, which is  $\Sigma = \mathbb{R}^{0|2}$ in our case.
The  $\mathcal{N}=2$ refers to the number of supersymmetries, the dimension of the parity-odd  part of the supersymmetry algebra.  In our case we have two odd vectorfields  $\mathfrak{Q}, \bar{\mathfrak{Q}}$, so the dimension equals to 2.
\end{Remark}
An   infinitesimal translations  on  $\mathbb{R}^{0|2}$ generate  infinitesimal translations on the space of maps via 
\be\label{eq_gen_superfield_diff_susy}
\begin{split}
 (\e \mathfrak{Q} + \bar{\e}\bar{\mathfrak{Q}})   \hat{x}(\theta,\bar{\theta}) &= \e \bar{\psi} + \bar{\e} \psi  +\theta \bar{\e} F -\bar{\theta} \e F \\
& =  \delta x + \theta \delta \bar{\psi} +\bar{\theta} \delta \psi + \theta\bar{\theta} \delta F =(\e Q +\bar{\e} Q) \hat{x}(\theta,\bar{\theta}). 
\end{split}
\ee
The action of translations on the space of maps in components 
\be\label{def_d=0_n=2_offshell_susy}
\begin{split}
\delta x &= (\e Q +\bar{\e} Q)x =  \e \bar{\psi} +\bar{\e} \psi,\;\;\;\delta \psi  =(\e Q +\bar{\e} Q)\psi =  -\e F,\\
\delta \bar{\psi} &= (\e Q +\bar{\e} Q)\bar{\psi} = \bar{\e} F,\;\;\;\delta F =(\e Q +\bar{\e} Q)F =  0.
\end{split}
\ee
 There are additional  types of superfields in our model 
\be
 \mathfrak{D} \hat{x}  = \p_\theta \hat{x} =  \psi + \bar{\theta} F, \;\; \bar{\mathfrak{D}} \hat{x}  = \p_{\bar{\theta}}\hat{x} =  \bar{\psi}  -\theta F,
\ee
in the sense that 
\be\label{eq_der_superfield_diff_susy}
 (\e \mathfrak{Q} + \bar{\e}\bar{\mathfrak{Q}}) \mathfrak{D} \hat{x}  = (\e Q +\bar{\e} Q)\mathfrak{D} \hat{x},\;\;  (\e \mathfrak{Q} + \bar{\e}\bar{\mathfrak{Q}}) \bar{\mathfrak{D}} \hat{x} = (\e Q +\bar{\e} Q)\bar{\mathfrak{D}} \hat{x},
\ee
which often named  {\it derivative superfields} or {\it fermionic fuperfields} and  {\it trivial superfield} $F =  \bar{\mathfrak{D}}\mathfrak{D}  \hat{x}$. The term ``fermionic"  is due to the lowest  component (in  $\theta$-expansion)  being  parity-odd, 
in contrast to the superfield $\hat{x}(\theta, \bar{\theta})$ with the lowest  component being parity-even. 

The key feature of the superfield formalism is an ability to construct function on the space of maps, invariant under the supersymmetry transformations (\ref{def_d=0_n=2_offshell_susy}), from an arbitrary function of superfields.

\begin{Proposition} The integral over the  superspace  $\mathbb{R}^{0|2}$ for  an  arbitrary function of $\hat{x}$ and   derivatives  is  invariant under the supersymmetry  transformations (\ref{def_d=0_n=2_offshell_susy}).
\end{Proposition}
\begin{proof} The most general superspace integral is of the form 
\be
S(x,F,\psi,\bar{\psi}) = \int d\theta d\bar{\theta}\;\;  \mathcal{F} (\hat{x},  \mathfrak{D} \hat{x} , \bar{\mathfrak{D}} \hat{x},   \bar{\mathfrak{D}}  \mathfrak{D}\hat{x})
\ee
The  supersymmetry transformation of the integral evaluates into 
\be
\begin{split}
\delta_\e S &=(\e Q + \bar{\e} \bar{Q}) S =    (\e Q + \bar{\e} \bar{Q}) \int d\theta d\bar{\theta}\;\;  \mathcal{F} (\hat{x},  \mathfrak{D} \hat{x} , \bar{\mathfrak{D}} \hat{x}, \bar{\mathfrak{D}}  \mathfrak{D}\hat{x})\\
&= \int d\theta d\bar{\theta}\;\; ( \e Q + \bar{\e} \bar{Q})\mathcal{F} (\hat{x},  \mathfrak{D} \hat{x} , \bar{\mathfrak{D}} \hat{x},  \bar{\mathfrak{D}}  \mathfrak{D}\hat{x})  \\
&= \int d\theta d\bar{\theta}\;\; ( \e\mathfrak{Q} + \bar{\e} \bar{\mathfrak{Q}}) \mathcal{F} (\hat{x},  \mathfrak{D} \hat{x} , \bar{\mathfrak{D}} \hat{x},  \bar{\mathfrak{D}}  \mathfrak{D}\hat{x})   = \int d\theta d\bar{\theta}\;\; ( \e \p_\theta + \bar{\e} \p_{\bar{\theta}}) \mathcal{F} = 0.
 \end{split}
\ee
An equality in the third line is due to the relations (\ref{eq_gen_superfield_diff_susy}, \ref{eq_der_superfield_diff_susy}) between the source  field action and the supersymmetry transformations for the superfield components.   The last equality holds due to Berezin integration rules.
\end{proof}
\begin{Example} The superspace integral  for the theory of superpotential
\be
\begin{split}
S_W &= \int d\theta d\bar{\theta}\;\;  \left[ \frac14  \mathfrak{D} \hat{x}  \bar{\mathfrak{D}}\hat{x} -H(\hat{x}) \right] = FH'(x) +  \psi \bar{\psi} H''(x)- \frac14  F^2.
\end{split}
\ee
For our analysis is convenient to make a   redefinition $W(x) = H'(x)$, so that the  action   simplifies into 
\be\label{eq_off_shell_susy_action}
S_W (x,F,\psi,\bar{\psi}) = -\frac14  F^2+  FW(x) + \psi  \bar{\psi} W'(x). 
\ee
\end{Example}

\subsection{Off-shell superpotential}\label{sub_sect_offshel_superpotential}
In our discussion of supersymmetric examples let us point out an important feature about the supersymmetry transformations.  The action  (\ref{eq_off_shell_susy_action})
is invariant under the supersymmetry transformations generated by  the vector fields
\be\label{eq_off_shell_susy_transform}
Q = \bar{\psi} \frac{\p}{\p x} - F \frac{\p }{\p \psi},\;\; \bar{Q} = \psi \frac{\p}{\p x} + F \frac{\p }{\p \bar{\psi}}.  
\ee
The algebra of vector fields 
\be\label{eq_on_shell_susy_algebra}
[ Q, \bar{Q}\} =[Q,Q\} = [\bar{Q},\bar{Q}\}= 0.
\ee
Algebra (\ref{eq_on_shell_susy_algebra})  is the  $d=0$ $\mathcal{N}=2$ supersymmetry algebra.

\subsection{On-shell superpotential}\label{sub_sect_onshel_superpotential}

We can   perform Gaussian integral over $F$  for the theory of superpotential (\ref{eq_off_shell_susy_action}), i.e. 
\be
\sqrt{-4\pi \hbar}\;e^{-\frac{1}{\hbar}S_{ind}(x,\psi,\bar{\psi})} = \int dF\; e^{-\frac{1}{\hbar}S_W(x,F,\psi, \bar{\psi})}
\ee 
to obtain another, referred to  as  the {\it induced}  (effective in physics literature) action 
\be\label{eff_action_W}
S_{ind} (x,\psi,\bar{\psi})= W^2(x) + W'(x) \psi\bar{\psi}.
\ee
The Gaussian nature of integration implies that the induced action is the original action, evaluated at the classical solution $F_{cl}$ to the equations of motions
\be\label{eq_F_class}
\frac{\p S}{\p F} = -\frac12 F + W(x) = 0\;\;\;\Longrightarrow\;\; F_{cl} =  2 W(x)
\ee
 for the  field $F$, i.e.
\be
S_{ind} (x,\psi,\bar{\psi}) = S(x,F,\psi, \bar{\psi})\Big|_{F = F_{cl} (x, \psi, \bar{\psi})}.
\ee
The induced action $S_{ind}$ is invariant under the induced supersymmetry transformations, generated by  
\be\label{eq_induced_susy}
 Q_{ind} =  \bar{\psi} \frac{\p}{\p x} - 2W(x) \frac{\p }{\p \psi},\;\; \bar{Q}_{ind}= \psi \frac{\p}{\p x} + 2W(x) \frac{\p }{\p \bar{\psi}} ,
\ee
which means that 
\be
Q_{ind} (S_{ind}) =  \bar{Q}_{ind} (S_{ind}) =0.
\ee
Induced supersymmetry  vector fields  (\ref{eq_induced_susy})   are  supersymmetries   (\ref{eq_off_shell_susy_transform}) evaluated on  the classical solution (\ref{eq_F_class}), i.e.
\be
 Q_{ind}  = Q\Big|_{F = F_{cl} (x, \psi, \bar{\psi})},\;\; \bar{Q}_{ind}  = \bar{Q}\Big|_{F = F_{cl} (x, \psi, \bar{\psi})}.
\ee
Note that the superalgebra of induced symmetries  is different from the superalegbra (\ref{eq_on_shell_susy_algebra})
\be\label{eq_on_shell_susy_algebra_supermult}
\begin{split}
[Q_{ind}, Q_{ind}\} &=-4W'(x) \bar{\psi} \frac{\p}{\p \psi} ,\;\;\;[\bar{Q}_{ind}, \bar{Q}_{ind}\} =4W'(x) \psi \frac{\p}{\p \bar{\psi}}  ,\\
 [ \bar{Q}_{ind},Q_{ind}\} &= 2 W'(x) \left(  \bar{\psi} \frac{\p}{\p \bar{\psi}}-\psi \frac{\p}{\p \psi}\right).
\end{split}
\ee
The supersymmetry algebra (\ref{eq_on_shell_susy_algebra_supermult}) in the physics notations is  $d=0$ $\mathcal{N}=2$ {\it on-shell} supersymmetry algebra. 
The term {\it on-shell} indicates that the nontrivial  commutators of supersymmetries  vanish on the solutions to  the equations of motion for induced  action $S_{ind}$ 
\be
\frac{\p S_{ind}}{ \p\psi} = W'(x) \bar{\psi},\;\;\; \frac{\p S_{ind}}{ \p \bar{\psi}} = -W'(x) \psi.
\ee
Therefore, as long as equations of motions are satisfied, the on-shell supersymmetry algebra (\ref{eq_on_shell_susy_algebra_supermult}) is the same as the off-shell algebra  (\ref{eq_on_shell_susy_algebra}).
The theory with the action (\ref{eff_action_W}) and on-shell supersymmetries (\ref{eq_on_shell_susy_algebra_supermult}) is known as the {\it on-shell theory of superpotential}.

\subsection{Localization} \label{sec_localization}

In this section we will briefly review the key methods for supersymmetric localization, for further details see \cite{pestun2017introduction}.

The quantum properties of a classical physical system are often captured (according to the Feynman's approach) in terms of partition function $Z$ (usually with extra sources) defined by 
\be\label{def_susy_part_fun}
Z = \int_X \mu_X\; e^{-\frac{1}{\hbar} S(\phi)}
\ee
for the integration measure $\mu_X$ on $X$. 

For a physical system with a symmetry, the Lie group  $G$ action on $X$, preserving the action $S$, the partition function integral can be simplified.  
If the action of $G$ is free then the integral reduces to the lower dimensional integral over the space $X/G$ of $G$-orbits on $X$ multiplied by the volume of the group, i.e
\be\label{eq_part_funct_reduction_orbits}
\int_X \mu_X\; e^{-\frac{1}{\hbar} S(\phi)} =  \hbox{Vol}(G) \int_{X/G} \mu_{X/G}\; e^{-\frac{1}{\hbar} S(\phi)},
\ee
The volume of the group $G$   is determined using the Haar measure on $G$. If there are regions where action is not free we need to treat them separately.

The key property of supersymmetric theories is that the contribution from the regions with the free symmetry action vanishes and the partition function equals to the integrals over neighborhoods of the fixed points. Such phenomenon is known as the localization: the supersymmetric partition function localizes to the zeroes of the supersymmetry vector field.

\begin{Proposition}
The integral over superspace region where the action of the supersymmetry  vector field  is free vanishes.
\end{Proposition}
\begin{proof}   In the region with the free action we use the equality (\ref{eq_part_funct_reduction_orbits}) to reduce the original integral to the integral over the space of orbits.   The Haar measure for the group of supersymmetry transformations $G = \mathbb{R}^{0|1}$ is the translation-invariant measure, i.e. $\mu_{G}(\e) = d\e$. The volume of the group with respect to the Haar measure  vanishes due to the Berezin rules for the  integration over odd variables, i.e.
\be
\hbox{vol}(\mathbb{R}^{0|1}) = \int_{\mathbb{R}^{0|1}} \mu_G(\e) = \int_{\mathbb{R}^{0|1} } d\e = 0.
\ee
Since the volume vanishes and the integral over $X/G$ is finite, then  the integral over superspace with the free action of supersymmetry  vanishes and the proof is complete.
\end{proof}

Our arguments on supersymmetric localization can further refined in the form the propositions below.
\begin{Proposition}\label{prop_susy_localization} For a physical system on a supermanifold $X$ with action $S\in C^\infty(X)_{even}$ and measure $\mu_X$, invariant under the divergence-free odd nilpotent  vector field $Q$, i.e. $Q^2=div_{\mu_X} Q=0$,  the deformed partition  function 
\be
Z(t) = \int_X \mu_X\;\; e^{-S(\phi) -tQ(V)} 
\ee
 is independent on $t$ for arbitrary an Grassmann-odd function $V \in C^\infty(X)_{odd}$.
\end{Proposition}
 \begin{proof} The $t$-derivative of  the deformed partition function evaluates into 
\be\label{eq_t_deriv_part_function}
\begin{split}
\p_t Z(t) &=  \p_t \int \mu_X\;\; e^{-S(\phi) -tQ(V)}   = -  \int \mu_X\;\; Q(V) e^{-S -tQ(V)}  \\
&= -  \int \mu_X\;\; Q(V e^{-S}) e^{ -tQ(V)} = -  \int \mu_X\;\; Q(V e^{-S} e^{ -tQ(V)}) =  0.
\end{split}
\ee
We used the invariance of action, $Q(S)=0$, in the third equality. For the fourth equality we simplified the vector field action   
\be\nn
Q (e^{-tQ(V)}) = Q \left(1 - t Q(V) + \frac{t^2}2 Q(V) Q(V)  +\ldots \right)=-t Q^2(V) e^{- Q(V)}=0
\ee 
and used the nilpotency condition in the last equality.
In the last equality of (\ref{eq_t_deriv_part_function}) we used the definition of divergence (\ref{def_superman_divergence}) and divergence-free assumption for the odd vector field $Q$ to complete the proof.
\end{proof}

\begin{Corollary} The deformed partition function for $t=0$ matches with the supersymmetric partition function $Z$.  Moreover, we can  take $t\to\infty$ limit without changing the partition function to get an equality
\be
Z =  Z(0) = Z(t) = \lim_{t\to \infty} Z(t) .
\ee
\end{Corollary} 
For positive definite  even function $Q(V)$ the limit   $t\to \infty$ simplifies the   integral ito
\be
 \lim_{t\to \infty} \int_X \mu_X\;\; e^{-S -tQ(V)}  =  \int_{X_{Q,V}} \mu_X\;\; e^{-S} 
\ee
The   integration region changes from  $X$  to the neighborhood $X_{Q,V} = \{ p \in X\;|\; Q(V)(p) <\sigma \}$ of size $\sigma>0$. 
Authors in \cite{Schwarz:1995dg} showed that  the  assumptions of the proposition \ref{prop_susy_localization}  are enough  to choose $V$, such that the localization subspace 
$X_{Q,V}$ is the  neighborhood of the zeroes for the odd vector field $Q$.

\begin{Remark} We can generalize the localization construction for the case $Q^2\neq 0$, if we further restrict the choice of  $V$ by an additional  condition  $Q^2 (V) = 0$.

\end{Remark}

\begin{Example}\label{ex_off_shell_deform_superpotential}  The proposition  \ref{prop_susy_localization} implies that the supersymmetric partition function  is invariant under the  certain  deformations.  For an action (\ref{eq_off_shell_susy_action}), invariant under the supersymmetries  (\ref{eq_off_shell_susy_transform}) we can choose an odd function 
\be\label{eq_odd_function_choice_off_super}
V = -H(x)\psi. 
\ee
to generate the deformation 
\be\label{eq_susy_deform_superpotential}
\begin{split}
S_W+ tQ(V) &=-\frac14  F^2+  FW(x) + \psi  \bar{\psi} W'(x) +t (\bar{\psi} \p_x - F \p_\psi) (-H(x) \psi) = S_{W+tH}.  
 \end{split}
\ee
Hence, the   proposition  \ref{prop_susy_localization} implies that the partition functions for two  supersymmetric systems with superpotentials  $W$ and $W+t H$ are  the same.  Indeed, using the change of variables $y = W(x)$, we can verify  that the partition function is independent on the shape of $W(x)$, but sensitive to the  behavior at $\pm \infty$.
\end{Example}

\begin{Example} We can attempt to use a localization for the on-shell system from section \ref{sub_sect_offshel_superpotential} by choosing  $V$ to obey $Q^2 (V )= 0$. 
By construction localization function $V$ is odd hence the most general expression is $V = f(x) \psi + g(x) \bar{\psi}$.
Using supersymmetry vector fields (\ref{eq_induced_susy}) and the on-shell supersymmetry algebra (\ref{eq_on_shell_susy_algebra_supermult}) we evaluate 
\be
\begin{split}
 Q_{ind}^2 V &=  -2W'(x) f(x) \bar{\psi}   = 0\;\; \Longrightarrow f =0\;\; \Longrightarrow Q_{ind} V =- f'(x) \psi \bar{\psi} - 2Wf(x) = 0, \\
  \bar{Q}_{ind}^2 V &=  2W'(x) g(x)\psi  = 0\;\; \Longrightarrow g =0\;\; \Longrightarrow \bar{Q}_{ind} V = g'(x) \psi \bar{\psi} + 2Wg(x) = 0.
 \end{split}
\ee
Hence, the only odd functions $V$,  which additionally satisfy $Q^2(V)= 0$, describe the trivial deformations.
\end{Example}

\subsection{Two deformations of the on-shell theory with superpotential} \label{sec_two_deformations}

In this section we will start with the simple on-shell supersymmetric system and describe the  two families of  deformations, which preserve the same amount of supersymmetry. 

 The starting system is the on-shell superpotential theory (\ref{eff_action_W}) with the  linear superpotential $W(x)=x$. The corresponding action is  quadratic in even and odd fields 
\be\label{eq_gauss_susy_action}
S_0 (x,\psi,\bar{\psi})= x^2 +  \psi\bar{\psi}.
\ee
The action  (\ref{eq_gauss_susy_action}) is   invariant under the  supersymmetry transformations 
generated by the parity odd vector fields 
\be\label{eq_gauss_susy_transform}
 Q  = \bar{\psi} \frac{\p}{\p x} - 2x \frac{\p }{\p \psi},\;\; \bar{Q} = \psi \frac{\p}{\p x} + 2x \frac{\p }{\p \bar{\psi}}.
\ee
The  algebra  of  symmetries is the $d=0\;\;\; \mathcal{N}=2$ on-shell  supersymmetry algebra 
\be
\begin{split}
[Q, Q\} &=-4 \bar{\psi} \frac{\p}{\p \psi},\;\;\; [\bar{Q}, \bar{Q}\} =4 \psi \frac{\p}{\p \bar{\psi}},\\
[ \bar{Q},Q\} &= -2  \left(\psi \frac{\p}{\p \psi} - \bar{\psi} \frac{\p}{\p \bar{\psi}}\right).
\end{split}
\ee
The partition function for the quadratic theory becomes the  Gaussian integral  and is easy to evaluate
\be\label{eq_gauss_part_function}
Z_0 = \frac{1}{\sqrt{\pi}}\int_{\mathbb{R}^{1|2}} dxd\psi d\bar{\psi}\; e^{-S_0} = \frac{1}{\sqrt{\pi}}\int_{\mathbb{R}^{1|2}} dxd\psi d\bar{\psi}\; e^{-x^2} (-\psi\bar{\psi})  = 1.
\ee

Let us consider the two  families of deformations,  parametrized by an even parameter $g$,  of the theory (\ref{eq_gauss_susy_action}), preserving the same amount of  on-shell supersymmetry
\begin{enumerate}

\item  {\bf Superpotential deformation}:  We can deform the quadratic action (\ref{eq_gauss_susy_action}) by deforming   the superpotential with the higher order term while preserving behavior at infinity, i.e. 
\be
W = x + \frac12 g x^3.
\ee  
The action changes into
\be
S_W = W^2(x) + W'(x) \psi\bar{\psi} = x^2 + g x^4 + \frac14 g^2 x^6  + \psi\bar{\psi} \left(1+\frac32 g x^2\right).
\ee
The action $S_W$ is invariant under  the supersymmetry  transformations
\be
 Q_W  = \bar{\psi} \frac{\p}{\p x} - (2x + g x^3) \frac{\p }{\p \psi},\;\; \bar{Q}_W = \psi \frac{\p}{\p x} + (2x+g x^3) \frac{\p }{\p \bar{\psi}}.
\ee
The  algebra  of  symmetries
\be
\begin{split}
[Q_W, Q_W\} &=-2(2+3g x^2) \bar{\psi} \frac{\p}{\p \psi},\;\;\;[\bar{Q}_W, \bar{Q}_W\} =2(2+3g x^2) \psi \frac{\p}{\p \bar{\psi}},\\
[ \bar{Q}_W,Q_W\} &= -(2+3g x^2)  \left(\psi \frac{\p}{\p \psi} - \bar{\psi} \frac{\p}{\p \bar{\psi}}\right).
\end{split}
\ee
is the  $d=0\;\; \mathcal{N}=2$ on-shell  supersymmetry algebra. 
\item {\bf Quadratic deformation}:  The $d=0$ analog of the $T\bar{T}$-deformations  is the modification of the  quadratic action  (\ref{eq_gauss_susy_action})  by  adding the square of the original action i.e.
\be\label{eq_action_quadr_deform}
S_g (x,\psi,\bar{\psi})=S_0 + gS_0^2 =  x^2 +  \psi\bar{\psi} + g ( x^2 +  \psi\bar{\psi})^2 = x^2 +g x^4 + \psi\bar{\psi} (1+2g x^2).
\ee
The action (\ref{eq_action_quadr_deform}) is  invariant under the supersymmetry transformations (\ref{eq_gauss_susy_transform}).
\end{enumerate}

\subsection{Two partition functions}

For the two deformations $S_W$ and $S_g$ we define  the corresponding partition functions
\be
Z_W = \frac{1}{\sqrt{\pi}}\int_{\mathbb{R}^{1|2}} dxd\psi d\bar{\psi}\; e^{-S_W},\;\;\; Z_g = \frac{1}{\sqrt{\pi}}\int_{\mathbb{R}^{1|2}}  dxd\psi d\bar{\psi}\; e^{-S_g}.
\ee
Since both theories are deformation of the quadratic theory we can view both partition functions  as deformations of the Gaussian partition function (\ref{eq_gauss_part_function}). Let us  evaluate them   perturbatively in $g$ as an expansion around the Gaussian case, i.e.
\be\label{eq_part_funct_quadr_deform_exact}
Z_g(g) = \frac{1}{\sqrt{\pi}}\int dxd\psi d\bar{\psi}\; e^{-x^2 -g x^4 - \psi\bar{\psi} (1+2g x^2)}  =Z_0 \cdot\left\<  e^{-g x^4 -2g x^2\psi\bar{\psi}}\right\>
\ee
and 
\be
Z_W(g)=\frac{1}{\sqrt{\pi}}\int dxd\psi d\bar{\psi}\; e^{-x^2 -g x^4 - \frac14 g^2 x^6  - \psi\bar{\psi} \left(1+\frac32 g x^2\right)}  =Z_0\cdot  \left\<  e^{-g x^4  -\frac32 g x^2\psi\bar{\psi}  - \frac14 g^2 x^6 } \right\>.
\ee
We  introduced  the super-Gaussian average 
\be
\<F(x,\psi, \bar{\psi})\>  =  \frac{1}{\sqrt{\pi}}\int dxd\psi d\bar{\psi} \;\;  F(x,\psi, \bar{\psi})\;e^{-x^2  - \psi\bar{\psi}}. 
\ee
The super-Gaussian averages for even monomials are
 \be\label{eq_super_gauss_integrals}
  \begin{split}
 \< x^{2k} (\psi\bar{\psi})^m \>  &=\left(-\frac{\p}{\p b}\right)^m  \left(-\frac{\p}{\p a}\right)^k \frac{1}{\sqrt{\pi}} \int_{\mathbb{R}^{1|2}}  dx d\psi d\bar{\psi}\;\; e^{ - ax^2 -b\psi\bar{\psi} } \Big|_{a=b=1} \\
 & = \left(-\frac{\p}{\p b}\right)^m  \left(-\frac{\p}{\p a}\right)^k \frac{b}{\sqrt{a}} \Big|_{a=b=1} .
  \end{split}
 \ee
Note that the source of   cancelations in supersymmetric partition function is an   extra minus sign for the fermionic ($\psi\bar{\psi}$) loops, i.e.
 \be
 \<x^{2n} \psi \bar{\psi}\> = - \<x^{2n}\>.
 \ee
The corrections at $\cO(g)$ are 
\be
\begin{split}
-\p_g\Big|_{g=0} Z_g &=  \<  x^4\> +2 \<\psi\bar{\psi}  x^2\>  =  \frac{3}{4} - 2 \cdot \frac12 =- \frac14,\\
-\p_g\Big|_{g=0} Z_W &=  \<  x^4\> +\frac32  \<\psi\bar{\psi}  x^2\>=  \frac{3}{4} - \frac32  \cdot \frac12 =0,
\end{split}
\ee
hence  the leading order corrections to the two partition functions
\be\label{eq_quadr_deform_part_fun_pert}
Z_g(g)  = Z_0 \left(1 +\frac14 g+ \cO(g^2)\right),\;\;\; Z_W(g)  = Z_0 \left(1+ \cO(g^2)\right).
\ee 
Note that the deformation of the superpotential does not have linear in $g$ corrections, while the quadratic deformation does.  Both integrals can be evaluated in terms of standard special functions and the results follow the same pattern: 
the superpotential deformation does not have corrections in all orders of $g$, while  the quadratic deformation has corrections at all orders in $g$.

The partition function   for the superpotential deformation  is  simplified if we integrate over the fermions first and then perform a variable change, i.e.
\be
\begin{split}
Z_W &= \frac{1}{\sqrt{\pi}}  \int_{\mathbb{R}} dx \;  \left(1+\frac32 g x^2\right) e^{ -x^2 -g x^4- \frac14 g^2 x^6 } 
=  \frac{1}{\sqrt{\pi}} \int^{+\infty}_{-\infty}dy\; e^{-y^2} =1.
\end{split}
\ee
The partition function  for the quadratic  deformation  is simplified by integrating out the fermions
\be
Z_g = \frac{1}{\sqrt{\pi}}  \int_{\mathbb{R}} dx \;  (1+2gx^2) e^{ -x^2 -g x^4 } = (1-2g \p_a) I(a,b)\Big|_{a=1, b=g}.
\ee
We can use  an integral formula for the quartic integral
\be
I(a,b) =\frac{1}{\sqrt{\pi}}  \int_{\mathbb{R}} dx\; e^{-a x^2 - bx^4} = \frac12 \sqrt{\frac{a}{\pi b}} e^{\frac{a^2}{8b}} K_{\frac14} \left(\frac{a^2}{8b}\right),\;\; a>0,\;\; b>0.
\ee
Large argument expansion for the modified Bessel function 
\be\nn
K_\a(z) = \sqrt{\frac{\pi}{2z}} e^{-z} \left(1 + \frac{4\a^2-1}{8z} + \frac{(4\a^2-1)(4\a^2-9)}{2! (8z)^2} +\ldots\right),\;\; |z|\gg1 
\ee
reproduces the perturbative corrections (\ref{eq_quadr_deform_part_fun_pert}) and proves that there are non-trivial   higher order order corrections in $g$-expansion.

We presented the two  deformations of the quadratic  on-shell supersymmetric action which preserve the on-shell supersymmetry.  Although both deformations are in the same class the corresponding  
partition functions are very different: one is independent on the deformation parameter, while the other has non-trivial dependence, even on the perturbative level.  These two examples will serve a prime motivating towards our introduction of the refined notion of the on-shell supersymmetry in section \ref{sec_refined_bv_susy}. 

\section{BV formalism}\label{sec_bv_formalism}

In this section we briefly review the BV formalism following lecture notes \cite{mnev2019quantum} and specify the formalism  to the symmetries with graded algebra symmetries.

\subsection{Definitions and theorems}

Let $\mathcal{M}$ be a supermanifold  of dimension $(n|m)$ with even coordinates $t^i,\;\; i=1,..,n$ and odd coordinates $\theta^\a,\;\; \a = 1,...,m $. 

\begin{Definition}
An {\it odd-symplectic  structure} $(\mathcal{M}, \omega)$  on $\mathcal{M}$ is a 2-form $\omega$ which is 
\begin{itemize}
\item closed, $d\omega=0$;
\item odd, i.e. in local coordinates $\omega =  \omega_{i \a}(t, \theta) dt^i \wedge d\theta^\a$, for a collection of even functions $\omega_{i \a}$ on $\mathcal{M}$;
\item is non-degenerate, i.e. the matrix of coefficients   $\omega_{i \a}(t, \theta) $ is invertible. 
\end{itemize}
\end{Definition}
The non-degeneracy  of $\omega$ implies that the dimension of $\mathcal{M}$ is $(n|n)$.  There is an analog to the Darboux theorem for the odd symplectic form, i.e. there exist a collection of Darboux charts on $\mathcal{M}$,  such that 
$\omega =  dx^i\wedge d\xi_i $.  Moreover, the Darboux coordinates exist globally, i.e. every odd symplectic manifold $\mathcal{M}$  with body $M$ is symplectomorphic to the $T^\ast[1] M$.   The non-degeneracy of $\omega$ allows to define a map 
\be\label{def_function_vect_field_map}
C^\infty(\mathcal{M}) \to \hbox{Vect} (\mathcal{M}):\; f \mapsto X_f,\;\;\;\; \iota_{X_f}\omega =  df.
\ee
In Darboux coordinates the corresponding vector field takes the form
\be
X_f = (-1)^{|f|} \p_{x^i} f \p_{\xi_i}  +  \p_{\xi_i} f \p_{x^i}.
\ee

\begin{Definition}\label{def_BV_bracket}
The map (\ref{def_function_vect_field_map})  defines the odd analogue of Poisson bracket, known as the BV-bracket 
\be
\{ f, g\} =  (-1)^{|f|} X_f (g). 
\ee 
\end{Definition}
 The BV bracket in Darboux coordinates  evaluates into
\be
\{ f, g\} =  \frac{\p f }{\p x^i}\frac{\p g }{\p \xi_i}+ (-1)^{|f|} \frac{\p f }{\p \xi_i}\frac{\p g }{\p x^i}. 
\ee
\begin{Remark} The commonly  (for example see \cite{mnev2019quantum,Schwarz:1992nx})  used expression for the BV bracket
\be
\{ f, g\} =  f \left(\frac{\overset{\leftarrow}{\p}  }{\p x^i}\frac{\overset{\rightarrow}{\p}  }{\p \xi_i} -  \frac{\overset{\leftarrow}{\p}  }{\p \xi_i}\frac{\overset{\rightarrow}{\p}  }{\p x^i}\right) g
\ee
is given in terms of  the left and right derivatives
\be
\overset{\rightarrow}{\p}_A f = \p_A f,\;\;\;  f   \overset{\leftarrow}{\p}_A  =  (-1)^{|z^A| (|f|+1)} \p_A f,\;\; z^A = (x^i, \xi_i).
\ee
An explicit substitution for the left and right derivatives gives a BV bracket  identical  to our definition \ref{def_BV_bracket}.
\end{Remark}
 \begin{Definition}  For an odd symplectic manifold $(\mathcal{M}, \omega)$ with the Berezinian $\mu$ we introduce the odd second order operator $\Delta_\mu: C^\infty(\mathcal{M}) \to C^\infty(\mathcal{M})$, the {\it Batalin-Vilkovisly (BV) Laplacian}, defined by
 \be
 \Delta_\mu f =\frac12 \hbox{div}_\mu  X_f.
 \ee 
  \end{Definition}

 \begin{Remark}   In symplectic geometry the divergence of the Hamiltonian vector field  with respect to the Liouville measure is always zero. The main reason for that is that the Liouville measure is, up to a numerical coefficient, equal to the 
 $\omega^n$ for a symplectic $2n$-fold $(M, \omega)$. Hence, the divergence of the vector field, expressed as   the Lie derivative of the measure is proportional to the Lie derivative of symplectic form, which is zero for the for the Hamiltonian vector field. In case of the odd symplectic geometry the Berezenian $\mu$ is not  equal to the power of the odd  symplectic form!  Moreover, the   square of the odd symplectic form vanishes. Namely,  in Darboux coordinates  the square evaluates into 
 \be\nn
 \omega^2 = (dx^i \wedge d\xi_i) \wedge  (dx^j \wedge d\xi_j)  =  dx^i \wedge dx^j \wedge  d\xi_i  \wedge d\xi_j =0.
 \ee
 \end{Remark}
 For a Berezinian  $\mu  = \rho (x,\xi)\;  d^n x D^n \xi $ written in Darboux coordinates the BV Laplacian 
 \be\label{eq_BV_lapl_darboux_coordinates_berezinian}
  \Delta_\mu f   = \frac{\p}{\p x^i} \frac{\p}{\p \xi_i} f + \frac12 \{\ln \rho, f\}.
 \ee
The BV Laplacian squares to zero, i.e.  $ \Delta_\mu^2 =0$, if  the Berezinian $\mu = \rho (x,\xi) d^nx D^n\xi$ obeys 
 \be\label{eq_measure_factor_BV_lapl_nilp}
  \frac{\p}{\p x^i} \frac{\p}{\p \xi_i} \sqrt{\rho} = 0.
 \ee
 
  \begin{Definition}
 For $(n|n)$-dimensional odd symplectic manifold $(\mathcal{M}, \omega)$ a Berezinian $\mu$ on $\mathcal{M}$ is called {\it compatible} with $\omega$, if there exists an atlas of Darboux charts, known as the special Darboux charts,  
 $(x^i, \xi_i)$ on $\mathcal{M}$ such that for all charts $\mu = d^nx D^n \xi$. 
 \end{Definition}
 The condition (\ref{eq_measure_factor_BV_lapl_nilp}) is necessary and sufficient for existence of special  Darboux charts.  Hence, the corresponding Berezinians are compatible. In special Darboux charts  the BV Laplacian (\ref{eq_BV_lapl_darboux_coordinates_berezinian}) always squares to zero and  simplifies  to 
 \be
 \Delta_\mu = \frac{\p}{\p x^i} \frac{\p}{\p \xi_i}.
 \ee 

 \begin{Remark}
 In the terminology of Schwarz \cite{Schwarz:1992nx} an odd-symplectic manifold $(\mathcal{M}, \omega, \mu)$ with compatible Berezenian us called an {\it SP-manifold}. The ``P-structure" refers to the odd-symplectic structure, while the 
 ``S-structure" refers to the  compatible Berezinian.
 \end{Remark}
Below we list  several useful identities for BV bracket and BV Laplacian.
\begin{itemize}
\item Graded-symmetry 
\be
\{f,g\} = -  (-1)^{(|f|+1)(|g|+1)} \{ g, f\}.
\ee
In particular,  when $f$ and $g$ are both even function  the BV-bracket is symmetric, i.e.
\be
\{f,g\}=  \{ g, f\}.
\ee
\item Leibniz  identity
\be
\begin{split}
\{ f, gh\} &= \{f, g\}\cdot h + (-1)^{(|f|+1)|g|} g\{ f, h\},\\
\{ fg, h\} &=f \{g, h\} + (-1)^{(|h|+1)|g|} \{ f, h\}\cdot g;
\end{split}
\ee
\item Jacobi identity 
\be\label{eq_bv_brack_jacobi}
\{f,\{g,h\}\}= \{\{f,g\}, h\} + (-1)^{(|f|+1)(|g|+1)} \{ g, \{ f, h\}\};
\ee
\item Quasi-Leibniz identity 
\be
\Delta (fg) = \Delta f \cdot  g +  (-1)^{|f|} f \cdot \Delta g +  (-1)^{|f|} \{ f, g\};
\ee
\item Derivation identity 
\be\label{eq_bv_brack_derivation}
\Delta \{ f, g\} = \{ \Delta f, g\} + (-1)^{|f|+1} \{ f, \Delta g\}.
\ee
\end{itemize}

\begin{Definition} A {\it Lagrangian submanifold} $\mathcal{L}$ of an odd symplectic manifold $(\mathcal{M}, \omega)$ is 
\begin{itemize}
\item isotropic, i.e. $\omega|_{\mathcal{L}} = 0$;
\item not a proper submanifold of another isotropic submanifold of $\mathcal{M}$.
\end{itemize} 
\end{Definition}
The definition implies that a Lagrangian submanifold $\mathcal{L}$ of an odd-symplectic manifold $\mathcal{M}$ of dimension $(n|n)$ has dimension $(k|n-k)$ for some $0\leq k\leq n$.

  For a  $(k|n-k)$-dimensional supermanifold $\mathcal{N}$ there is an  odd symplectic $(n|n)$-dimensional supermanifold $T^\ast[1]\mathcal{N}$. 
 Let $X^a$ be local coordinates on $\mathcal{N}$, while $(X^a, \Xi_a)$ are coordinates on $T^\ast[1] \mathcal{N}$ then the canonical odd symplectic form is 
 \be
 \omega = (-1)^{|X^a|} dX^a \wedge d\Xi_a.
 \ee
An extra sign factor is required by the invariance of the canonical symplectic form with respect to the coordinate transformations on the base supermanifold $\mathcal{N}$. The detailed construction of the odd cotangent bundle for the supermanifold and the  derivation of the canonical symplectic form is presented in \cite{manin2013gauge}.

For an odd  function $\Psi \in C^\infty(\mathcal{N})_{odd}$ on $\mathcal{N}$  we can associate a Lagrangian  submanifold
 \be
 \mathcal{N}_\Psi  = \hbox{graph} (d\Psi) \subset T^\ast[1] \mathcal{N}.
 \ee
The Lagrangian submanifold  $\mathcal{N}_\Psi$  in local coordinates is given by 
 \be
 \Xi_a  = \frac{\p }{\p X^a} \Psi(X).
 \ee
 For a supermanifold $\mathcal{N}$  we have an identity
 \be
 \hbox{Ber} (T^\ast[1] \mathcal{N})|_{\mathcal{N}} = \hbox{Ber} (\mathcal{N})^{\otimes 2},
 \ee
 hence we have a canonical map sending Berezinians $\mu$ on $T^\ast[1] \mathcal{N}$ to Berezinians $\sqrt{\mu|_\mathcal{N}}$ on $\mathcal{N}$. Locally, for  local coordinates $X^a$
 on $\mathcal{N}$, coordinates $(X^a, \Xi_a)$  on $T^\ast[1] \mathcal{N}$ and a Berezinian $\mu = \rho(X, \Xi) d^nX d^n\Xi$ the corresponding Berezinian in $\mathcal{N}$
 is $\sqrt{\mu|_\mathcal{N}} = \sqrt{\rho(X, 0)} d^nX$.
 
 \begin{Definition}\label{def_bv_integral}
 For odd symplectic manifold with compatible Berezenian  $(\mathcal{M}, \omega, \mu)$ a {\it BV-integral} of a function $f \in C^\infty(\mathcal{M})$, such that $\Delta_\mu f = 0$, over Lagrangian submanifold $\mathcal{L}$    is an integral of the form
 \be\label{eq_def_bv_int}
 \int_{\mathcal{L}} f\; \sqrt{\mu|_{\mathcal{L}}}.
 \ee
 \end{Definition}

\begin{Theorem}\label{thm_bv_invariance} (Batalin-Vilkovisky-Schwarz \cite{Schwarz:1992nx})  For an odd-symplectic manifold with compatible Berezinian  $(\mathcal{M}, \omega, \mu)$ and  compact body the following holds:
\begin{enumerate}
\item For any $g \in C^\infty(\mathcal{M})$  and Lagrangian submanifold  $\mathcal{L} \subset \mathcal{M}$ the BV integral 
 \be
 \int_{\mathcal{L}} \Delta_\mu g\; \sqrt{\mu|_{\mathcal{L}}} = 0.
 \ee
 \item For pair of Lagrangian submanifolds $\mathcal{L}$ and $\mathcal{L}'$ with homologous bodies in the body of $\mathcal{M}$ and a function, such that $\Delta_\mu f = 0$ the following holds
  \be
 \int_{\mathcal{L}}  f\; \sqrt{\mu|_{\mathcal{L}}} =   \int_{\mathcal{L}'}  f\; \sqrt{\mu|_{\mathcal{L}'}} = 0.
\ee
\end{enumerate}
\end{Theorem}

\begin{Definition} (Khudaverdian \cite{khudaverdian2000semidensities}) For an odd symplectic  $(n|n)$-supermanifold $(\mathcal{M}, \omega)$ the  canonical BV Laplacian   on half-densities  $\Delta:$ Dens$^\frac12 \mathcal{M} \to$ Dens$^\frac12 \mathcal{M}$
in a local Darboux chart with coordinates $x^i, \xi_i$ is given by 
\be
\Delta:\qquad \rho(x,\xi) \;\mathcal{D}^\frac n2 \xi d^\frac n2 x  \mapsto  \left( \sum_i \frac{\p}{\p x^i}\frac{\p}{\p \xi_i} \rho(x,\xi) \right) \;\mathcal{D}^\frac n2 \xi d^\frac n2 x.
\ee
\end{Definition}
A Berezinian $\mu$ on $\mathcal{M}$   defines a map $C^\infty(\mathcal{M}) \to $Dens$^\frac12 \mathcal{M}$, given by the $f \mapsto \sqrt{\mu}\; f$, hence we can define a second order odd operator  $\Delta^H_\mu: C^\infty(\mathcal{M}) \to C^\infty(\mathcal{M})$, such that 
\be
\sqrt{\mu} \;\Delta^H_\mu f =  \Delta (\sqrt{\mu}\; f).
\ee 
The operator $\Delta^H_\mu$ and the BV laplacian $\Delta_\mu$  are related by a constant shift
\be
\sqrt{\mu}\; \Delta^H_\mu f =\sqrt{\mu}\;  \Delta_\mu f + \Delta (\sqrt{\mu}).
\ee
For a pair of odd symplectic manifolds $(\mathcal{M}', \omega)$  and $(\mathcal{M}'', \omega'')$  the Cartesian  product  $\mathcal{M} = \mathcal{M}' \times  \mathcal{M}''$ is an odd symplectic manifold with the symplectic form 
$\omega  = \omega' + \omega''$.  Let $P: \mathcal{M} \to \mathcal{M}'$ be a projection on the first factor. 
\begin{Definition}
For a Lagrangian submanifold $\mathcal{L}'' \subset \mathcal{M}''$ we define {\it the fiber BV integral}
\be\label{def_fiber_bv_int}
P_\ast^{(\mathcal{L}'')} = \int_{\mathcal{L}''}\;\;: \hbox{Dens}^\frac12 \mathcal{M} \to \hbox{Dens}^\frac12 \mathcal{M}'.
\ee
\end{Definition}

\begin{Theorem}\label{thm_bv_stokes} (Stokes theorem for fiber BV integrals)   For an odd-symplectic manifold with compatible Berezinian  $(\mathcal{M}, \omega, \mu)$ and  compact body the following holds:
\begin{itemize} \item The  fiber BV integral (\ref{def_fiber_bv_int}) and the  BV Laplacian on half-densities  obey
\be
P_\ast^{(\mathcal{L}'')} \Delta  = \Delta' P_\ast^{(\mathcal{L}'')}  .
\ee
\item For a pair of  Lagrangian submanifolds  $\mathcal{L}''_0$ and $\mathcal{L}''_1$, homotopic in $\mathcal{M}''$ a half-density $\phi\in$ Dens$^{\frac12}\mathcal{M}$, such that $\Delta \phi =0$
\be
P_\ast^{(\mathcal{L}_0'')} \phi -P_\ast^{(\mathcal{L}_1'')}\phi = \Delta'(\ldots).
\ee
\end{itemize}
\end{Theorem}

\subsection{Master equation}

In this section we will  adopt the BV formalism, invariance   theorem  \ref{thm_bv_invariance} and BV-Stokes theorem \ref{thm_bv_stokes} for mathphysical applications. 

For our applications of BV formalism the odd symplectic manifold $\mathcal{M} = T^\ast[1]X$ is an odd cotangent bundle  for a supermanifold $X$ of dimension $(n|m)$. It is conventional to call even coordinates on $X$ as $x^i$ while the odd ones as $\psi^\a$.
The coordinates on $\mathcal{M}$ include the corresponding antifields: odd  $ x^\ast_i$ and even $\psi^\ast_\a$. 

The canonical  symplectic form  on $\mathcal{M} = T^\ast[1]X$ is 
\be\label{ed_can_sympl_cotang_superm}
\omega =  (-1)^{|\phi^a|} d\phi^a \wedge d\phi^\ast_a  = dx^i \wedge dx_i^\ast - d\psi^\a \wedge  d\psi^\ast_\a.
\ee
Note that for a supermanifold $X$ the symplectic form (\ref{ed_can_sympl_cotang_superm}) is not written in  Darboux coordinates, so we need a coordinate transform 
\be
x^I  = (x^i, \psi^\ast_\a ),\;\;\; \xi_I = ( -\psi_\a, x^\ast_i)
\ee
which flips the sign for odd coordinates $\psi_\a$.  The BV bracket for the symplectic form (\ref{ed_can_sympl_cotang_superm}) in field-antifield coordinates
\be
\{ f, g\} =  (-1)^{|\phi^a|}( \p_{\phi^a} f\; \p_{\phi^\ast_a} g+ \p_{\phi^\ast_a} f\; \p_{\phi^a} g ).
\ee
The BV Laplacian  for the standard Berezinian on $T^\ast[1]X$ becomes 
 \be
 \Delta_\mu = \frac{\p}{\p x^I} \frac{\p}{\p \xi_I} = (-1)^{|\phi_a|}  \frac{\p}{\p \phi^a} \frac{\p}{\p \phi^\ast_a} =  \frac{\p}{\p x^i} \frac{\p}{\p x^\ast_i} -  \frac{\p}{\p \psi^\a} \frac{\p}{\p \psi^\ast_\a}.
 \ee
We introduce the exponential parametrization  with  a parameter 
\be\label{form_exp_param_bv_action}
f = e^{-\frac{1}{\hbar} \mathcal{S}},
\ee
so that the condition $\Delta_\mu  f = 0$ becomes  
\be
\Delta_\mu e^{-\frac{1}{\hbar}\mathcal{S}}  =\frac{1}{\hbar^2}  e^{-\frac{1}{\hbar}\mathcal{S}} \left( \frac12 \{\mathcal{S}, \mathcal{S}\}- \hbar \Delta_\mu \mathcal{S} \right) = 0.
\ee
The leading order in  $\hbar$-expansion is the classical master equation.
\begin{Definition} For  an odd symplectic space  $(\mathcal{M}, \omega)$  an even    function $\mathcal{S} \in C^\infty(\mathcal{M})$, known as the classical  BV action,  is a solution to the classical master equation if 
\be\label{cl_mast_eq}
\{ \mathcal{S}, \mathcal{S}\}  = 0.
\ee
\end{Definition}
\begin{Definition} For  an odd symplectic space with compatible Berezenian  $(\mathcal{M}, \omega, \mu)$ a  formal series of even functions  $\mathcal{S}\in C_{even}^\infty(\mathcal{M})[[\hbar]]$, known as the   quantum (or effective) BV action, is a solution to the     quantum master equation if
\be
\label{quantum_mast_eq}
\frac12  \{\mathcal{S}, \mathcal{S}\} =   \hbar \Delta_\mu \mathcal{S}. 
\ee
\end{Definition}
In terms of the expansion coefficients 
\be\label{BV_action_hbar_expansion}
\mathcal{S} =\mathcal{S}_0+\hbar \mathcal{S}_1+\hbar^2 \mathcal{S}_2+...
\ee
the quantum master equation (\ref{quantum_mast_eq}) becomes a system of equations 
\be
\begin{split}
 \{\mathcal{S}_0, \mathcal{S}_0\}& = 0, \\
  \{\mathcal{S}_0, \mathcal{S}_1\}& = \Delta_\mu \mathcal{S}_0, \\
   \{\mathcal{S}_0, \mathcal{S}_2\} &+ \frac12   \{\mathcal{S}_1, \mathcal{S}_1\}=\Delta_\mu \mathcal{S}_1,\\ 
   &...
 \end{split}
\ee
Note that the leading order term $\mathcal{S}_0$  in the $\hbar$-expansion (\ref{BV_action_hbar_expansion})  solves the classical master equation. The construction of the higher order terms $\mathcal{S}_{k>0}$ for a given $\mathcal{S}_0$ is often called the BV-quantization.

\begin{Remark}
There is a special case when the solution $\mathcal{S}$ to the classical master equation (\ref{cl_mast_eq}) is also a  solution to the quantum master equation (\ref{quantum_mast_eq}), i.e. the BV action $\mathcal{S}$ additionally 
satisfies $\Delta_\mu \mathcal{S}=0$.  All our explicit examples for the quantum master equation solutions are of this form.
\end{Remark}

If the formal series  expansion (\ref{BV_action_hbar_expansion}) has non-zero convergent radius we can use the exponential  function (\ref{form_exp_param_bv_action}) and study the corresponding BV  integrals.
In particular, the fiber   BV integral  (\ref{def_fiber_bv_int})  defines an  induced BV action  via
\be\label{eq_BV_induction}
\mathcal{C}(\hbar)\cdot e^{-\frac{1}{\hbar} \mathcal{S}^{ind} (\hbar)}  \sqrt{\mu'}= \int_{\mathcal{L}''} \sqrt{\mu}\; e^{-\frac{1}{\hbar} \mathcal{S}} .
\ee
The  $\mathcal{C}(\hbar)$ is a function of $\hbar$ only. The BV-Stokes theorem \ref{thm_bv_stokes} implies that if $\mathcal{S}$ solves the quantum master equation so is the induced BV action $\mathcal{S}_{ind}$.

\subsection{BV formulation for  symmetries}

A physical system with a symmetry  gives a natural  solution to the master equation. Below we formalize the notion of the physical system with a symmetry and give a detailed form of the master equation solution.

\begin{Definition}\label{def_phys_syst_off_symm} A classical physical system with  a symmetry  is a collection of data $(X, S, \mathfrak{g}, v_\a)$,  such that 
\begin{itemize}
\item  $X$  is a supermanifold with coordinates (also known as fields) $\phi^a$ ;
\item  $S$ is the (classical) action of the system,  an even function on $X$;
\item $\mathfrak{g}$ is a super Lie algebra with structure constants $ f_{\a\b}^\g$;
\item $v_\a$ is a collection of vector fields on $X$ representing the  $\mathfrak{g}$-action   on $X$, i.e. $v : \mathfrak{g} \to \hbox{Vect}(X)$.   The super commutator of vector fields obeys 
\be\label{def_Lie_alg_action}
[v_\a, v_\beta\} = f_{\a\b}^\g v_\g,\;\; \a = 1,..,\dim \mathfrak{g};
\ee

\item  the action $S$ is invariant under the symmetries, i.e. the following relations hold
\be\label{def_action_invar}
v_\a (S) = v_\a^a \frac{\p}{\p \phi^a} S(\phi) = 0.
\ee
\end{itemize}
\end{Definition}

Each classical physical system with a  symmetry provides  a solution to the classical master equation (\ref{cl_mast_eq}) formalized in proposition below.

\begin{Proposition}\label{prop_cls_master_solution_off_shell} A classical physical system with a symmetry $(X, S, \mathfrak{g}, v_\a)$  defines a function  
\be\label{bv_action_offshel_symm}
\mathcal{S} (\phi, \phi^\ast, c, c^\ast)  = S(\phi) + c^\a  v^a_\a   \phi^\ast_a    - \frac12   (-1)^{|v_\a|(|v_\b|+1)}  c^\a c^\beta  f_{\a\b}^\g c_\g^\ast 
\ee
 on the  odd symplectic space 
\be\label{bv_structure_offshel_symm}
\mathcal{M} = T^\ast[1](X\times \mathfrak{g}[1]),\;\;\; \omega = (-1)^{|\phi^a|} d\phi^a \wedge d\phi_a^\ast+ (-1)^{|c^\a|} dc^\a \wedge dc^\ast_\a,
\ee
which solves the   classical master equation (\ref{cl_mast_eq}).  
\end{Proposition}

\begin{proof}  For a super Lie algebra $\mathfrak{g}$ we  have even and  odd  vector fields $v_\a$, so we  choose the parity of $c^\a$ such that the second term in (\ref{bv_action_offshel_symm}) is even. The Lie algebra bracket ensures that the 
last term is even as well, hence   the whole function  (\ref{bv_action_offshel_symm}) is even. We  check the classical master equation (\ref{cl_mast_eq}) order by order in antifields.
\begin{itemize}
\item Trivial equation:  The action $S = S(\phi)$ does not have antifield dependence, hence
\be\label{eq_cl_bv_trivial}
\{ S(\phi), S(\phi)\} = 0. 
\ee
\item   $\cO(\phi^\ast)^0$: The classical master equation at this order  is equivalent to the invariance of the action, i.e.
\be
\{ S(\phi),   c^\a  v^a_\a   \phi^\ast_a  \} =  -c^\a v_\a (S). 
\ee
\item  $\cO(\phi^\ast)^1$:  The classical master equation at this order  is equivalent to the algebra of symmetries, i.e.
\be
\begin{split}
\{ c^\a  v^a_\a \phi^\ast_a &,   c^\b  v^b_\b \phi^\ast_b\} + 2 \left\{ c^\a v^a_\a \phi^\ast_a , -\frac12  (-1)^{|v_\a|(|v_\b|+1)} c^\a c^\b  f_{\a\b}^\g c_\g^\ast   \right\}  \\
&=- (-1)^{|v_\a|(|v_\b|+1)} c^\a c^\beta( [ v_\a, v_\b\}^b - f_{\a\b}^\g v^b_\g) \phi_b^\ast.
\end{split}
\ee
 \item  $\cO(c^\ast)^1$: The classical master equation at this order  is  
 \be
 \begin{split}
  \{ & c^\a c^\b  c_\g^\ast  (-1)^{|v_\a|(|v_\b|+1)}   f_{\a\b}^\g , c^\d c^\e   (-1)^{|v_\d|(|v_\e|+1)}  f_{\d\e}^\sigma c_\sigma^\ast   \} \\
  &= -2  (-1)^{|v_\a|(|v_\b|+1)} c^\a c^\b c^\g f_{\a\b}^\d (-1)^{(|v_\g|+1)|v_\d|} f_{\d\g}^\sigma c^\ast_\sigma = 0.
  \end{split}
  \ee
The last term vanishes due to the  (super) Jacobi identity for the   the structure constants of the super Lie algebra. An explicit form of the Jacobi identity  can be derived using the graded commutator 
of the three identical odd vector fields $c^\a v_\a$, i.e.
\be
0=[  [  c^\a v_\a, c^\b v_\b\}, c^\g v_\g\} =  (-1)^{|v_\a|(|v_\b|+1)} c^\a c^\b c^\g f_{\a\b}^\d (-1)^{(|v^\g|+1)|v_\d|} f_{\d\g}^\sigma  v_\sigma = 0.
\ee
\end{itemize}
We  showed that the classical master equation holds  at all non-trivial orders in antifields hence the proof is complete.
\end{proof}

For the future discussion it is convenient to introduce a separate notations for the second and third term in the BV action (\ref{bv_action_offshel_symm})
\be\label{def_symmetryt_charge_structure_const}
\mathcal{Q} =   c^\a  v^a_\a \phi^\ast_a ,\;\;\;\; \mathcal{F} = -  \frac12  (-1)^{|v_\a|(|v_\b|+1)}  c^\a c^\beta  f_{\a\b}^\g c_\g^\ast.
\ee
The term $\mathcal{Q}$ is responsible for the BV description of the symmetries, while the term $\mathcal{F}$ is responsible for the structure constants of the super-Lie algebra.

\begin{Remark} The order of multiplication in monomial terms was chosen to minimize the the number of sign factors. The alternative representations of the same expression  are
\be
\mathcal{Q} =   c^\a  v^a_\a \phi^\ast_a  = (-1)^{|\phi^\ast_a|}  \phi^\ast_a  c^\a  v^a_\a = -(-1)^{|\phi^a|}  \phi^\ast_a  c^\a  v^a_\a.
\ee
\end{Remark}
\begin{Remark} An extra sign in front of the structure constant term might seem unusual to the readers familiar with the BV formalism for even symmetries. However, the extra  signs  for the case of all even (and all odd as well) symmetries  simply disappear!
\end{Remark}
\begin{Remark}
The presence of the sign factor in our expression (\ref{def_symmetryt_charge_structure_const}) for  structure constant term can be also justified using  the symmetry considerations. In particular, we can relabel the summation indices $(\a \leftrightarrow \b)$, rearrange the $c^\a$ and interchange the indices in the  structure constants to return the expression to the original form.  
\be
\begin{split}
\mathcal{F} &= -  \frac12  (-1)^{|v_\a|(|v_\b|+1)}  c^\a c^\beta  f_{\a\b}^\g c_\g^\ast  =-  \frac12 (-1)^{|v_\b|(|v_\a|+1)}    c^\b c^\a  f_{\b\a}^\g  c_\g^\ast  \\
  &=-  \frac12(-1)^{|v_\b|(|v_\a|+1)} (-1)^{1+|v_\a||v_\b|}   (-1)^{( |v_\a|+1)( |v_\b|+1)}   c^\a c^\b  f_{\a\b}^\g  c_\g^\ast \\
  &=-  \frac12  (-1)^{|v_\a|(|v_\b|+1)}      c^\a c^\b  f_{\a\b}^\g c_\g^\ast =\mathcal{F}.
  \end{split}
\ee
\end{Remark}

A quantum system with a symmetry requires that the symmetry preserves the partition function (\ref{def_susy_part_fun}).
We can modify the  definition \ref{def_phys_syst_off_symm}   to construct a solution to the  quantum master equation (\ref{cl_mast_eq}).

\begin{Definition} A quantum  system with  a  symmetry $(X, S, \mathfrak{g}, v_\a, \mu_X)$ is  a classical physical system with  a symmetry $(X, S, \mathfrak{g}, v_\a)$  supplemented by a  measure $\mu_X$ on $X$, such that 
\be\label{def_qm_offshel_symm_measure}
\hbox{div}_{\mu_X} v_\a=  - (-1)^{|v^\b|}  f^\beta_{\a \beta }.
\ee
\end{Definition}

\begin{Remark} In the physics literature  a quantum  system with  a  symmetry is usually defined with the  trivial righthandside in  (\ref{def_qm_offshel_symm_measure}), i.e. with the $\mathfrak{g}$-invariant measure  $\mu_X$.
\end{Remark}

\begin{Remark} For purely even algebra ($|v^\a|=0$) and    equality $f^\beta_{\a\beta} = 0$ follows from the {\it unimodularity} of $\mathfrak{g}$ (matrices of adjoint representation are trace-free). The unimodularity follows from the existence of the Haar measure on $G = \exp (\mathfrak{g})$. For a super-Lie algebra the right hand side in (\ref{def_qm_offshel_symm_measure}) is the supertrace   of the adjoint action. 
In particular the supertrace is trivial for the   supersymmetry algebra in $d=0$.  
 \end{Remark}

\begin{Proposition}\label{prop_quant_master_solution_off_shell}  For  a quantum physical system with  a symmetry $(X, S, \mathfrak{g}, v_\a, \mu_X )$  the function (\ref{bv_action_offshel_symm})
is a $\hbar$-independent solution to the  quantum   master equation (\ref{quantum_mast_eq})  on the  odd symplectic space (\ref{bv_structure_offshel_symm})
with Berezinian 
\be
\mu  = \rho(\phi)^2  d^n \phi \; d^n \phi^\ast \; d^m c \; d^m c^\ast
\ee
written in coordinates where the  measure on X is $\mu_X= \rho(\phi) d^n \phi$.
\end{Proposition}
\begin{proof} 
The function  (\ref{bv_action_offshel_symm}) solves the classical master equation and is independent on $\hbar$, hence we need to check that $\Delta_\mu \mathcal{S} = 0$ for all orders in antifields.
\begin{itemize}
\item Trivial equation: $S = S(\phi)$ is independent on antifields, hence  
\be
\Delta_\mu (S)= 0 .
\ee
\item The BV Laplacian for the  remaining terms in   (\ref{bv_action_offshel_symm})  is expressed via   divergences of the symmetry vector fields  and the super trace of the structure constants, i.e.
\be
\begin{split}
\Delta_\mu & \left(c^\a  v^a_\a   \phi^\ast_a  - \frac12 (-1)^{|v_\a|(|v_\b|+1)} c^\a c^\b   f_{\a\b}^\g  c_\g^\ast   \right) \\
&\qquad\qquad=-c^\a  \left(\hbox{div}_{\mu_X} v_\a + (-1)^{|v^\a|} c^\b   f_{\b\a}^\a\right)=0.
\end{split}
\ee
\end{itemize}
\end{proof}

\begin{Remark}
For the classical master equation solution $\mathcal{S}_0$ the non-trivial $\Delta_\mu \mathcal{S}_0$  implies that the symmetry in quantum theory requires a modification of the integration measure $\mu_X$.   In particular if $\mathcal{S}_0 +\hbar \mathcal{S}_1$  is a  solution to the quantum master equation, i.e. 
\be\label{eq_1_loop_master_correction}
\{ \mathcal{S}_0, \mathcal{S}_1\} = \Delta_\mu \mathcal{S}_0,\;\;\; \Delta_\mu \mathcal{S}_1=0,
\ee
 the $\mathcal{S}_1$ modifies the integration measure by the multiplicative factor
\be
\mu_X  \to \mu_X \cdot e^{- \mathcal{S}_1}.
\ee
In general there is an obstruction to the existence of a solution to the (\ref{eq_1_loop_master_correction}), known as the quantum anomaly in physical literature.
\end{Remark}

\subsection{BV formulation for on-shell symmetries}

The classical physical system  on-shell symmetry has modified version of the Lie algebra action (\ref{def_Lie_alg_action}) in definition  \ref{def_phys_syst_off_symm}. 
The BV formulation for the  on-shell  symmetries was studied in case of gauge symmetries  \cite{henneaux1992quantization}, and supersymmetric theories in 4d  \cite{Baulieu:1990uv}. The  key feature of the corresponding BV actions is 
the presence of additional  quadratic in antifields. In the next several subsections we will briefly review the reason why such terms appear and argue that the systems with on-shell symmetries could have arbitrary high order terms in antifields.

Let us give a formal definition for the classical physical system with on-shell symmetry.
\begin{Definition}\label{Def_cl_syst_on_shell_symm} A classical physical system with  on-shell symmetry is a collection of data  $(X, S, \mathfrak{g}, v_\a)$, such that 
\begin{itemize}

\item  $X$  is a supermanifold with coordinates (also known as fields) $\phi^a$ ;
\item  $S$ is the (classical) action of the system,  an even function on $X$;
\item $\mathfrak{g}$ is a super Lie algebra with structure constants $ f_{\a\b}^\g$;
\item $v_\a$ is a collection of vector fields on $X$ representing the  $\mathfrak{g}$-action   on $X$, i.e. $v : \mathfrak{g} \to \hbox{Vect}(X)$.   The algebra  of vector fields obeys 
\be\label{def_symm_action_onshell}
\left([v_\a, v_\beta\}  - f_{\a\b}^\g v_\g \right)\Big|_{X_{crit}} = 0;\;\;\;\; \a = 1,..,\dim \mathfrak{g},
\ee
for  critical set $X_{crit} = \{ x \in X\;|\; dS(x)=0\}$;

\item  the action $S$ is invariant under the symmetries, i.e. the following relations hold
\be
v_\a (S) = v_\a^a \frac{\p}{\p \phi^a} S(\phi) = 0.
\ee
\end{itemize}
\end{Definition}

\begin{Proposition}\label{prop_on_shell_bv} For a classical physical system with  an  on-shell symmetry $(X, S, \mathfrak{g}, v_\a)$  the function (\ref{bv_action_offshel_symm}) is a solution to the classical master equation (\ref{cl_mast_eq}) up to linear  order in antifields on the  odd symplectic space (\ref{bv_structure_offshel_symm}).
\end{Proposition}
\begin{proof} The only difference between the on-shell and off-shell symmetry is the  more relaxing condition (\ref{def_symm_action_onshell}) on the commutator of vector fields in definition comparing to the requirement (\ref{def_Lie_alg_action})  in definition \ref{def_phys_syst_off_symm}. Hence,   a classical physical  system with  on-shell  symmetry gives only an approximate solution   to the classical master equation (\ref{cl_mast_eq}). In particular, the BV action (\ref{bv_action_offshel_symm})  fails to be a solution due to the $\cO(\phi^\ast)$ terms, i.e.
\be
\begin{split}
\{\mathcal{Q}, \mathcal{Q} \} + 2 \left\{ \mathcal{Q} , \mathcal{F} \right\}   =- (-1)^{|v_\a|(|v_\b|+1)} c^\a c^\beta( [ v_\a, v_\b\}^b - f_{\a\b}^\g v^b_\g) \phi_b^\ast \neq 0.
\end{split}
\ee
\end{proof}

We  can modify the  the BV action (\ref{bv_action_offshel_symm}) by adding a quadratic term $\Pi^{(2)}$ in antifields
\be
\mathcal{S} (\phi, \phi^\ast, c, c^\ast)  = S(\phi) +    c^\a   v^a_\a  \phi^\ast_a  - \frac12   (-1)^{|v_\a|(|v_\b|+1)} c^\a c^\beta  f_{\a\b}^\g   c_\g^\ast + \Pi^{(2)}.
\ee
An extra term modifies the  classical master equation at  $\cO(\phi^\ast)$ or higher. The $\cO(\phi^\ast)$ equation becomes 
\be
\{\mathcal{Q}, \mathcal{Q} \}+ 2 \left\{\mathcal{Q} , \mathcal{F}\right\} +2 \{S(\phi),\Pi^{(2)}\}  =0.
\ee
For a given on-shell symmetries we can view the equation above as an equation for the quadratic term  $\Pi^{(2)}$, i.e.
\be\label{ed_pi_2_cl_master}
 \{S(\phi),\Pi^{(2)}\}  =  \mathcal{R} \equiv -\frac12 \{\mathcal{Q}, \mathcal{Q} \}-  \left\{\mathcal{Q} , \mathcal{F}\right\}.
\ee
The BV bracket with the classical action $S(\phi)$ is a first order differential operator, while the Jacobi identity (\ref{eq_bv_brack_jacobi}) and the trivial master equation (\ref{eq_cl_bv_trivial}) implies that it is a differential, i.e.  
\be\label{def_D_S_differential}
D_S =  \{S,\cdot\}  = (-1)^{|\phi^a|} \frac{\p S}{\p \phi^a} \frac{\p}{\p \phi^\ast_a},\;\;\; D_S^2 = 0.
\ee
Hence the equation (\ref{ed_pi_2_cl_master}) becomes a standard problem of the form
\be\label{equ_pi_2}
D_S\Pi^{(2)} =   \mathcal{R}.
\ee

\begin{Lemma}\label{lemma_pi_2_exist} For  classical physical system with on-shell symmetry the  equation (\ref{equ_pi_2}) has  a solution.
\end{Lemma}
\begin{proof}
There are two requirements for the solution existence 
\begin{enumerate}
\item The function $\mathcal{R}$ is $D_S$-closed;
\item The function $\mathcal{R}$ belongs to the trivial class in  $D_S$-cohomology.
\end{enumerate}
The first requirement follows from the Jacobi identity (\ref{eq_bv_brack_jacobi})  for the BV bracket. In particular 
\be
\begin{split}
D_S\mathcal{R} &= \{ S, \mathcal{R}\} =-\frac12 \{ S,  \{\mathcal{Q}, \mathcal{Q} \}\}-  \{ S, \left\{\mathcal{Q} , \mathcal{F}\right\} \} \\
& =\{ \mathcal{Q},  \{S, \mathcal{Q} \}\}  +  \{ \mathcal{Q},  \left\{S, \mathcal{F}\right\} \} +  \{ \mathcal{F},  \left\{S, \mathcal{Q}\right\} \}=0.
\end{split}
\ee 
The last equality is due to the master equation at order $\cO(\phi^\ast)^0$ for the BV action $S + \mathcal{Q}+\mathcal{F}$.

We will provide a detailed discussion of the $D_S$-cohomology in section \ref{sec_d_s_cohomology}, while here we just use the fact that the  $D_S$-cohomology class  remains the same if we restrict the cycle to the critical set $X_{crit}$.
The restriction of  $\mathcal{R}$ to the critical set $X_{crit}$ is trivial, i.e. 
\be
\mathcal{R}\Big|_{X_{crit}} =\frac12  (-1)^{|v_\a|(|v_\b|+1)} c^\a c^\beta( [ v_\a, v_\b\}^b - f_{\a\b}^\g v^b_\g) \phi_b^\ast   \Big|_{X_{crit}} = 0.
\ee
The last equality follows from  the definition \ref{Def_cl_syst_on_shell_symm} of the classical physical system with on-shell symmetry. This completes the proof of the lemma.
\end{proof}
Although  lemma \ref{lemma_pi_2_exist} ensures that the classical physical system with on-shell symmetry defines  a solution to  the classical master equation up to second order in antifields, it does not give us an explicit form for the $\Pi^{(2)}$-term. 
Below we  formulate a particular case when an explicit form of $\Pi^{(2)}$ is constructible.

\begin{Proposition}\label{prop_on_shell_bv_leading} For a classical physical system with  an  on-shell symmetry $(X, S, \mathfrak{g}, v_\a)$   and the refined   on-shell  symmetry algebra in the form 
\be\label{eq_susy_onshell_refined}
[v_\a, v_\beta\}^a  - f_{\a\b}^\g v_\g^a   =(-1)^{|\phi^b|} \Pi_{\a\b}^{ba} \p_b S 
\ee
there is a solution to the classical master equation (\ref{cl_mast_eq}) up to quadratic order in antifields on the  odd symplectic space (\ref{bv_structure_offshel_symm})    in the form 
\be
\mathcal{S}   = S +  c^\a   v^a_\a    \phi^\ast_a - \frac12   (-1)^{|v_\a|(|v_\b|+1)}  c^\a c^\beta  f_{\a\b}^\g  c_\g^\ast+ \frac14  (-1)^{|v_\a|(|v_\b|+1) }  \phi^\ast_a \phi^\ast_b c^\a c^\b \Pi_{\a\b}^{ab}.
\ee 
\end{Proposition}
\begin{proof} We just need to  check the classical master equation (\ref{cl_mast_eq}) at linear order in antifields 
\be
\begin{split}
\{\mathcal{Q}, \mathcal{Q} \} &+ 2 \left\{\mathcal{Q} , \mathcal{F}\right\}  +2 \{S(\phi),\Pi^{(2)}\}  \\
&= -   (-1)^{(|v_\b|+1) |v_\a|} c^\a c^\beta( [ v_\a, v_\b\}^a -f_{\a\b}^\g v_\g^a- (-1)^{|\phi^b|}\Pi_{\a\b}^{ba} \p_b S) \phi_a^\ast= 0.
\end{split}
\ee
The last equality  follows from refined on shell supersymmetry relation (\ref{eq_susy_onshell_refined})  and hence completes the proof of the proposition.
\end{proof}
The readers might  be surprised comparing an existence lemma \ref{lemma_pi_2_exist} and proposition \ref{prop_on_shell_bv_leading} with an explicit quadratic term.  The problem of finding the refinement (\ref{eq_susy_onshell_refined}) for the on-shell  symmetry algebra is as hard as  the construction of homotopy for the differential $D_S$  to express a  solution for $\Pi^{(2)}$.
However, for some  physical applications, supersymmetric systems in particular, there are additional arguments that allow for an  easier construction of the on-shell symmetry refinement.

In example \ref{ex_polyvectors_functions} we saw that the functions on $T^\ast[1]X$ can be identified with the polyvector fields on $X$.
The $T^\ast[1]X$ comes with the canonical odd symplectic form and the    BV-bracket  on functions is mapped to  the Nijenhuis-Schouten bracket between the corresponding polyvectors.   The Nijenhuis-Schouten bracket is the Lie bracket on vector fields extended to the polyvector fields via Leibniz identity.  Hence, following the literature, we will often call  $ \Pi_{\a\b}^{ab}$ and $\Pi^{(2)}$ the BV bivectors, while the higher order expressions in antifields the  BV polyvectors.

\subsection{BV action with polyvectors}\label{sec_bv_action_with_polyevect}

An additional quadratic term in antifields in the BV action leads to the  non-trivial classical master equation at $\cO(\phi^\ast)^2$ and $\cO(\phi^\ast)^3$ orders. Similarly, the quantum master equation will acquire an  additional order in antifields. 
For future applications it is convenient to introduce the following definition.

\begin{Definition} \label{def_quant_suystem_bivector_prop}  A  quantum physical system with  on-shell symmetry of bivector type  $(X, S, \mathfrak{g}, v_\a, \Pi_{\a\b}^{ab}, \mu_X)$  is a classical physical system with  on-shell symmetry $(X, S, \mathfrak{g}, v_\a)$,   an integration measure $\mu_X$ and a refined on-shell condition (\ref{eq_susy_onshell_refined}) for a  bivector  $\Pi_{\a\b}^{ab}$, such that 
\begin{itemize}
\item  The bivector  transforms covariantly under the symmetry action,  what can be expressed using the graded  Nijenhuis-Schouten bracket introduced in \cite{de1997z}, i.e.
\be\label{def_pi_2_algebra_inv}
[ c^\a c^\b \Pi_{\a\b}, c^\g v_{\g} \}^{ab} + (-1)^{|v_\sigma|}c^\a c^\b c^\g  f_{\a\b}^\sigma \Pi^{ab}_{\g\sigma}= 0;
\ee
\item The Nijenhuis-Schouten bracket  of bivector with itself is trivial, i.e. 
\be\label{def_pi_2_bracket_trivial}
 [c^\a c^\b \Pi_{\a\b},  c^\g c^\d\Pi_{\g\d}\} = 0;
\ee
\item The integration measure transforms covariantly, i.e.
\be
\hbox{div}_{\mu_X} v_\a=  -  (-1)^{|v^\b|}   f^\beta_{\beta \a};
\ee
\item The  bivector is divergence-free with respect to the integration measure $\mu_X = \rho(\phi)\; d^n\phi$, i.e. 
\be\label{def_divergence_bivector}
(-1)^{ |v_\a|(|v_\b|+1)}   (-1)^{ |\phi^b||\phi^a|} ( \p_a +\p_a \ln \rho)  c^\a c^\b  \Pi^{ab}_{\a\b} =0.
\ee
\end{itemize} 
\end{Definition}

\begin{Proposition} For  a quantum physical system with on-shell symmetry of a bivector type $(X, S, \mathfrak{g}, v_\a, \Pi_{\a\b}^{ab}, \mu_X)$  the function
\be\label{eq_quant_bv_action_bivector_refined}
\mathcal{S} =S +   c^\a   v^a_\a   \phi^\ast_a - \frac12   (-1)^{|v_\a|(|v_\b|+1)} c^\a c^\beta  f_{\a\b}^\g  c_\g^\ast  + \frac14  (-1)^{|v_\a|(|v_\b|+1) }  \phi^\ast_a \phi^\ast_b c^\a c^\b \Pi_{\a\b}^{ab} 
\ee 
is the $\hbar$-independent  solution to the  quantum   master equation (\ref{quantum_mast_eq})  on the  odd symplectic space (\ref{bv_structure_offshel_symm})
with Berezinian 
\be
\mu  = \rho(\phi)^2  d^n \phi \; d^n \phi^\ast \; d^m c \; d^m c^\ast
\ee
written in coordinates where  measure $\mu_X= \rho(\phi) d^n \phi$.
\end{Proposition}
\begin{proof} The proposition \ref{prop_on_shell_bv_leading} ensures that the function (\ref{eq_quant_bv_action_bivector_refined}) solves the classical master equation up to quadratic order in antifields. Hence, we need to check the classical master equation  at $\cO(\phi^\ast)^2$ and $\cO(\phi^\ast)^3$ orders. The function (\ref{eq_quant_bv_action_bivector_refined}) is independent of $\hbar$, hence the quantum master equation requires to check $\Delta_\mu \mathcal{S} =0$ at linear order in antifields.
\begin{itemize}
\item $\cO(\phi^\ast)^2$: The classical master equation at this order  is equivalent to the  covariant transformation of the   bivector  invariant under the symmetry action, i.e.
\be\nn
\begin{split}
\{ &\Pi^{(2)}, \mathcal{Q}\} + \{ \Pi^{(2)}, \mathcal{F}\}\\
&=-\frac14(-1)^{|\phi^a|+|\phi^b|+ |v_\a|(|v_\b|+1)}   \left( [ c^\a c^\b \Pi_{\a\b}, c^\g v_{\g} \}^{ab} + (-1)^{|v_\sigma|}c^\a c^\b c^\g  f_{\a\b}^\sigma \Pi^{ab}_{\g\sigma}\right)\phi^\ast_a\phi^\ast_b =0;
\end{split}
\ee
\item $\cO(\phi^\ast)^3$:  The classical master equation at this order  is equivalent to the  trivial bracket of bivector with itself, i.e.
\be\nn
\begin{split}
\{ \Pi^{(2)}, \Pi^{(2)}\}  =\frac18  (-1)^{|v_\a|(|v_\b|+1)+ |v_\g|(|v_\d|+1)+|\phi^a|+|\phi^b|+|\phi^c|}  [ c^\a c^\b \Pi_{\a\b}, c^\g c^\d \Pi_{\g\d}\}^{abc} \phi^\ast_a\phi^\ast_b \phi^\ast_c= 0;
\end{split}
\ee
\item  The quantum master   equation at $\cO(\phi^\ast)$ order
\be\nn
\Delta_\mu \Pi^{(2)}   = 2 (-1)^{ |v_\a|(|v_\b|+1)}   (-1)^{ |\phi^b||\phi^a|} ( \p_a +\p_a \ln \rho)  c^\a c^\b  \Pi^{ab}_{\a\b}  \phi^\ast_b   =0.
\ee
\end{itemize} 
\end{proof}

For physical system with  on-shell symmetry $(X, S, \mathfrak{g}, v_\a)$ the corresponding   bivector $ \Pi_{\a\b}^{ab}$ might not satisfy the conditions (\ref{def_pi_2_algebra_inv}) and (\ref{def_pi_2_bracket_trivial}),  so we do not have an on-shell symmetry of bivector type. But  it might be possible to add higher order terms in antifields to the BV action to construct a solution to the classical master equation, i.e.
\be\nn
\mathcal{S} = S +   c^\a   v^a_\a   \phi^\ast_a  - \frac12   (-1)^{|v_\a|(|v_\b|+1)}  c^\a c^\beta  f_{\a\b}^\g c_\g^\ast  + \frac14  (-1)^{ |v_\a|(|v_\b|+1)}  \phi^\ast_a \phi^\ast_b c^\a c^\b \Pi_{\a\b}^{ab}  + \Pi^{(3)} + ...
\ee
The master equation at $\cO(\phi^\ast)^2$ order  becomes 
\be\label{eq_master_quadratic_order}
\{ \Pi^{(2)}, \mathcal{Q}+\mathcal{F}\} + \{ \Pi^{(3)}, S\} = 0.
\ee
There is a possible obstruction towards existence of $\Pi^{(3)}$, which can be formulated  in terms of differential (\ref{def_D_S_differential}).
\begin{Lemma} An expression  $\{ \Pi^{(2)}, \mathcal{Q}+ \mathcal{F}\}$ is $D_S$-closed  for a classical physical system with on-shell symmetry and any consistent bivector. 
\end{Lemma}
\begin{proof} The Jacobi identity for BV-brackets (\ref{eq_bv_brack_jacobi})   allows us to rearrange
\be
\begin{split}
D_S& (\{ \Pi^{(2)}, \mathcal{Q}+\mathcal{F}\} )  = \{ \{S,  \Pi^{(2)}\}, \mathcal{Q}+\mathcal{F}\}  - \{ \Pi^{(2)}, \{S , \mathcal{Q}+\mathcal{F}\}\}.
\end{split}
\ee
The second term vanishes due to the invariance of the classical action under the symmetry transformations.  The first term is simplified using  the defining equation (\ref{ed_pi_2_cl_master})  for the bivector and the Jacobi identity for the superalgebra 
\be
\begin{split}
\{ \mathcal{Q}+\mathcal{F}, \{ S, \Pi^{(2)}\}\}= - \frac12  \left\{ \mathcal{Q}+\mathcal{F},  \{ \mathcal{Q}+\mathcal{F}, \mathcal{Q}+ \mathcal{F}\} \right\} =0.
\end{split}
\ee
The last equality is the Jacobi identity for the BV bracket of the three identical even functions.
\end{proof}

An existence of solution to (\ref{eq_master_quadratic_order}) also requires that the   $\{ \Pi^{(2)}, \mathcal{Q}+\mathcal{F}\} $ is the  trivial class in $D_S$-cohomology. The similar analysis can be extended to higher polyvectors with the corresponding equation
\be\label{eq_polyvector_equation}
D_S \Pi^{(k)} =\mathcal{R}^{(k)}\equiv -\{ \Pi^{(k-1)}, \mathcal{Q}+ \mathcal{F}\} -\frac12  \sum_{l=2}^{k-2} \{ \Pi^{(l)}, \Pi^{(k-l)}\}. 
\ee 
Similarly to the $\Pi^{(3)}$ case we only need to check the cohomological obstruction for the solution existence, since closeness follows from the lemma below.
\begin{Lemma} The $\mathcal{R}^{(k)}$ defined by  the equation (\ref{eq_polyvector_equation}) is $D_S$-closed. 
\end{Lemma}
\begin{proof}   Let us consider the Jacobi identity for the BV bracket of the three identical functions 
\be\label{eq_jacobi_trivial_proof}
\{ \mathcal{S}, \{\mathcal{S},\mathcal{S}\}\}  =0.
\ee
The $\cO(\phi^\ast)^{k-1}$ order of the Jacobi identity (\ref{eq_jacobi_trivial_proof}) has the following schematic form 
\be\label{eq_jacobi_trivial_k_order}
0=\{ \mathcal{S}, \{\mathcal{S},\mathcal{S}\}\} \Big|_{\cO(\phi^\ast)^{k-1}} = \left\{S, \{\mathcal{S},\mathcal{S}\}\Big|_{\cO(\phi^\ast)^k}\right\}  + \sum_{l=1}^{k-1}  \left\{\mathcal{S}\Big|_{\cO(\phi^\ast)^l}, \{\mathcal{S},\mathcal{S}\}\Big|_{\cO(\phi^\ast)^{k-l}}\right\}. 
\ee
By assumption polyvectors $\Pi^{(l)}, l=2,..,k-1$ are such that the master equation for $\mathcal{S}$ is satisfied to $\cO(\phi^\ast)^{k-1}$ order, i.e.
\be
\{\mathcal{S},\mathcal{S}\} = \cO(\phi^\ast)^{k-1}.
\ee
Hence we can drop the last sum in (\ref{eq_jacobi_trivial_k_order}). The $\cO(\phi^\ast)^k$ order for the classical master equation 
\be
 \{\mathcal{S},\mathcal{S}\}\Big|_{\cO(\phi^\ast)^k} =  \{S,\Pi^{(2)}\}  - \mathcal{R}^{(k)}.
\ee
The Jacobi identity allows us to evaliuate $\{ S, \{S,\Pi^{(2)}\} \} = 0$ hence we arrive into 
\be
\{ S, \mathcal{R}^{(k)}\} = D_S \mathcal{R}^{(k)} = 0
\ee
and the proof is complete.
\end{proof}

The quantum master equation for the $\hbar$-independent BV action  further requires 
\be \label{eq_quant_master_k_order}
\Delta_\mu \Pi^{(k)} = 0. 
\ee
In section \ref{sec_cohom_susy_homotopy} we will provide a  $D_S$-homotopy such that  for all $k>2$ the requirement (\ref{eq_quant_master_k_order}) holds.

\begin{Remark} A possible way of constructing solutions to the quantum master equation is to use the BV induction  (\ref{eq_BV_induction}).  Given a quantum system with off-shell symmetry we can integrate out some of the variables 
to get the induced BV action for a physical system with on-shell supersymmetry.  We will provide an explicit example of the BV induction in section \ref{sec_bv_induction_superpotential}.
\end{Remark}

\subsection{Homotopy choice  and  canonical BV transformations}\label{sec_homotopy_bv_transform}

Given a  solution $\Pi$ to the problem 
\be
D_S \Pi = \mathcal{R}
\ee
we can construct another solution $\Pi'$ by adding an exact term and a cohomology, i.e.
\be
\Pi'  = \Pi+ D_S \Lambda + \pi,\;\;\; \pi \in H^\ast (D_S) .
\ee
We can relate the exact terms in BV action  to the canonical BV transformations.

\begin{Definition} \label{def_can_bv_transform} Two solutions $\mathcal{S}, \mathcal{S}'$ for the quantum master equation are equivalent $\mathcal{S}\sim \mathcal{S}'$, if there exists a canonical BV transformation: a family  $\mathcal{S}_t, \mathcal{R}_t \in C^\infty(\mathcal{M})[[\hbar]]$, parametrized by $t\in [0,1]$, such that $\mathcal{S}_0=\mathcal{S}$  and $\mathcal{S}_1=\mathcal{S}'$, and the following equation holds
\be\label{def_BV_canonical_diff}
\frac{d}{dt} \mathcal{S}_t  = \{  \mathcal{S}_t, \mathcal{R}_t\} -\hbar \Delta_\mu \mathcal{R}_t,
\ee
The odd  function  $\mathcal{R}_t$ is called the generator of the canonical BV transformation.
\end{Definition}

\begin{Lemma} The    function  $\mathcal{S}_t \in  C^\infty(\mathcal{M})[[\hbar]]$ solves the quantum master equation.
\end{Lemma}
\begin{proof}  Let  us choose $\e>0$ and  use the canonical transformation (\ref{def_BV_canonical_diff}) to  evaluate 
\be
\mathcal{S}_{t+\e} = \mathcal{S}_{t} + \e  \{  \mathcal{S}_t, \mathcal{R}_t\} -\hbar  \e \Delta_\mu \mathcal{R}_t +\cO(\e^2).
\ee
 The quantum master  equation for $\mathcal{S}_{t+\e} $  can be simplified  using the Jacobi identity  (\ref{eq_bv_brack_jacobi}) and derivation identity  (\ref{eq_bv_brack_derivation}) into 
\be
\begin{split}
\frac12 \{ \mathcal{S}_{t+\e}& , \mathcal{S}_{t+\e} \} -\hbar \Delta_\mu \mathcal{S}_{t+\e} \\
& = \e \{ \mathcal{S}_{t} ,    \{  \mathcal{S}_t, \mathcal{R}_t\}   \} - \e \hbar  \{  \mathcal{S}_t,  \Delta_\mu \mathcal{R}_t \} -\e \hbar \Delta_\mu (\{  \mathcal{S}_t, \mathcal{R}_t\}-\hbar   \Delta_\mu \mathcal{R}_t) +\cO(\e^2)\\
& =  \e \left\{ \frac12  \{  \mathcal{S}_t, \mathcal{S}_t\} - \hbar \Delta_\mu \mathcal{S}_t , \mathcal{R}_t  \right\} +\cO(\e^2) =\cO(\e^2).
\end{split}
\ee
The last equality holds because  $\mathcal{S}_t$ is a solution to the quantum master equation. 
Hence we conclude that  $\mathcal{S}_{t+\e}$ is a solution to the quantum master equation up to  $\cO(\e^2)$ corrections.  This implies (via subdivision of $\e$ into $N$ shifts of length $\e/N$ and taking the limit $N \to \infty$) that 
$\mathcal{S}_{t+\e}$ is also a solution to the quantum master equation. 
\end{proof}

Let us carefully investigate the effects of the $D_S$-exact terms. A  bivector change is
\be\label{eq_change_pi_2}
\delta \Pi^{(2)} = D_S \Lambda^{(3)} = \{S, \Lambda^{(3)}\}
\ee
 for an arbitrary function $\Lambda^{(3)}$ on the odd symplectic manifold, that is degree 3 in antifields. The change (\ref{eq_change_pi_2}) of  bivector  induces a 3-vector change, given by a solution to the defining equation (\ref{eq_master_quadratic_order}), i.e.  
\be
D_S \delta \Pi^{(3)} =- \{\mathcal{Q}+ \mathcal{F}, \delta \Pi^{(2)}\}= - \{\mathcal{Q}+ \mathcal{F}, \{S, \Lambda^{(3)}\}\}.
\ee
We can use the classical master equation to rearrange 
\be\nn
\begin{split}
D_S \delta \Pi^{(3)} &= - \{\mathcal{Q}+ \mathcal{F}, \{S, \Lambda^{(3)}\}\} = - \{ \{\mathcal{Q}+ \mathcal{F}, S\}, \Lambda^{(3)}\}  +\{ S, \{ \mathcal{Q}+ \mathcal{F}, \Lambda^{(3)}\}\} \\
& =   \{S,\{\mathcal{Q}+ \mathcal{F}, \Lambda^{(3)}\}\} =  D_S( \{\mathcal{Q}+ \mathcal{F}, \Lambda^{(3)}\}).
\end{split}
\ee
The solution for 3-vector change 
\be\label{eq_change_pi_3}
\delta \Pi^{(3)} =  \{\mathcal{Q}+ \mathcal{F}, \Lambda^{(3)}\}+ D_S \Lambda^{(4)} = \{S, \Lambda^{(4)}\} + \{\mathcal{Q}+ \mathcal{F}, \Lambda^{(3)}\} .
\ee
 for an arbitrary function $\Lambda^{(4)}$ on the odd symplectic manifold, that is degree 4 in antifields.
The equation for the 4-vector change 
\be
D_S \delta  \Pi^{(4)} =- \{ \delta \Pi^{(3)}, \mathcal{Q}+ \mathcal{F}\}- \{ \delta \Pi^{(2)}, \Pi^{(2)}\} - \frac12 \{ \delta \Pi^{(2)}, \delta \Pi^{(2)}\}
\ee 
gives us a solution 
\be\label{eq_change_pi_4}
 \delta  \Pi^{(4)} =\{S, \Lambda^{(5)}\} +  \{ \mathcal{Q}+ \mathcal{F} ,\Lambda^{(4)}\}  +\{ \Pi^{(2)}, \Lambda^{(3)}\} +\cO(\Lambda^{(3)})^2.
\ee 
 for an arbitrary function $\Lambda^{(5)}$ on the odd symplectic manifold, that is degree 5 in antifields.

We  can arrange the changes (\ref{eq_change_pi_2}, \ref{eq_change_pi_3}, \ref{eq_change_pi_4})  into 
\be
\delta \mathcal{S} =\delta \Pi^{(2)}+\delta \Pi^{(3)}+\delta \Pi^{(4)} =  \{ \mathcal{S}, \Lambda\} + \cO(\phi^\ast)^5+\cO(\Lambda)^2,
\ee
which is identical to the infinitesimal canonical BV transformation (\ref{def_BV_canonical_diff}) for a solution to classical master equation with parameter $\Lambda   =\Lambda^{(3)} +\Lambda^{(4)}+ \Lambda^{(5)} $.
Further analysis shows that the change in  $D_S$-homotopy is equivalent to the canonical BV transformation.

\subsection{Obstructions and Koszul cohomology}\label{sec_d_s_cohomology}

In previous section we observed that the obstruction  to the master equation solution lies in non-trivial cohomology classes  for 
\be\label{eq_Koszul_diff}
D_S =  \{S,\cdot\}  = (-1)^{|\phi^a|} \frac{\p S}{\p \phi^a} \frac{\p}{\p \phi^\ast_a},\;\;\; D_S^2 = 0.
\ee
The key statement about the cohomology  is formulated in proposition below.

\begin{Proposition} Physical system with a symmetry has non-trivial cohomology classes.
\end{Proposition}
\begin{proof} We will provide the careful proof for the case of polynomial actions with all even variables, while provide an arguments towards generalization of the statement for the case of smooth actions on supermanifolds.
\end{proof}

For  all even fields $\phi^a= x^a$, the corresponding antifields $x^\ast_a$ are odd and can be used as a grading. The corresponding complex is the Koszul complex 
\be
D_F = F_i(x) \frac{\p}{ \p x^\ast_i}
\ee
for the sequence $F_i (x) = \p_i S$.

 For polynomial  $S(x)$ can use  the classical results \cite{fulton2013intersection} on Koszul complexes to evaluate the cohomology.
 \begin{Definition} For polynomial ring $R = \mathbb{C}[x_1,...,x_n]$ sequence of polynomial functions $\{F_i(x)\}, \;\; i=1..n$ is a regular sequence if 
 for $i=1..n$ the function $F_i(x)$ is a non-zero divisor in $R / (F_1,..., F_{i-1})$.
 \end{Definition}
\begin{Theorem}\label{thm_regular_koszul_cohom}   For regular  sequence  $\{F_i(x)\}$ the  cohomology of Koszul complex 
\be
H_0 (D_F) = \mathbb{C}[x]/( F_1,.., F_n),\;\;\; H_{k}(D_F) = 0,\;\;  k>0.
\ee
\end{Theorem} 
In case of non-regular sequences $\{F_i\}$ we can use syzygies (relations) to estimate the cohomology in higher gradings.
\begin{Definition}
For a generating set $F_1,..,F_n$ a relation, or first syzygy,  is a $n$-tuple $G^1,...,G^n$ of functions such that 
\be
G^k F_k = G^1 F_1+...+G^n F_n = 0. 
\ee
\end{Definition}
Syzygies form  a module over $\mathbb{C}[x]$.
The first  syzygies are  in one to one correspondence with  1-cycles  $Z_1(D_F)$, i.e.
\be
Z_1(D_F) = \{ G^k x^\ast_k\;|\; G^k F_k  = 0\}.
\ee
The  1-boundaries 
\be
B_1(D_F) = \{ c^{ij}  F_i (x) x^\ast_j \;|\; c^{ij} = -c^{ji} \in \mathbb{C}[x]\}
\ee
form a subspace in the space $C_1^F$   defined by 
\be
C_1^F   =  \{ c^{ij}F_i (x)  x^\ast_j  \;|\; c^{ij}  \in \mathbb{C}[x]\} \subset C_1.
\ee
\begin{Definition}  The  syzygy-theoretic part of the Koszul cohomology is 
\be
 \widetilde{H}_1(F) = Z_1(D_F) / ( Z_1(D_F)\cap C_1^F ).
 \ee
\end{Definition}
We can estimate the first Koszul cohomology using  the syzygy-theoretic part of the Koszul cohomology.
\begin{Proposition}\label{prop_Koszul_estimation} There exist a surjective map 
\be
H_1(D_F) \to \widetilde{H}_1(F) 
\ee
\end{Proposition}
\begin{proof} By construction the map $i: B_1(D_1) \hookrightarrow  C_1^F$ is an injective map,  moreover the map  $B_1(D_1) \hookrightarrow C_1^F \cap Z_1(D_F)$ is also injective,   hence the corresponding quotients of the same space $Z_1(D_F)$ have surjective map.
\end{proof}
For a regular sequence $\{ F_i\}$ the  module of 1-syzygies is generated by the relations $G_{ij} = (G^1_{ij},...,G^n_{ij})$,  such that  $G^i_{ij} = F_j,\;\; G^j_{ij} = -F_i$ and $G^k_{ij} = 0$ otherwise. Both syzygy-theoretic part and the whole first 
Koszul cohomology vanish in that case.

For a physical system with a symmetry  the invariance of the classical action (\ref{def_action_invar})  implies existence of  syzygies 
\be
v_\a (S) = v_\a^i \p_i S = v_\a^i F_i(x) = 0. 
\ee
If  the vector field components $ v_\a^i(x) $ are non-trivial module gradient ideal $(\p_i S)$,  there is non-trivial syzygy-theoretic part of cohomology $H_1(F)$ and by proposition \ref{prop_Koszul_estimation} the first Koszul cohomology $H_1(D_S)$ are non-trivial.

\begin{Example}\label{exampl_koszul_cohom_rotation} Let us consider a physical system on $\mathbb{R}^2$ with coordinates $x,y$,  action  $S(x,y) = x^2+y^2 + g( x^2 + y^2)^2$ and $SO(2)$-rotational symmetry generated by the vector field 
\be
v = x\p_y - y\p_x.
\ee
The 1-syzygy associated to this symmetry  is $G = (x,-y)$, since  $v(S) = x F_y - y F_x =0$. The syzygy gives us a cohomology representative $x y^\ast  - y x^\ast$ when $g \neq 0$.
Indeed, the all  possible exact elements are of the form 
\be
D_F (c(x,y) x^\ast y^\ast) = c(x,y) (2x+4gx^3+4gxy^2) y^\ast -   c(x,y) (2y + 4g y^3 +4 gy x^2) x^\ast
\ee
and there are no polynomial solutions to the system  for $g\neq 0$
\be
\begin{split}
x &= c(x,y) (2x+4gx^3+4gxy^2)  \\
y &=  c(x,y) (2y + 4 gy^3 +4 gy x^2).
\end{split}
\ee
For $g=0$ there is a solution $c(x,y) = \frac12$ and the 1-cycle $x y^\ast  - y x^\ast$ is exact and cohomology is trivial.

\end{Example}

We can multiply module an ideal $(F)$ the  elements of $\widetilde{H}_1(F)$ to construct $D_F$-closed elements of higher degree
\be
c^{\a_1...\a_k}\;(G_{\a_1}^i x_i^\ast)\cdot...\cdot (G_{\a_k}^i x_i^\ast)\;\;\; \hbox{mod} (F).
\ee
The k-cycles    belong to   k-cohomology $H_k (D_F)$, hence for a physical system with multiple symmetries we expect to have non-trivial higher Koszul cohomology. 

\begin{Example}  Let us modify an example \ref{exampl_koszul_cohom_rotation} with an extra variable, i.e. consider a physical system on $\mathbb{R}^3$ with coordinates $x,y,z$, the action  $S(x,y,z) = x^2+y^2+z^2  +(x^2+y^2+z^2)^2 $ with $SO(3)$-rotational symmetry generated by 
\be
v_z = x\p_y - y\p_x,\;\; v_x = z\p_y - y\p_z,\;\; v_y = x\p_z - z\p_x.
\ee
The 1-syzygies associated to this symmetry 
\be
v_z(S) = x F_y - y F_x =0,\;\;\; v_y(S) = x F_z - z F_x =0,\;\;\; v_x(S) = z F_y - y F_z =0
\ee
give us  cohomology representatives 
\be
x y^\ast  - y x^\ast,\;\;\; x z^\ast  - z x^\ast ,\;\; z y^\ast  - y z^\ast.  
\ee
The pairwise products also give cohomology, for example  the degree 2 cohomology
\be\label{eq_example_deg_2_cohom}
(x y^\ast  - y x^\ast)(x z^\ast  - z x^\ast) = x^2y^\ast z^\ast +xy z^\ast x^\ast +xz x^\ast y^\ast.
\ee
The only exact element at degree 2 is 
\be\nn
\begin{split}
D_F &(c(x,y,z)x^\ast y^\ast z^\ast) = c(x,y,z)(2x+4x^3+4xy^2+4xz^2) y^\ast z^\ast \\
 & c(x,y,z)(2y + 4 y^3 +4 y x^2+ 4y z^2) x^\ast z^\ast +c(x,y,z) (2z +4 z^3 +4z x^2+4z y^2) x^\ast y^\ast
\end{split}
\ee
and there is no polynomial $c(x,y,z)$  to make (\ref{eq_example_deg_2_cohom}) exact.

\end{Example}

\begin{Remark} Our construction for  Koszul cohomology in terms of syzygy   is similar to the construction of the Tate resolution for the Jacobi ring $ \mathbb{C}[x]/( F_1,.., F_n)$. The key difference is that we keep the 
classical action and the odd symplectic manifold fixed while it is possible to modify the action by including the global symmetries. Such approach, discussed in \cite{felder2014classical},  leads to a different perturbative solution of the classical master equation
and it does not have obstructions.  
\end{Remark}

\subsection{Koszul cohomology for supersymmetric systems}

We want to generalize the results and arguments from previous section for   polynomial $S(x)$ to the case of a smooth function $S(x, \psi)$ on a supermanifold. We will do the generalization in two steps. The first step is to turn polynomial actions into 
smooth actions, while the second one is to change the commuting variables to the supercommuting ones.   Unfortunately, we cannot  keep the same level  of precision as in polynomial case, but we plan to continue our work on  this subject. Hence, the present section will contain conjectures, supporting arguments and examples.

The first step is to generalize the regular sequence case.
\begin{Definition} The gradient ideal $(\p_i S)$ forms a regular sequence in $C^\infty(X)$ if the  intersection of graph$(dS)$ and zero section of $T^\ast X$ is transversal. In particular, when $S(x)$ has simple isolated critical points the corresponding sequence is always regular.
\end{Definition}

For  a smooth function  $S: X \to \mathbb{R}$  the Koszul cohomology can be expressed using the critical submanifold  $X_{crit} = \{ x \in X\; |\; dS(x) = 0\}$.
\begin{Proposition} For  a regular sequence  $\{\p_iS(x)\}$ the Koszul  cohomology
\be
H_0 (D_S) = C^\infty(X_{crit}) ,\;\;\; H_{k}(D_S) = 0,\;\;  k>0.
\ee
\end{Proposition}
\begin{proof}
Surprisingly,  the differential (\ref{eq_Koszul_diff}), is identical to the B-model differential $Q_W$. For a superpotential $W$ with isolated critical points we constructed \cite{Losev:2023bhj} a  non-degenerate paring to show that the cohomology are the same as predicted by the  proposition.
\end{proof}

The cohomology class of a function $\Pi(x)$ is given by  its restriction to the critical set, i.e. 
\be
\left[\;\Pi(x)\;\right] = \Pi(x)\Big|_{X_{crit}}.
\ee
In particular $\Pi(x)$ has trivial cohomology class when it vanishes on a critical set $X_{crit}$.

The generalization of a syzygy is a vector field on $X$ preserving $S$, while the 1-cycles  are constructed from such vector fields, i.e. 
\be
Z_1(D_S) =  \{ v^k x^\ast_k\;|\;  v \in \hbox{Vect}(X),\;\; v (S)=0\}.
\ee
The syzygy-theoretic part of the cohomology 
\be
 \widetilde{H}_1(D_S) \simeq \{ v \in \hbox{Vect}(X),\;\; v (S)=0,\;\; v|_{X_{crit}}\neq 0\}.
\ee
\begin{Proposition} For a physical system with action $S: X \to \mathbb{R}$ and a  symmetry $v \in Vect(X)$, which acts non-trivially on critical set $X_{crit}$ the  Koszul cohomology are non-trivial in higher gradings.
\end{Proposition}
\begin{proof}  By the same logic of 1-boundary embedding we have a surjective map $H_1(D_S) \to \widetilde{H}_1(D_S)$ while the $\widetilde{H}_1(D_S)$ is non-trivial for the system with a symmetry, acting non-trivially on critical set.
\end{proof}
\begin{Example} For a system on $\mathbb{R}^2$ with coordinates $x,y$, the action  $S(x,y) = x^2+y^2 + g ( x^2 +y^2)^2$  the $SO(2)$ rotational symmetry acts non-trivially on a critical set for $g\neq 0$, while for $g=0$ the action is trivial.
For $g=0$  the sequence $\p_i S$ is regular  and there are no higher degree  Koszul cohomology.
\end{Example}

For a supersymmetric system we have both even and odd variables hence we need a generalization of the Koszul complex results to the case of superalgebra.  We conjecture that  the commutative algebra results for the Koszul cohomology continue to hold, possibly with some minor adjustments.  In particular, the 1-cochains have non-trivial  dependence on the global symmetry parameters $c^\a$
\be
Z_1(D_S) =  \{ c^\a v_\a^k \phi^\ast_k\;|\;  v \in \hbox{Vect}(X),\;\; v (S)=0\}.
\ee

The supersymmetry transformations $Q$ and $\bar{Q}$ imply existence of two syzygies. Moreover, 
the $\mathcal{N}=2$ supersymmetry is equipped with $U(1)$ R-symmetry, rotating between $\psi$ and $\bar{\psi}$, generated by 
\be
v_{U(1)} = \psi^k \frac{\p}{\p \psi^k} - \bar{\psi}^k \frac{\p}{\p \bar{\psi}^k}.
\ee

\subsection{Cohomology for a single supermultiplet}\label{sec_cohom_susy_homotopy}
In the case of a single supermultiplet we have two additional symmetries generated by 
\be
A = \psi \frac{\p}{\p \bar{\psi}},\;\;\bar{A} = \bar{\psi} \frac{\p}{\p \psi}.
\ee
Hence, in general, we expect to have non-trivial $H_1(D_F)$ and possibly higher $H_k(D_F)$ for a supersymmetric system and we need to check for the obstructions to the master equation solution for an on-shell supersymmetric system.

\begin{Example}\label{ex_quadratic_action_cohom}
The  supersymmetric system with  quadratic  action  (\ref{eq_gauss_susy_action}) the  Koszul differential is
\be
D_S =2x  \frac{\p}{\p x^\ast} -\bar{\psi} \frac{\p}{\p \psi^\ast}+  \psi  \frac{\p}{\p \bar{\psi}^\ast}.
\ee
 The gradient ideal   $(F_x, F_{\psi}, F_{\bar{\psi}})$ is generated by 
\be
F_{x} = 2x ,\;\;\; F_{\psi}=-\bar{\psi},\;\;\; F_{\bar{\psi}}= \psi 
\ee
 which form a regular sequence, hence there are no  antifield-dependent cohomology. 
 \end{Example}
\begin{Example} For the quadratic deformation (\ref{eq_action_quadr_deform}) the  Koszul differential is
\be
D_S =2x (1+2gx^2 + 2g \psi\bar{\psi}) \frac{\p}{\p x^\ast} - (1+2g x^2)\bar{\psi} \frac{\p}{\p \psi^\ast}+ (1+2gx^2) \psi  \frac{\p}{\p \bar{\psi}^\ast}.
\ee
The functions  $(F_x, F_{\psi}, F_{\bar{\psi}})$ are
\be
F_{x} = 2x (1+2gx^2 + 2g \psi\bar{\psi}),\;\;\; F_{\psi}=- (1+2g x^2)\bar{\psi},\;\;\; F_{\bar{\psi}}= (1+2gx^2) \psi 
\ee
do not form a regular sequence.  There are  non-trivial syzygies  at level-1 from supersymmetry vector fields  and U(1) R-symmetry 
\be
\begin{split}
0&=\psi  F_{\bar{\psi}} - \bar{\psi} F_{\psi},  \\
0&=\psi  F_{x} - 2x F_{\bar{\psi}},  \\
0&=\bar{\psi}  F_{x} +2x F_{\psi}.  
\end{split}
\ee
The BV bracket for   the supersymmetry transfromations
\be
\{  \mathcal{Q}, \mathcal{Q}\}  =\{  \ve \bar{\psi}x^\ast -2 \ve x\psi^\ast,  \ve \bar{\psi} x^\ast -2 \ve x\psi^\ast \}   =  4 \ve^2\bar{\psi} \psi^\ast 
\ee
is non-trivial in cohomology over polynomial ring, i.e. 
\be
[\ve^2\bar{\psi} \psi^\ast] \neq 0 \in H^\ast (D_S, \mathbb{C}[x, \psi, \bar{\psi}, \ve]).
\ee

However, in perturbation theory over  $g$  an expression   $1+2g x^2$ is invertible, i.e.
\be
\frac{1}{1+2gx^2} = 1-2g x^2 +4g^2x^4+...
\ee
hence we can represent the syzygy components on terms of the gradient ideal 
\be
\psi =\frac{1}{1+2gx^2}  F_{\bar{\psi}},\;\;\; \bar{\psi} = -\frac{1}{1+2gx^2}  F_{\psi},\;\;\; 2x = \frac{1}{1+2gx^2 + 2g \psi\bar{\psi}} F_{x}.  
\ee
We conclude that  the sequence $(F_x, F_{\psi}, F_{\bar{\psi}})$ is  regular over formal series in $x$ and $g$ and the higher degree Koszul cohomology are trivial.
The solution for a   bivector is 
\be
D_S(\Pi^{(2)}) = -\frac12 \{  \mathcal{Q}, \mathcal{Q}\} = - 2 \ve^2 \bar{\psi} \psi^\ast \Longrightarrow \Pi^{(2)}  = \frac{\ve\psi^\ast \ve\psi^\ast}{1+2gx^2}.
\ee
\end{Example}

\begin{Example}\label{ex_obstruction} We can take a large $g$ limit of the action (\ref{eq_action_quadr_deform}) by keeping just quadratic  deformation term.  The  Koszul differential  in this limit is 
\be
D_S =4x (x^2 + \psi\bar{\psi}) \frac{\p}{\p x^\ast} - 2 x^2\bar{\psi} \frac{\p}{\p \psi^\ast}+ 2x^2 \psi  \frac{\p}{\p \bar{\psi}^\ast}.
\ee
 The functions  $(F_x, F_{\psi}, F_{\bar{\psi}})$ given by 
\be
F_{x} = 4x (x^2 +  \psi\bar{\psi}),\;\;\; F_{\psi}= -2 x^2\bar{\psi},\;\;\; F_{\bar{\psi}}= 2x^2 \psi 
\ee
do not form a regular sequence over polynomial ring.  There are  non-trivial syzygies  at level-1
\be
\begin{split}
0&=\psi  F_{\bar{\psi}} - \bar{\psi} F_{\psi},  \\
0&=\psi  F_{x} - 2x F_{\bar{\psi}},  \\
0&=\bar{\psi}  F_{x} +2x F_{\psi}.  
\end{split}
\ee
The first cohomology are non-trivial, moreover the BV square of the symmetry vector field represents non-trivial class
\be
[\ve^2 \bar{\psi} \psi^\ast] \neq 0 \in H^\ast (D_S, \mathbb{C}[x, \psi, \bar{\psi}, \ve]).
\ee
The $x^2$ remains non-invertible in cohomology over the ring of formal series in $x$, i.e.
\be
[\ve^2 \bar{\psi} \psi^\ast]  \neq  0 \in H^\ast (D_S, \mathbb{C}[[x, \psi, \bar{\psi}, \ve]]).
\ee
We can find  a  formal bivector, solving the defining  equation  
\be
\Pi^{(2)}  = \frac{\ve \psi^\ast \ve\psi^\ast}{2x^2},\;\;  D_S(\Pi^{(2)}) = -\frac12 \{  \mathcal{Q}, \mathcal{Q}\}  =-  2\ve^2 \bar{\psi} \psi^\ast. 
\ee
However the bivector  has  pole at $x=0$ and hence does not exist for all values of $x$.
\end{Example}

In several examples we inverted the $D_S$ differential on simple expressions. For the more general applications it is convenient to construct the $D_S$-homotopy and use it to express the solutions. For a single supermultiplet the most general action is 
\be
S = h(x) + f(x)\psi\bar{\psi},
\ee
while the corresponding Koszul differential is
\be
D_S = (h'(x) + f'(x) \psi\bar{\psi}) \frac{\p}{\p x^\ast} -  f(x) \bar{\psi} \frac{\p}{\p \psi^\ast}+ f(x) \psi  \frac{\p}{\p \bar{\psi}^\ast}.
\ee

There is a natural candidate for the homotopy 
\be
G = \frac{\bar{\psi}^\ast}{f}  \frac{\p}{\p\psi}- \frac{\psi^\ast}{f}  \frac{\p}{\p\bar{\psi}}. 
\ee 
The graded commutator with Koszul differential evaluates into
\be\label{eg_comm_unnorm_homotopy}
[ D_S, G\}  =  \psi  \frac{\p}{\p\psi} + \bar{\psi}  \frac{\p}{\p \bar{\psi}} + \psi^\ast  \frac{\p}{\p \psi^\ast} + \bar{\psi}^\ast  \frac{\p}{\p \bar{\psi}^\ast}+ \frac{f'}{f} (\bar{\psi}^\ast \bar{\psi}  +\psi^\ast \psi ) \frac{\p}{\p x^\ast}.
\ee
On the $x^\ast$-independent subspace the right hand side of  (\ref{eg_comm_unnorm_homotopy}) becomes a graded version of the Euler vector field, measuring the total degree in odd fields and corresponding antifields. The division by a total degree can be realized as an integral, hence    on the $x^\ast$-independent subspace  we have the $D_S$-homotopy in the following form 
\be\label{eq_d_s_homotopy}
K \mathcal{R} =  \left( \frac{\bar{\psi}^\ast}{f}  \frac{\p}{\p\psi}- \frac{\psi^\ast}{f}  \frac{\p}{\p\bar{\psi}}\right) \int^\infty_0 dt\;\mathcal{R}(e^{-t}\psi^\ast, e^{-t} \bar{\psi}^\ast,  e^{-t}\psi, e^{-t}\bar{\psi}, x).
\ee
\begin{Lemma} The homotopy $K$  obeys $[\Delta_\mu, K\} = 0$ for the standard measure $\mu$ on a subspace of linear in $\psi, \bar{\psi}$ functions.
\end{Lemma}
\begin{proof}  For the  standart measure $\mu  = dxdx^\ast d\psi d\psi^\ast d\bar{\psi} d\bar{\psi}^\ast$ the BV Laplacian 
\be
\Delta_\mu =\frac{\p^2 }{\p x  \p x^\ast } -\frac{\p^2 }{\p \psi  \p \psi^\ast } -\frac{\p^2 }{\p \bar{\psi}  \p \bar{\psi}^\ast} 
\ee
The graded commutator  evaluates into
\be
\begin{split}
[\Delta_\mu, G\} &=\left[ \frac{\p^2 }{\p x  \p x^\ast } -\frac{\p^2 }{\p \psi  \p \psi^\ast } -\frac{\p^2 }{\p \bar{\psi}  \p \bar{\psi}^\ast},  \frac{\bar{\psi}^\ast}{f}  \frac{\p}{\p\psi}- \frac{\psi^\ast}{f}  \frac{\p}{\p\bar{\psi}} \right\} \\
& = \left( \frac{f'\bar{\psi}^\ast}{f^2}  \frac{\p}{\p\psi} - \frac{f'\psi^\ast}{f^2}  \frac{\p}{\p\bar{\psi}} \right) \frac{\p}{\p x^\ast} + \frac{2}{f}\frac{\p^2 }{\p \psi  \p \bar{\psi} }.
\end{split}
\ee
The graded commutator vanishes on $x^\ast$-independent subspace, where the homotopy is defined. The linearity in  $\psi, \bar{\psi}$ makes the second term vanish hence completing the proof.
\end{proof}

In section \ref{sec_homotopy_bv_transform} we investigated the homotopy dependence of the BV action, so the choice of homotopy typically changes the antifield  dependence  of the BV action. Our choice of  homotopy (\ref{eq_d_s_homotopy}) was motivated by the properties of the low-dimensional supersymmetric systems described below.

Many of the on-shell supersymmetric models in low dimensions, studied in the literature, have a common feature: the commutators for the  supersymmetry generators require equations of motion when acting on fermions and the equations of motions are also just the one for fermions. Hence,  according to the proposition \ref{prop_on_shell_bv_leading} the corresponding  bivectors will have only antifields for the odd variables,  such as $\psi^\ast, \bar{\psi}^\ast$, and no dependence on the antifields for the even variables,  such as  $x^\ast$. One possible explanation for this particular feature is that   the low dimensional supersymmetric theories  typically have a single even auxiliary field, hence its elimination via the BV integration  gives us BV polyvector which is independent on $x^\ast$, i.e. are  polynomial functions on antifields  $\psi^\ast$ and $\bar{\psi}^\ast$.

\section{BV formulation of $d=0$ supersymmetry}

In this section we will adopt the BV formalism results from section \ref{sec_bv_formalism}  to describe the off-/on- shell supersymmetry for a single $d=0,\; \mathcal{N}=2$ supermultiplet.  Let us list the relevant simplifications 

\begin{itemize}
\item The $d=0$ supersymmetry algebra (\ref{N=2_offshel_susy}) has trivial structure constants i.e. $f_{\a\b}^\g =0$. Hence  we do not have an $\mathcal{F}$-term (\ref{def_symmetryt_charge_structure_const}) and  (\ref{def_qm_offshel_symm_measure}) simplifies to the  invariance of the measure.  

\item Since we do not have an $\mathcal{F}$-term we do not need to include $T^\ast[1] \mathfrak{g}[1]$ into the odd symplectic space (\ref{bv_structure_offshel_symm})  and the variables $c^\a$ become parameters of the master equation solution. 

\item All supersymmetry generators are odd hence all parameters  $c^\a$ are even, and  we can make them arbitrary large. Conventionally even parameters $c^\a$ are denoted as $\ve^\a$.  
 We will limit our attention to the $\mathcal{N}=2$ supersymmetry hence we have only two parameters, denoted by  $\ve$ and $\bar{\ve}$. 

\item The antifields $\psi^\ast, \bar{\psi}^\ast$ for the odd fields  $\psi, \bar{\psi}$  are even,  hence we can have an infinitely many polyvectors in BV action.

\end{itemize}

\subsection{BV description of the off-shell superpotential}\label{sec_off_shell_superpotent}

In section \ref{sub_sect_offshel_superpotential} we described a theory of superpotential   with off-shell supersymmetry.  According to the proposition (\ref{prop_cls_master_solution_off_shell}) we can construct the corresponding solution to the classical master equation. 
The fields of  the system are coordinates on  supermanifold $X = \mathbb{R}^{2|2}$, two even coordinates  $x, F$ and two odd coordinates $\psi, \bar{\psi}$.  The  action is 
\be
S(x, F, \psi, \bar{\psi})  =- \frac14  F^2+  FW(x) + \psi  \bar{\psi} W'(x). 
\ee
The action is  invariant under the  supersymmetry algebra $\mathfrak{g} = \mathbb{R}^{0|2}$ acting  via  the odd vector fields  on $X$ 
\be\label{eq_smtr_vect_fields}
Q = \bar{\psi} \frac{\p}{\p x} - F \frac{\p }{\p \psi},\;\; \bar{Q} = \psi \frac{\p}{\p x} + F \frac{\p }{\p \bar{\psi}}.  
\ee
We introduce the  antifields $ \phi^\ast  = (F^\ast, x^\ast, \psi^\ast, \bar{\psi}^\ast)$ with the  opposite parity. The odd symplectic space is the super Euclidean space
\be\label{eq_sympl_off_shell_superp_model}
\mathcal{M} = \mathbb{R}^{4|4},\;\; \omega = dx\wedge dx^\ast  + dF \wedge dF^\ast  - d\psi \wedge d\psi^\ast - d\bar{\psi}\wedge d\bar{\psi}^\ast.
\ee 
 The proposition \ref{prop_cls_master_solution_off_shell} applied to our supersymmetric system  gives us the BV action in the form 
\be\label{eq_off_susy_superpotential}
\mathcal{S}= -\frac14   F^2 +  FW(x) + \psi  \bar{\psi} W'(x)-   \ve  F  \psi^\ast+   \bar{\ve}  F  \bar{\psi}^\ast  + ( \bar{\ve} \psi + \ve \bar{\psi})x^\ast. 
\ee
The standard  measure $\mu_X = dF  dx d\psi  d\bar{\psi}$  on $X$  is preserved by the vector fields (\ref{eq_smtr_vect_fields}), i.e. 
\be
\hbox{div}_{\mu_X} Q = \hbox{div}_{\mu_X} \bar{Q} = 0.
\ee
Hence according to the proposition (\ref{prop_quant_master_solution_off_shell}) the BV action (\ref{eq_off_susy_superpotential}) solves the quantum master equation  (\ref{quantum_mast_eq})
for the BV Laplacian 
\be
\Delta_\mu = \frac{\p^2 }{\p \phi^a  \p \phi^\ast_a} =\frac{\p^2 }{\p x  \p x^\ast }+\frac{\p^2 }{\p F  \p F^\ast } -\frac{\p^2 }{\p \psi  \p \psi^\ast } -\frac{\p^2 }{\p \bar{\psi}  \p \bar{\psi}^\ast } 
\ee
constructed with respect to the   Berezinian $\mu  = dF  dF^\ast  dx dx^\ast  d\psi  d\psi^\ast  d\bar{\psi} d\bar{\psi}^\ast$  on $\mathcal{M}$.

\subsection{BV description for on-shell superpotential}\label{sec_on_shell_superpotent}

In section \ref{sub_sect_onshel_superpotential} we described a theory of superpotential   with on-shell supersymmetry.  According to the proposition (\ref{prop_on_shell_bv}) we can construct the corresponding approximate solution to the classical master equation. 
The fields in a system are coordinates on a supermanifold $X = \mathbb{R}^{1|2}$, single even coordinate  $x$ and two odd coordinates $\psi, \bar{\psi}$.  The  action is 
\be
S(x, \psi, \bar{\psi})  =W^2(x) +W'(x) \psi\bar{\psi}. 
\ee
The action is invariant under the  supersymmetry algebra $\mathfrak{g} = \mathbb{R}^{0|2}$  acting  via  the odd vector fields  on $X$ 
\be\label{eq_vect_fields_susy}
 Q  = \bar{\psi} \frac{\p}{\p x} - 2W(x) \frac{\p }{\p \psi},\;\; \bar{Q} = \psi \frac{\p}{\p x} + 2W(x) \frac{\p }{\p \bar{\psi}}.
\ee
We introduce the   antifields $ \phi^\ast  = ( x^\ast, \psi^\ast, \bar{\psi}^\ast)$ with  opposite parity. The odd symplectic space is the super Euclidean space
\be
\mathcal{M} = \mathbb{R}^{3|3},\;\; \omega = dx\wedge dx^\ast   - d\psi \wedge d\psi^\ast - d\bar{\psi}\wedge d\bar{\psi}^\ast.
\ee 
 The proposition (\ref{prop_on_shell_bv})  applied to our supersymmetric system  gives us the BV action in the form 
\be\label{eq_on_susy_superpotential}
\mathcal{S}= W^2(x) + W'(x) \psi\bar{\psi}-2  \ve  W(x) \psi^\ast + 2\bar{\ve} W(x) \bar{\psi}^\ast   +( \bar{\ve} \psi + \ve \bar{\psi}) x^\ast . 
\ee
The supersymmetry algebra (\ref{eq_on_shell_susy_algebra_supermult}) has distinct feature that the factors in front of the vector fields are identical to the equations of motion for odd fields. Hence, we can present a
refined version of the on-shell condition 
\be\label{eq_on_shell_superpotential_refined_algebra}
\begin{split}
[Q, Q\} &=-4W'(x) \bar{\psi} \frac{\p}{\p \psi} =- 4 \frac{\p S}{\p \psi }  \frac{\p}{\p \psi},\;\;\;  [\bar{Q}, \bar{Q}\} =4W'(x) \psi \frac{\p}{\p \bar{\psi}}= -4 \frac{\p S}{\p \bar{\psi} }  \frac{\p}{\p \bar{\psi}},\\
 [ \bar{Q},Q\} &= 2 W'(x) \left(\bar{\psi} \frac{\p}{\p \bar{\psi}} -\psi \frac{\p}{\p \psi}  \right) = 2 \frac{\p S}{\p \psi }  \frac{\p}{\p \bar{\psi}} +2 \frac{\p S}{\p \bar{\psi} }  \frac{\p}{\p \psi}.
\end{split}
\ee
We use the proposition \ref{prop_on_shell_bv_leading} to evaluate the bivector $\Pi^{(2)} =  ( \psi^\ast  \ve -  \bar{\psi}^\ast \bar{\ve} )^2$ from the refined version of the on-shell algebra. In particular,  the corrected BV action is 
\be\label{eq_on_susy_superpotential_bivector}
\mathcal{S}= W^2(x) + W'(x) \psi\bar{\psi}-2\ve    W(x) \psi^\ast + 2  \bar{\ve} W(x)  \bar{\psi}^\ast+  ( \bar{\ve} \psi + \ve \bar{\psi})x^\ast+  ( \psi^\ast  \ve -  \bar{\psi}^\ast \bar{\ve} )^2.
\ee
Alternatively we can solve the bivector  equation 
\be\label{eq_on_shell_superpoetntial_bivect_equation}
\begin{split}
D_{S} \Pi^{(2)} &=-\frac12  \{ \mathcal{Q}, \mathcal{Q}\} = 2W'(x)( \bar{\ve} \psi +\ve\bar{\psi})(  \bar{\ve}   \bar{\psi}^\ast -\ve   \psi^\ast ).
\end{split}
\ee
In example \ref{ex_quadratic_action_cohom} we observed that there is no higher degree cohomology for the $D_{S}$ in case $W(x)=x$. For a general $W(x)$ the higher degree cohomology of $D_S$ are non-trivial, but the on-shellness of the supersymmetry algebra 
ensures that righthandside of (\ref{eq_on_shell_superpoetntial_bivect_equation}) belongs to the trivial class in $D_S$. Since it is also independent on $x^\ast$ we can use the homotopy (\ref{eq_d_s_homotopy}) to evaluate 
\be\label{eq_on_shell_superpot_bivect}
\Pi^{(2)} =-\frac12 K  \{ \mathcal{Q}, \mathcal{Q}\}=   ( \psi^\ast  \ve -  \bar{\psi}^\ast \bar{\ve} )^2.
\ee
The bivector (\ref{eq_on_shell_superpot_bivect}) is a constant bivector, hence it automatically satisfies conditions (\ref{def_pi_2_algebra_inv}, \ref{def_pi_2_bracket_trivial}). Hence we conclude that the BV action (\ref{eq_on_susy_superpotential_bivector})  is a solution to the classical master equation.

The standard  measure $\mu_X = dx d\psi  d\bar{\psi}$  on $X$  is preserved by the vector fields (\ref{eq_vect_fields_susy})
\be
\hbox{div}_{\mu_X} Q = \hbox{div}_{\mu_X} \bar{Q} = 0.
\ee
Hence according to the proposition \ref{eq_quant_bv_action_bivector_refined} the BV action (\ref{eq_on_susy_superpotential_bivector})  solves the quantum master equation (\ref{quantum_mast_eq}) for BV-Laplacian 
\be
\Delta_\mu = \frac{\p^2 }{\p \phi^a  \p \phi^\ast_a} =\frac{\p^2 }{\p x  \p x^\ast } -\frac{\p^2 }{\p \psi  \p \psi^\ast } -\frac{\p^2 }{\p \bar{\psi}  \p \bar{\psi}^\ast } 
\ee
constructed with respect to the   Berezinian  $\mu  = dx dx^\ast    d\psi d\psi^\ast d\bar{\psi} d\bar{\psi}^\ast$ on $\mathcal{M}$.

\subsection{BV induction for the theory with superpotential}\label{sec_bv_induction_superpotential}

\begin{Proposition}
The quantum physical system  for on-shell superpotential  with BV action (\ref{eq_on_susy_superpotential_bivector}) and the quantum theory for off-shell superpotential (\ref{eq_off_susy_superpotential}) are related by the BV induction (\ref{eq_BV_induction}).
\end{Proposition}

\begin{proof}
The odd symplectic space for off-shell theory $\mathcal{M} = \mathbb{R}^{4,4}$  is a direct of  the odd symplectic space for on-shell theory  $\mathcal{M}' = \mathbb{R}^{3,3}$ and an 
odd symplectic space $\mathcal{M}''= \mathbb{R}^{1,1}$ describing an auxiliary field $F$. We choose a Lagrangian sub-manifold $\mathcal{L}'' = \{ F^\ast = 0\}$ inside $\mathcal{M}''$
and perform a BV induction  (\ref{eq_BV_induction}) 
\be\label{eq_BV_integral_superpotential_induction}
\mathcal{C}(\hbar)\cdot e^{- \frac1 {\hbar} \mathcal{S}_{ind}} \sqrt{\mu'}= \int_{\mathcal{L}''} \sqrt{\mu}\;\;e^{- \frac 1{\hbar} \mathcal{S}}.
\ee
The Berezinians  on $\mathcal{M}$ and $\mathcal{M}'$ are standard ones,  hence the   integral  over $F$ is Gaussian
\be\nn
\mathcal{C}(\hbar)  e^{- \frac1 {\hbar}\mathcal{S}_{ind}} = \int_{\mathbb{R}} dF \; \exp \frac{1}{\hbar}\left(\frac14 F^2 -   FW(x) - \psi  \bar{\psi} W'(x) +  F   \ve \psi^\ast-  F \bar{\ve}\bar{\psi}^\ast -( \bar{\ve} \psi + \ve \bar{\psi})  x^\ast \right).
\ee
We choose a   function $\mathcal{C}(\hbar) =   \sqrt{-4\pi}$, so the  integral evaluates into 
\be
\mathcal{S}_{ind} =  \psi  \bar{\psi} W'(x)     +  (W(x) -  \psi^\ast  \ve + \bar{\psi}^\ast \bar{\ve} )^2+  ( \bar{\ve} \psi + \ve \bar{\psi})x^\ast. 
\ee
The induced BV action  is identical to the on-shell superpotential  with BV action (\ref{eq_on_susy_superpotential_bivector})  and completes the proof of the proposition.
\end{proof}

We can reconstruct the on-shell supersymmetric system from the  induced  BV action $\mathcal{S}_{ind}$. The classical action $S(\phi)$  is  the antifield-independent part
\be
S_{ind} (x,\psi,\bar{\psi})= \mathcal{S}_{ind} \Big|_{\phi^\ast=0} =  W^2(x) + W'(x) \psi\bar{\psi}
\ee
which identical to the induced action (\ref{eff_action_W}). 
The supersymmetry vector fields  are encoded in the  $\phi^\ast$-linear  part and are  identical to the transformations (\ref{eq_induced_susy}).

\subsection{BV formulation for two deformations}
For the quadratic deformation model with the action (\ref{eq_action_quadr_deform})
the on-shell algebra of supersymmetry vector fields  is the same for all values of $g$, while the  algebra refinements are different due to the difference in the equations of motion, i.e. 
\be\nn
\begin{split}
[Q, Q\} &=-4 \bar{\psi} \frac{\p}{\p \psi} =-4 \frac{\p S_0}{\p \psi } \frac{\p}{\p \psi}  =-\frac{4}{1+2gx^2} \frac{\p S_g}{\p \psi } \frac{\p}{\p \psi}  ,\\
[\bar{Q}, \bar{Q}\} &=4 \psi \frac{\p}{\p \bar{\psi}}= 4 \frac{\p S_0}{\p \bar{\psi} } \frac{\p}{\p \bar{\psi}}  =\frac{4}{1+2gx^2} \frac{\p S_g}{\p \bar{\psi} } \frac{\p}{\p \bar{\psi}} ,\\
[ \bar{Q},Q\} &= 2  \left(\bar{\psi} \frac{\p}{\p \bar{\psi}} -\psi \frac{\p}{\p \psi} \right) =   2\left( \frac{\p S_0}{\p \psi }  \frac{\p}{\p \bar{\psi}} + \frac{\p S_0}{\p \bar{\psi} }  \frac{\p}{\p \psi}\right)  =  \frac{2}{1+2gx^2}\left( \frac{\p S_g}{\p \psi }  \frac{\p}{\p \bar{\psi}} +2 \frac{\p S_g}{\p \bar{\psi} }  \frac{\p}{\p \psi} \right).
\end{split}
\ee
The two refinements just differ by an overall factor of $1+2gx^2$,  hence by  proposition \ref{prop_on_shell_bv_leading}  the corresponding bivectors are differ by the same factor. We have already evaluated the bivector (\ref{eq_on_shell_superpot_bivect}) for an arbitrary superpotential and we can use it to evaluate 
\be\label{eq_quadr_def_bivect}
\Pi^{(2)}\Big|_{g=0} = (  \psi^\ast  \ve -  \bar{\psi}^\ast \bar{\ve} )^2,\;\;\; \Pi^{(2)} =  \frac{(  \psi^\ast  \ve -  \bar{\psi}^\ast \bar{\ve} )^2}{1+2gx^2}.
\ee
The bivector $\Pi^{(2)} $  for $g=0$ is a constant bivector hence it trivially satisfies conditions (\ref{def_pi_2_algebra_inv}, \ref{def_pi_2_bracket_trivial}), while the bivector $\Pi^{(2)} $ is non-trivial and 
\be\label{eq_pi_3_quadr_def_defining}
\begin{split}
\{\Pi^{(2)},\Pi^{(2)}\} = 0,\;\; \{\Pi^{(2)}, \mathcal{Q}\}=4gx ( \bar{\ve} \psi + \ve \bar{\psi})   \frac{(  \psi^\ast  \ve -  \bar{\psi}^\ast \bar{\ve} )^2}{(1+2gx^2)^2}.  
\end{split}
\ee
We can use the procedure from section \ref{sec_bv_action_with_polyevect} to introduce the higher polyvectors  to construct a solution to the classical master equation.  An expression in (\ref{eq_pi_3_quadr_def_defining}) is $x^\ast$-independent hence we can use the homotopy (\ref{eq_d_s_homotopy}) for the Koszul differential
\be
D_S = 2x(1 +2 gx^2 + 2g  \psi\bar{\psi}  ) \frac{\p}{\p x^\ast} - (1+2g x^2)\bar{\psi} \frac{\p}{\p \psi^\ast} + (1+2g x^2)\psi \frac{\p}{\p \bar{\psi}^\ast}.
\ee
to evaluate the 3-vector term
\be\label{eq_quadr_def_3_vect}
\Pi^{(3)} = -K \{\Pi^{(2)}, \mathcal{Q}\} =   \frac{4gx( \ve  \psi^\ast -   \bar{\ve}  \bar{\psi}^\ast )^3}{3(1+2gx^2)^3}.
\ee
\begin{Proposition}\label{prop_bv_action_def_model} The solution to the  quantum master equation  for the quadratic deformation model is 
\be\label{eq_bv_action_quadr_def_model_full}
\begin{split}
\mathcal{S} &=   \psi\bar{\psi} (1+2g x^2)   +    (\ve \bar{\psi} +  \bar{\ve} \psi )x^\ast   + \sum_{n=0}^\infty \frac{1}{n!} \left(\frac{    \bar{\ve}  \bar{\psi}^\ast - \ve  \psi^\ast  }{1+2gx^2}  \frac{d}{dx}  \right)^n  (x^2+g x^4).
\end{split}
\ee
\end{Proposition}
\begin{proof}   The $n=0,1,2,3$ terms in  the sum (\ref{eq_bv_action_quadr_def_model_full}) match with the even part of the action, symmetry transformations, bivector  (\ref{eq_quadr_def_bivect}) and 3-vector  (\ref{eq_quadr_def_3_vect}).
We will prove the proposition using induction in $n$.  Suppose for all $1<k\leq n$ the $k$-vectors take the form 
\be
\Pi^{(k)} = \frac{1}{k!} \left(\frac{  \bar{\ve}  \bar{\psi}^\ast  -\ve  \psi^\ast }{1+2gx^2}  \frac{d}{dx}  \right)^k  (x^2+g x^4)
\ee
then  they satisfy $\{\Pi^{(k)},\Pi^{(l)}\} = 0$ for all $k, l\geq 2$ and  $\Delta_\mu \Pi^{(k)} = 0$.

 Hence the  equation (\ref{eq_polyvector_equation}) for  the polyvector  $\Pi^{(n+1)}$ becomes a linear problem
\be\label{eq_liner_iteration_polyvectors}
D_S \Pi^{(n+1)} =- \{ \mathcal{Q}, \Pi^{(n)}\}  =(\bar{\ve} \psi +\ve \bar{\psi}) \frac{d}{dx} \Pi^{(n)}.
\ee
The right hand side of the linear problem (\ref{eq_liner_iteration_polyvectors}) is independent on $x^\ast$, hence we can use  the homotopy (\ref{eq_d_s_homotopy})  to evaluate 
\be\nn
\begin{split}
\Pi^{(n+1)} &= -K  \{ \mathcal{Q}, \Pi^{(n)}\}  =    \frac{1}{1+2g x^2}   \left(   \bar{\psi}^\ast   \frac{\p}{\p \psi} - \psi^\ast   \frac{\p }{\p \bar{\psi}} \right) \int^\infty_0 dt\; e^{-(n+1)t} (\bar{\ve} \psi +\ve \bar{\psi})  \frac{d}{dx} \Pi^{(n)}  \\
&= \frac{1}{n+1}\frac{ \bar{\ve}  \bar{\psi}^\ast -\ve  \psi^\ast}{1+2g x^2}\frac{d}{dx}\Pi^{(n)} = \frac{1}{(n+1)!}  \left(\frac{  \bar{\ve}  \bar{\psi}^\ast-\ve  \psi^\ast    }{1+2gx^2}  \frac{d}{dx}  \right)^{n+1}  (x^2+g x^4).
\end{split}
\ee
Our  expression for $\Pi^{(n+1)}$ matches with the assumptions of induction, hence the proof is complete.
\end{proof}

Our analysis of the BV actions for  two supersymmetric models, described in  section \ref{sec_two_deformations}, unveils  the key difference between them. The model with  superpotential deformation  has the BV action (\ref{eq_on_susy_superpotential_bivector}), which is quadratic in antifields, while  the BV action (\ref{eq_bv_action_quadr_def_model_full}), for the  quadratic deformation model,    has  infinitely  many terms in antifield expansion. In section  \ref{sec_localization_refined_conject} we will relate the structure of BV action and possible simplifications for the corresponding partition functions.

\subsection{BV refinement of the on-shell supersymmetry}\label{sec_refined_bv_susy}

Based on our discussion of the master equation solutions from section \ref{sec_bv_formalism} we suggest a refined classification for the physical systems with   supersymmetries   based on the structure of the corresponding BV action with respect to the antifield terms.
\begin{enumerate}
\item {\bf Linear}. Linear in antifields BV action implies that the supersymmetry is realized off-shell.  An example of the system from this class in the off shell description of the superpotential from section \ref{sec_off_shell_superpotent}.
\item {\bf Quadratic}. The master equation solution is at most quadratic in antifields. An example of the system from this class in the on-shell description of the superpotential \ref{sec_on_shell_superpotent}. 
\item {\bf Polynomial}. The order by order  construction of the master equation solution terminates at finitely-many terms.  We will provide an example of the system from this class in section \ref{sec_polynomial_class_bv_system}.
\item {\bf Formal series}. There are no obstructions, but the order by order construction of the master equation solution never terminates.  An example of the system from this class in the quadratic deformation  system  from proposition 
\ref{prop_bv_action_def_model}. 
\item {\bf Obstruction}. The cohomological  obstructions prevent from completion of the leading BV action (classical action and symmetries) into a solution to the classical master equation. In example \ref{ex_obstruction} we showed that there is an obstruction towards existence of bivector term.
\end{enumerate}

The  canonical BV  transformations from definition \ref{def_can_bv_transform} can mix fields and antifields as in proposition \ref{prop_can_BV_transfrom_exp_action}.   Hence, our refined  classification should be applied  to the representative with the smallest degree in antifields in a given class of BV equivalent solutions. 
In particular, the BV action (\ref{eq_bv_action_saddle_integral_full}) from proposition \ref{prop_can_BV_transfrom_exp_action} after the BV canonical transformation has a trivial dependence in antifields.

\section{Supersymmetric localization in BV formalism}

\subsection{Off-shell  localization}

In this section we recast the localization theorem from section \ref{sec_localization} in the language of BV formalism, to prepare for its  generalization in the next section.

\begin{Theorem} For a quantum physical  system with off-shell odd symmetry $Q$ and an odd function $V$ on supermanifold $X$ the deformed partition function
 \be
 Z(t) = \int_X \mu_X\;\; e^{-\frac{1}{\hbar}(S (\phi) + t Q(V))}
 \ee
  is independent of $t$.
 \end{Theorem}
 \begin{proof}   According to the proposition \ref{prop_cls_master_solution_off_shell},  quantum physical system with off-shell odd symmetry $Q$  describes a solution to the quantum master equation  
 \be\label{eq_master_eq_solution_localiz}
 \mathcal{S} = S (\phi) + \ve Q^a   \phi^\ast_a
 \ee
 for an odd symplectic space $\mathcal{M} = T^\ast[1]X$ with   Berezinian $\mu = \rho^2(\phi)\; d^n \phi \; d^n \phi^\ast$ on $\mathcal{M}$, constructed using the integration measure 
 $\mu_X = \rho(\phi) d^n \phi$ on $X$. Let us set $\ve=1$ and use an  odd function $V$ to define a  family of Lagrangian submanifolds $\mathcal{L}_{tV}:  \phi^\ast_a =   t \p_a V$ on $\mathcal{M}$.  
The evaluation of BV action for $\ve=1$ on $\mathcal{L}_{tV}$ is the $t$-deformed action, i.e. 
 \be
 \begin{split}
  \mathcal{S}\Big|_{\mathcal{L}_{tV}} &= S (\phi) +  \ve Q^a  \phi^\ast_a \Big|_{\mathcal{L}_{tV}} = S (\phi)  +  t  Q^a \p_a V    = S +tQ(V).
  \end{split}
 \ee
The BV integral  (\ref{eq_def_bv_int}) over Lagrangian submanifold $\mathcal{L}_{tV}$ is identical to the deformed partition function, i.e. 
  \be
\int_{ \mathcal{L}_{tV}}  \sqrt{\mu}\; e^{-\frac{1}{\hbar}\mathcal{S}}  =  \int_{X}  \mu_X\; e^{-\frac{1}{\hbar}(S (\phi)  + t Q(V) )}  = Z(t).
 \ee
  Since  the BV action (\ref{eq_master_eq_solution_localiz}) solves the quantum master equation and  $\mathcal{L}_{tV}$ describes  the  family of Lagrangian submanifolds,  related by a continuous deformation,  then  the invariance  theorem \ref{thm_bv_invariance} implies that the  BV integral is independent on $t$.
 \end{proof}

\subsection{On-shell localization with bivector}

We can extend the of-shell localization in BV description  to the physical systems with on-shell supersymmetry of bivector type.

\begin{Theorem}\label{thm_bivector_localization} For a system with a  on-shell supersymmetry of  bivector  type (definition \ref{def_quant_suystem_bivector_prop}) the deformed partition function
 \be
 Z(t)  = \int_X \mu_X\;\; e^{-\frac{1}{\hbar}( S (\phi) +2 t   Q(V) + t^2  \p_a V   \p_b V \Pi^{ab} )}
 \ee
  is independent of the parameter  $t$ for all choices of an odd function $V$.
 \end{Theorem}
 \begin{proof}  It is   illuminating to provide two proofs for the theorem. The first proof relies on the BV integration, while the second proof  uses just  the defining properties for a system with an on-shell supersymmetry of bivector type
 according to the definition \ref{def_quant_suystem_bivector_prop} .
 
 {\it Proof via the BV formalism}:
  According to the proposition \ref{eq_quant_bv_action_bivector_refined}, quantum physical system with off-shell odd symmetry $Q$  describes  solution to the quantum master equation  
 \be\label{eq_master_eq_on_shell_solution_localiz}
 \mathcal{S} = S (\phi) + \phi^\ast_a \ve Q^a + \frac14   \phi^\ast_a \phi^\ast_b  \ve^2 \Pi^{ab} 
 \ee
 for an odd symplectic space $\mathcal{M} = T^\ast[1]X$ with   Berezinian $\mu = \rho^2(\phi)\; d^n \phi \; d^n \phi^\ast$ on $\mathcal{M}$, constructed using the integration measure 
 $\mu_X = \rho(\phi) \; d^n \phi$ on $X$. Let us set $\ve=2$ and use an  odd function $V$ to define a  family of Lagrangian submanifolds $\mathcal{L}_{tV}:  \phi^\ast_a =  t \p_a V$ on $\mathcal{M}$.  
 The evaluation of BV action for $\ve=1$ on $\mathcal{L}_{tV}$
 \be
 \begin{split}
  \mathcal{S}\Big|_{\mathcal{L}_{tV}} &=S (\phi) +2   Q^a \phi^\ast_a  \Big|_{\mathcal{L}_{tV}}+   \phi^\ast_a \phi^\ast_b    \Pi^{ab} \Big|_{\mathcal{L}_{tV}}   =S (\phi) +2 t Q^a \p_a V +   t^2  \p_a V   \p_b V \Pi^{ab}.     
  \end{split}
 \ee
The BV integral  (\ref{eq_def_bv_int}) over Lagrangian submanifold $\mathcal{L}_{tV}$ is identical to the deformed partition function, i.e. 
  \be
\int_{ \mathcal{L}_{tV}}  \sqrt{\mu}\; e^{-\frac{1}{\hbar}\mathcal{S}}  =  \int_{X}  \mu_X\; e^{-\frac{1}{\hbar}(S (\phi)  + t Q(V) + t^2   \p_a V   \p_b V \Pi^{ab} )}  = Z(t).
 \ee
  Since  the BV action (\ref{eq_master_eq_on_shell_solution_localiz}) solves the quantum master equation and  $\mathcal{L}_{tV}$ form a  family of Lagrangian submanifolds  related by a continuous deformation  then  the invariance  theorem \ref{thm_bv_invariance} implies that the  BV integral is independent on $t$.
 
{\it Proof using the properties of bivector}: 
 Let us evaluate the  $t$-derivative
\be\nn
\begin{split}
-\hbar \p_t Z(t) &= -\hbar \p_t \int_X \mu_X\;\; e^{-\frac{1}{\hbar}(S (\phi) + 2t  Q(V) +   t^2     \p_a V  \p_b V  \Pi^{ab})}  =   \int_X \mu_X\;\; (  2Q(V) + 2 t   \p_a V  \p_b V  \Pi^{ab})e^{-\frac1\hbar S_t}   \\
&  =   \int_X\mu_X\;\;    2Q(V   e^{-\frac1\hbar S_t} )     +\int_X \mu_X\;\; V   2Q(  e^{-\frac1\hbar S_t} )    +  2t    \int_X \mu_X\;\; \p_a  (    V   \p_b V   \Pi^{ab} e^{-\frac1\hbar S_t} )  \\
&\qquad- 2t    \int_X \mu_X\;\;  (-1)^{|\phi^a|}   V \p_a \p_b V \Pi^{ab} e^{-\frac1\hbar S_t}     - 2t\int_X \mu_X\;\; (-1)^{|\phi^a||\phi^b|}     V   \p_b V  \p_a  \Pi^{ab}  e^{-\frac1\hbar S_t}   \\
&\qquad-  2t  \int_X \mu_X\;\;  (-1)^{|\phi^a||\phi^b| + |\Pi_{ab}||\phi^a|}   V  \p_b V  \Pi^{ab}   \p_a e^{-\frac1\hbar S_t} , 
\end{split}
\ee
where we introduced 
\be
S_t = S (\phi) + 2t  Q(V) +   t^2     \p_a V  \p_b V  \Pi^{ab}.
\ee
Let us carefully evaluate all terms: The first term in the second line  vanishes due to the   invariance of the integration measure under the action of  supersymmetry vector field, i.e. 
\be
\int_X\mu_X\;\;    2Q(V   e^{-\frac1\hbar S_t})    = -2\int_X\mu_X\;\;  \hbox{div}_{\mu_X} Q \cdot V  e^{-\frac1\hbar S_t} = 0.
\ee
The fourth term is  the total derivative which we can rewrite using the measure $\mu_X = \rho(\phi) d^n \phi $ in local coordinates 
\be\label{eq_total_deriv_rewrite}
2t \int_X \mu_X\;\; \p_a  (    V   \p_b V   \Pi^{ab} e^{-\frac1\hbar S_t})  = -2t \int_X \mu_X\;\;  (-1)^{|\phi^a||\phi^b|}  V   \p_b V  \p_a \ln \rho \cdot \Pi^{ab} e^{S_t}.
 \ee
The first term on the third line vanishes due to the symmetry arguments. In particular,  we  can apply the relabeling and  reshuffling of derivatives to get the same expression with an additional minus sign
\be
\begin{split}
 (-1)^{|\phi^a|}  &   \p_a \p_b V   \Pi^{ab}  =  (-1)^{|\phi^b|}    \p_b \p_a V  \Pi^{ba}    = (-1)^{|\phi^b|+(|\phi^a|+1)(|\phi^b|+1)+|\phi^a||\phi^b|} \p_a \p_b V \Pi^{ab} \\
  & =  - (-1)^{|\phi^b|}    (-1)^{|\phi^a|+|\phi^b|}  \p_a \p_b V \Pi^{ab} =  - (-1)^{|\phi^a|}    \p_a \p_b V \Pi^{ab} .
 \end{split}
\ee
The second term in the third line combines with the total derivative term (\ref{eq_total_deriv_rewrite}) and  vanishes due to the  trivial divergence of bivector  (\ref{def_divergence_bivector})
\be
\begin{split}
 (-1)^{ |\phi^a||\phi^b|}\p_a  \Pi^{ab} + (-1)^{ |\phi^a||\phi^b|} \p_a \ln \rho  \cdot  \Pi^{ab}=  0.
 \end{split}
\ee
The remaining terms can be simplified  into the following expression 
\be\label{eq_t_powers_simplification}
\begin{split}
\hbar^2& \p_t Z(t) = \int_X \mu_X\;\; 2V   e^{-\frac1\hbar S_t} \;[ 2tQ^2(V) + t^2  Q  (  \p_a V \p_b V    \Pi^{ab})]\\
&  +  \int_X \mu_X\;\;   V e^{-\frac1\hbar S_t}  \; (-1)^{|\phi^a|}     \p_b V  \Pi^{ab}  ( -2t\p_a S -4t^2 \p_a(Q^c \p_c V)  -2 t^3 \p_a (  \p_c V \p_d V    \Pi^{cd})).
\end{split}
\ee
Let us carefully look at terms at different orders in $t$ in equation (\ref{eq_t_powers_simplification}). The linear  order simplifies into 
\be\label{eq_t_exp_linear_order}
\begin{split}
2t \int_X \mu_X\;\; V   e^{-\frac1\hbar S_t}  [ 2Q^2(V) -(-1)^{|\phi^a|}     \p_b V  \Pi^{ab}\p_a S] =0.
\end{split}
\ee
An expression (\ref{eq_t_exp_linear_order}) vanishes due to the   on-shell  refinement  (\ref{eq_susy_onshell_refined}) of the on-shell supersymmetry algebra for the vector field $Q$. 

The  quadratic order in $t$-expansion  can be further simplified
\be\label{eq_quadrat_order_t_exp}
\begin{split}
&2t^2\int_X \mu_X\;\;  V   e^{-\frac1\hbar S_t} \;[  Q   (     \p_a V  \p_b V  \Pi^{ab})-2  (-1)^{|\phi^a|}    \p_b V  \Pi^{ab}   \p_a(Q^c \p_c V)  ]   \\
& =2t^2 \int_X \mu_X\;\; V e^{-\frac1\hbar S_t} \;[  (-1)^{|\phi^a|+|\phi^b|} \p_a V \p_b V Q^c \p_c \Pi^{ab} -2  (-1)^{|\phi^a|}    \p_b V  \Pi^{ab}   \p_a Q^c  \p_c V    ] =0.
\end{split}
\ee
The last equality in (\ref{eq_quadrat_order_t_exp}) follows from the equality (\ref{def_pi_2_algebra_inv}).

The cubic  term in the  $t$-expansion  can be further simplified
\be\label{eq_cubic_order_t_expansion}
\begin{split}
-2t^3& \int_X \mu_X\;\;  e^{-\frac1\hbar S_t}  \; (-1)^{|\phi^a|}   V  \p_b V  \Pi^{ab}   \p_a (  \p_c V \p_d V    \Pi^{cd})  \\
&=4t^3 \int_X \mu_X\;\;  e^{-\frac1\hbar S_t}  \; (-1)^{ |\phi^c||\phi^a| +|\phi^c| }   V    \p_a   \p_c V   \p_b V  \Pi^{ab}   \p_d V    \Pi^{cd}\\
&\qquad -2t^3 \int_X \mu_X\;\; e^{-\frac1\hbar S_t}  \; (-1)^{|\phi^a|}     \p_c V \p_d V \cdot  V  \p_b V  \Pi^{ab}     \p_a \Pi^{cd}  =0.
\end{split}
\ee
The term in the second line with two derivatives vanishes due to the symmetry of the sign factors. In particular, we can change the labels $(a\leftrightarrow c)$ and rearrange the  expression to get the same expression with a minus sign
\be
\begin{split}
&(-1)^{ |\phi^c||\phi^a| +|\phi^c| }   V    \p_a   \p_c V   \p_b V  \Pi^{ab}   \p_d V    \Pi^{cd}  = (-1)^{ |\phi^c||\phi^a| +|\phi^a| }   V    \p_c   \p_a V   \p_d V  \Pi^{cd}   \p_b V    \Pi^{ab} \\
&\qquad\qquad =-(-1)^{ |\phi^c||\phi^a| +|\phi^c| }   V    \p_a   \p_c V   \p_b V  \Pi^{ab}   \p_d V    \Pi^{cd}.
\end{split}
\ee
The last term in (\ref{eq_cubic_order_t_expansion}) vanishes due to the (\ref{def_pi_2_bracket_trivial}) relation for bivector.  We showed that all orders in $t$-expansion of (\ref{eq_t_powers_simplification}) vanish, hence  the proof of the theorem is complete.
\end{proof}

Comparing two proofs we observe that the  BV formalism is way more efficient in proving localization-like formulas! However, the existence  of the second proof allows us to formulate and proof  the {\bf new localization theorem}  without any reference to the  BV formalism. Given a refined form  (\ref{eq_susy_onshell_refined}) of the on-shell supersymmetry algebra  and as long as the bivector satisfies the definition \ref{def_quant_suystem_bivector_prop}   the corresponding quantum physical system is subject to the localization theorem \ref{thm_bivector_localization}!

\subsection{Localization in a theory of superpotential}

Let us recast an example \ref{ex_off_shell_deform_superpotential} of supersymmetric deformation  using the BV localization formalism.  For a BV description of the off-shell theory with superpotential  (\ref{eq_off_susy_superpotential}) we  set $\ve=1$, $\bar{\ve}=0$ and  choose  Lagrangian submanifold given by the odd function  $V_H = -t\psi H(x)$. Note that the function is identical to the function (\ref{eq_odd_function_choice_off_super}).  Using symplectic form (\ref{eq_sympl_off_shell_superp_model}) the corresponding Lagrangian  submanifold  is given by 
\be\label{eq_lagr_sub_superpotential}
\psi^\ast =  \frac{\p V_H}{\p \psi} =- tH(x),\;\;\; x^\ast  =  \frac{\p V_H}{\p x} = -tH' (x) \psi.
\ee
The BV action restricted to the  Lagrangian sub-manifold (\ref{eq_lagr_sub_superpotential})
\be
\mathcal{S} \big|_{\cl_{H}}= -\frac14   F^2 +  F(W(x)+t H(x)) + \psi  \bar{\psi} (W'(x) +t H'(x) ). 
\ee
The BV action  is the calssical action for the off-shell theory with deformed superpotential $W + tH$, identical to the deformation (\ref{eq_susy_deform_superpotential}).

We can use the same Lagrangian submanifold  (\ref{eq_lagr_sub_superpotential}) to evaluate the restriction of the BV action (\ref{eq_on_susy_superpotential_bivector})  for the on-shell theory of superpotential
\be
\mathcal{S}_W |_{\mathcal{L}_H}=  \psi  \bar{\psi} W'(x)     +  (W + tH   )^2+  tH' (x) \psi   \bar{\psi} = S_{W+tH}.
\ee
The result of restriction is identical to the classical action for the on-shell theory with deformed superpotential $W + tH$. Moreover the deformation is identical to the off-shell case. 

We can use the deformation invariance of the partition function and  a limit $t \to \infty$ to localize  the partition function to  zeroes of $H$!  In particular, for the on-shell supersymmetric system  with superpotential $W=x + \frac12 g x^3$ we can choose $H = x$ to  remove the $g$-dependence of  partition function.

Note that there is a slight restriction on  choice of  function $H$:  functions $H(x)$  and $W(x)$  should have the same behavior at $x \to \pm \infty$.  An extra restriction is due to  non-compactness of the body for  integration supermanifold in partition function, which is a real line in our example.

\subsection{Localization for the quadratic deformation model}

In case of quadratic deformation model (\ref{eq_action_quadr_deform}) we can evaluate  partition function perturbatively using    BV action to the particular order in coupling $g$.   The BV action (\ref{eq_bv_action_quadr_def_model_full}) up to quadratic corrections in $g$ is 
\be
\begin{split}
\mathcal{S} &= x^2+g x^4 +   \psi\bar{\psi} (1+2g x^2)   +    (\ve \bar{\psi} +  \bar{\ve} \psi )x^\ast   +  (2x+4 gx^3)(  \bar{\ve}  \bar{\psi}^\ast - \ve  \psi^\ast )  \\
&\qquad +(1-2g x^2)(  \bar{\ve}  \bar{\psi}^\ast -\ve  \psi^\ast )^2  -  \frac{4gx}{3} (   \bar{\ve}  \bar{\psi}^\ast -\ve  \psi^\ast )^3  - \frac{g}{3}(   \bar{\ve}  \bar{\psi}^\ast -\ve  \psi^\ast )^4 +\cO(g^2).
\end{split}
\ee
Let us choose $\bar{\ve}=0,\;\; \ve=1$ and the  linear Lagrangian submanifold 
\be
\psi^\ast =  \frac{\p V}{\p \psi} =- tx,\;\;\; x^\ast  =  \frac{\p V}{\p x} =- t \psi.
\ee
The restriction of the BV action to the Lagrangian submanifold 
\be\nn
\begin{split}
\mathcal{S}\Big|_{\cl} &= x^2 +g x^4 + \psi\bar{\psi} (1+2g x^2)  +t   \psi   \bar{\psi}  + 2x(tx)  \\
&\qquad +(1-2g x^2)(  tx )^2  -  \frac{4gx}{3} ( tx )^3  - \frac{g}{3}(tx )^4 +\cO(g^2).
\end{split}
\ee
The simplified version
\be
\begin{split} 
\mathcal{S}\Big|_{\cl} &= x^2(1+t)^2 +(1+t) \psi\bar{\psi} + \frac43 (1+t)  g x^4 + 2g x^2\psi\bar{\psi}   -  \frac{1}{3}g (t+1)^4 x^4    +\cO(g^2).
\end{split}
\ee
We can perform the measure-preserving variable change 
\be
x \to x(1+t),\;\; \psi \to \psi (1+t),\;\; \bar{\psi} \to \bar{\psi},\;\;\; \mu_X = dx d\psi d\bar{\psi} \to \mu_X,
\ee
so the classical action simplifies into 
\be
\begin{split}
\mathcal{S}\Big|_{\cl} &= x^2 + \psi\bar{\psi}-  \frac{1}{3}g  x^4   + \frac{4}{3(1+t)^3}  g x^4 + \frac{2g}{(1+t)^3} x^2\psi\bar{\psi}    +\cO(g^2).
\end{split}
\ee
The leading order $g$-correction to the partition function can be expressed using  the super-Gaussian averages (\ref{eq_super_gauss_integrals}) in the form 
\be
Z_t = \int_{\mathcal{L}}  \sqrt{\mu} \;\; e^{-\mathcal{S}\Big|_{\cl}} = \left\< e^{  \frac{1}{3}g  x^4   - \frac{4}{3(1+t)^3}  g x^4 - \frac{2g}{(1+t)^3} x^2\psi\bar{\psi} +\cO(g^2)}\right\>.
\ee
The Gaussian averages  for the leading order correction 
\be
\begin{split}
\p_g\Big|_{g=0} Z_t &= \left(   \frac{1}{3} -\frac{4}{3 (1+t)^3}  \right) \<  x^4\> -\frac{2}{(1+t)^3} \<\psi\bar{\psi}  x^2\> \\
&  =   \frac{3}{4}\left(   \frac{1}{3} -\frac{4}{3 (1+t)^3}  \right)+ \frac12 \cdot \frac{2}{(1+t)^3} =\frac14.\\
\end{split}
\ee
The large $t$-limit of the BV action on a Lagrangian submanifold simplifies to 
\be\label{eq_bv_action_lagrangian_quadr_def_model_limiting}
\begin{split}
\lim_{t\to \infty}\mathcal{S}\Big|_{\cl} &=  x^2+   \psi   \bar{\psi}    -  \frac{1}{3}g  x^4 +\cO(g^2)
\end{split}
\ee
while the partition function matches with the perturbative computation (\ref{eq_quadr_deform_part_fun_pert})
\be\label{eq_qudr_model_large_t_part_function}
\lim_{t\to \infty}Z_t= \left\< e^{  \frac{1}{3}g  x^4 +\cO(g^2)  }\right\> = 1+  \frac{1}{3}g\<x^4\> +\cO(g^2)=  \left(1 + \frac14 g +\cO(g^2)\right).
\ee
For   more general  function  $V =- tH(x)\psi$, the Larangian submanifold is 
\be
 \psi^\ast   =- tH(x),\;\;\;  x^\ast  =- tH'(x) \psi.
\ee
The BV action becomes
\be
\begin{split}
\mathcal{S}\Big|_{\cl_H} &= x^2 +g x^4 + \psi\bar{\psi} (1+2g x^2) +2t  xH(x) +t H'(x)  \psi   \bar{\psi}  \\
&\qquad+ (1-2gx^2)t^2 H^2(x) -  \frac{4}{3}g t^3 xH^3(x) -\frac13 gH(x)^4 t^4+\cO(g^2).
\end{split}
\ee
We can rewrite the BV action 
\be\nn
\begin{split}
\mathcal{S}\Big|_{\cl_H} &= (x+t H(x))^2  + \psi\bar{\psi} (1+t H'(x)  + 2g x^2)  + \frac43 g x^3 (x+t H(x)) -\frac13 g (x+H(x) t)^4+\cO(g^2).
\end{split}
\ee
Using the  change of variables  $y =x+ tH(x)$,\;\; $\psi \to (1+t  H'(x))\psi$, which preserves the integration measure for partition function $\mu_X = dx d\psi d\bar{\psi}$ we can rewrite 
\be\nn
\begin{split}
\mathcal{S}\Big|_{\cl_H} &= y^2  + \psi\bar{\psi}  +   \frac{ 2g x^2}{x+t H(x)} \psi \bar{\psi}  + \frac{4g}{3} x^3y -\frac13 g y^4+\cO(g^2).
\end{split}
\ee
In the limit $t\to \infty$  the value of $x$ for fixed $y$ will go to zero  and the BV action  simplifies to the same expression (\ref{eq_bv_action_lagrangian_quadr_def_model_limiting}).

\subsection{BV localization for the theories with bivector}

\begin{Conjecture} The partition function for the  quantum physical system with on-shell supersymmetry of bivector type (definition \ref{def_quant_suystem_bivector_prop})  simplifies to the  Gaussian integral.
\end{Conjecture}
In this section we present the supporting evidence for this conjecture. For the  case a single supermultiplet  the BV action for the on-shell supersymmetric system of  bivector  type  takes the following schematic form 
\be
\mathcal{S} = S(x,\psi) +  \ve \bar{\psi} x^\ast   - \ve  f (x) \psi^\ast  + \ve^2 \Pi(x) \psi^\ast \psi^\ast.
\ee
The definition of an on-shell  model  with BV bivector implies 
\be
[Q, \Pi\}  =0  \Longrightarrow \bar{\psi}  \p_x \Pi =0, 
\ee
 hence $\Pi(x) = \Pi$ is  constant! Using the linear    Lagrangian submanifold, with generator  $V = -t \psi x$, for the BV integral gives as
the BV action 
\be
\mathcal{S}\Big|_{\mathcal{L}_V} = S(x,\psi) + t \psi  \bar{\psi} +  t xf(x)+ t^2 \Pi x^2.
\ee
The BV integral for the partition function in the limit $t \to \infty$ becomes the Gaussian integral 
\be\label{eq_gauss_part_function_single_supermult}
Z =\lim_{t\to \infty} \int_{\mathbb{R}^{2|1}} dxd\psi d\bar{\psi}\;\; e^{-S(x,\psi) - t\psi \bar{\psi}  -   tx f(x) - t^2 \Pi x^2} = \int dy\; e^{ -  \Pi y^2} = \frac{\sqrt{\pi}}{\sqrt{\Pi}}.
\ee
Alternatively, we can use $V_{ f} = -  t \psi f(x)$ to ensure that  the $\cO(t)$ part of the integral is also Gaussian 
\be
\mathcal{S}\Big|_{\mathcal{L}_{ f}} = S(x,\psi) + t f'(x)\bar{\psi}\psi   +   t f(x)^2+ t^2 \Pi f(x)^2.
\ee
The partition function 
\be
\begin{split}
Z& =\lim_{t\to \infty} \int_{\mathbb{R}^{2|1}} dxd\psi d\bar{\psi}\;\; e^{-S(x,\psi) - tf'(x) \bar{\psi}\psi   -   t f^2(x) - t^2 \Pi f^2(x)}\\
& =\lim_{t\to \infty}  t \int dy\; e^{ - (t+t^2 \Pi) y^2 } =\lim_{t\to \infty} \frac{\sqrt{\pi}}{\sqrt{\Pi} + t^{-1}} =\frac{\sqrt{\pi}}{\sqrt{\Pi}}
\end{split}
\ee
is identical to the case (\ref{eq_gauss_part_function_single_supermult}) with linear Lagrangian submanifold choice.

 \subsection{Localization considerations}\label{sec_localization_refined_conject}

For our refined classification from section let us formulate a conjecture(s) about the supersymmetric  partition function simplification.

\begin{enumerate}
\item  {\bf Linear}. Conventional off-shell localization in BV formalism. The supersymmetric partition function  localizes to  zeroes of the supersymmetric vector field. 
\item {\bf Quadratic}.  We can choose a linear Lagrangian submanifold with a parameter to make the antifields parametrically large, so that the BV integral will be dominated by the Gaussian integral with bivector representing the corresponding quadratic form. Surprisigly many of the  on-shell supersymmetic models from physics \cite{Baulieu:1990uv,Alexandrov:2007pd} belong to this class.
\item {\bf Polynomial}.  We can choose a linear Lagrangian submanifold with a parameter to make the antifields parametrically large, so that the BV integral will be dominated by  polynomial polyvector terms. This  might be useful for the theory of polynomial maps and polynomial inverses in context of  the quantum field theory formulation for the Jacobian conjecture \cite{Abdesselam:2002cy}. 
\item  {\bf Formal series}. We can choose a linear Lagrangian submanifold with a parameter to make the antifields parametrically large so that the BV integral will be dominated by  the  polyvector terms.  In section \ref{sec_Nicolai_map} we show that the corresponding partition function  is the same as the   Nicolai transform for the original supersymmetric sytem. The Nicolai map is often limited to the quadratic actions in fermions  with recent attempts to improve for the quartic actions \cite{Casarin:2023xic}. The BV constructions does not have such limitations, so it can serve as a generalization of the Nicolai map construction to the  non-quadratic actions in fermions.
\item  {\bf Obstruction}.  No  solution to the master equation, so we cannot use the BV integration methods.   
\end{enumerate}

\section{Supersymmetry for one-dimensional  integrals}

In this section we recast the  one-dimensional  integrals as  supersymmetric partition functions and study the corresponding BV actions. 

\subsection{Exponential integral}\label{sec_exp_integral_bv}
Let us consider a textbook exponential  integral 
\be\label{eq_exp_integral}
I_h = \int_{\mathbb{R}} dx\; e^{ - \frac{1}{\hbar}h(x)}.
\ee
We can evaluate this integral using the saddle point approximation as a formal series in $\hbar$
\be
I_h  = \hbar^{\frac12}\sum_{x_0:  h'(x_0) = 0} e^{ - \frac{1}{\hbar}V(x_0)}\;\; \sqrt{\frac{2\pi}{h''(x_0)}} (1 +\cO(\hbar)).
\ee
We can introduce the odd variables  $\psi, \bar{\psi}$ to turn an integral  (\ref{eq_exp_integral}) into a supersymmetric partition function 
\be
I_h(\hbar) =\hbar \int_{\mathbb{R}^{1|2}} dxd\psi d\bar{\psi}\; e^{ - \frac{1}{\hbar}(h(x) +  \psi  \bar{\psi} ) }  =  \int_X \mu_X \; e^{-\frac{S}{\hbar}} =Z_h
\ee
for a   supresymmetric system on $X =\mathbb{R}^{1|2} $ with action 
\be\label{eq_cl_action_exponent}
S_h(x,\psi,\bar{\psi}) = h(x) + \psi \bar{\psi}, 
\ee
and integration measure 
\be
\mu_X = dxd\psi d\bar{\psi}.
\ee
The action (\ref{eq_cl_action_exponent})  is invariant under the  the odd vector fields 
\be\label{eq_exp_vect_fields}
Q = \bar{\psi} \frac{\p}{\p x} -  h' \frac{\p}{\p \psi},\;\;\; \bar{Q} = \psi \frac{\p}{\p x} + h' \frac{\p}{\p \bar{\psi}}.
\ee
Note that the odd vector fields (\ref{eq_exp_vect_fields}) are similar to the vector fields for the superpotential case (\ref{eq_vect_fields_susy}). 
The vector fields (\ref{eq_exp_vect_fields}) preserve the integration measure, i.e. 
\be
\hbox{div}_{\mu_X} Q = \hbox{div}_{\mu_X} \bar{Q} =0.
\ee
\begin{Remark} The zeroes of the vector field (\ref{eq_exp_vect_fields})  are critical points of $h$ and match with the saddle points of the integral (\ref{eq_exp_integral}).
\end{Remark}
The vector fields (\ref{eq_exp_vect_fields}) form the  $d=0\;\;\mathcal{N}=2$ on-shell supersymmetry  algebra,  moreover the refined version of the on-shell condition is similar to the superpotential case (\ref{eq_on_shell_superpotential_refined_algebra}). In particular, 
\be\label{eq_refined_supersymmetry_exponen}
\begin{split}
[Q, Q\} &=-2h''(x) \bar{\psi} \frac{\p}{\p \psi} =- 2h''(x) \frac{\p S_h}{\p \psi }  \frac{\p}{\p \psi},\;\;\;  [\bar{Q}, \bar{Q}\} =2h''(x) \psi \frac{\p}{\p \bar{\psi}}= -2h''(x) \frac{\p S_h}{\p \bar{\psi} }  \frac{\p}{\p \bar{\psi}},\\
 [ \bar{Q},Q\} &=  h''(x) \left(\bar{\psi} \frac{\p}{\p \bar{\psi}} -\psi \frac{\p}{\p \psi}  \right) = h''(x) \frac{\p S_h}{\p \psi }  \frac{\p}{\p \bar{\psi}} +h''(x) \frac{\p S_h}{\p \bar{\psi} }  \frac{\p}{\p \psi}.
\end{split}
\ee
Hence by proposition \ref{prop_on_shell_bv_leading}  we have an approximate solution  classical master equation solution 
\be\label{eq_bv_action_saddle_integral}
\mathcal{S}_h= h(x) +  \psi\bar{\psi} - \ve h'(x) \psi^\ast+   \bar{\ve} h'(x)\bar{\psi}^\ast  + ( \bar{\ve} \psi + \ve \bar{\psi}) x^\ast+\frac12 h''(x)  ( \psi^\ast  \ve -  \bar{\psi}^\ast \bar{\ve} )^2.
\ee
We  can extend an approximate BV action  (\ref{eq_bv_action_saddle_integral}) to the full solution of the quantum master equation
\be
\mathcal{S}_h = h(x) +  \psi\bar{\psi} - \ve h'(x)\psi^\ast +  \bar{\ve} h'(x)  \bar{\psi}^\ast + ( \bar{\ve} \psi + \ve \bar{\psi})x^\ast  +\sum_{k=2}^\infty \frac{h^{(k)}(x)}{k!}   ( \bar{\psi}^\ast \bar{\ve} - \psi^\ast  \ve )^k.
\ee
We will present the detailed proof for the more general case  in proposition \ref{prop_sol_quant_master_exp_preexp}. The sum over polyvectors is just a  Taylor expansion for $h(x+ \bar{\ve} \bar{\psi}^\ast-\ve \psi^\ast )$,  
 hence we can simplify the BV action into
\be\label{eq_bv_action_saddle_integral_full}
\mathcal{S}_h = h(x+\bar{\ve} \bar{\psi}^\ast-\ve \psi^\ast )  + \psi \bar{\psi}+ ( \bar{\ve} \psi + \ve \bar{\psi})x^\ast. 
\ee
 The classical action  (\ref{eq_cl_action_exponent}) is a trivial solution to the quantum master equation because it has trivial antifield dependence. 
\begin{Proposition}\label{prop_can_BV_transfrom_exp_action}  The BV action (\ref{eq_bv_action_saddle_integral_full}) is a canonical BV transformation of (\ref{eq_cl_action_exponent}).
\end{Proposition}
\begin{proof}  
 Let us consider  a diffeomorphism 
\be
\Phi:\;\; x \mapsto x +\bar{\ve} \bar{\psi}^\ast-\ve \psi^\ast ,\;\;\psi \mapsto \psi -\ve x^\ast,\;\;\;  \bar{\psi} \mapsto \bar{\psi} +\bar{\ve} x^\ast,
\ee
which  preserves the odd symplectic form, i.e.
\be
\begin{split}
\Phi^\ast \omega  &=  d(x+\bar{\ve} \bar{\psi}^\ast-\ve \psi^\ast)\wedge dx^\ast - d(\psi -\ve x^\ast)\wedge d\psi^\ast  -d(\bar{\psi} +\bar{\ve} x^\ast) \wedge d\bar{\psi}^\ast  \\
&= dx \wedge dx^\ast - d\psi \wedge d\psi^\ast   - d\bar{\psi}\wedge d\bar{\psi}^\ast . 
\end{split}
\ee
The classical action (\ref{eq_cl_action_exponent}) in the new variables  becomes the BV action  (\ref{eq_bv_action_saddle_integral_full}).
\end{proof}
\begin{Remark}  We are yet to understand the significance of the BV action being related to the classical action by a BV canonical transformation. We suspect that there might be potential implications for the BV localization, but an additional 
investigation is required. 
\end{Remark}

\subsection{Saddle point integral}

We can represent the generic saddle point integral  
\be
\int_{\mathbb{R}} dx\; f(x) e^{ - \frac{1}{\hbar}h(x)} =\hbar \int_{\mathbb{R}^{1|2}} dxd\psi d\bar{\psi} \; e^{ -\frac{1}{\hbar} \left( h(x) + f(x) \psi\bar{\psi} \right)}
\ee
as an on-shell supersymmetric system  on $X = \mathbb{R}^{1|2}$ with the action 
\be
S = h(x) + f(x) \psi\bar{\psi}
\ee
and measure $\mu_X = dx d\psi d\bar{\psi}$.
The supersymmetry vector fields are
\be\label{eq_susy_saddle_int}
Q = \bar{\psi} \frac{\p}{\p x} - \frac{h' (x)}{f(x)} \frac{\p}{\p \psi},\;\;\; \bar{Q} = \psi \frac{\p}{\p x} + \frac{h' (x)}{f(x)} \frac{\p}{\p \bar{\psi}}.
\ee
The supersymmetry  algebra closes on-shell and we can write down a refined version, which is similar to the exponential case (\ref{eq_refined_supersymmetry_exponen}) with the different overall factor.
The overall factor can be restored from a single commutator, for example
\be
[Q, Q\} =-2 \frac{h'' (x) f(x) - h'(x) f'(x)}{f(x)^2} \bar{\psi} \frac{\p}{\p \psi} =-2 \frac{h'' (x) f(x) - h'(x) f'(x)}{f(x)^3} \frac{\p S}{\p \psi} \frac{\p}{\p \psi}
\ee
Hence the approximate solution to the classical master equation has  bivector 
\be\label{eq_bivector_saddle_int}
\Pi^{(2)} =\frac12  (\ve\psi^\ast - \bar{\ve} \bar{\psi}^\ast)^2   \frac{1}{ f} \frac{d}{dx} \frac{h' (x)}{f(x)} .  
\ee
We can construct the higher order terms in antifields to arrive into all order solution.

\begin{Proposition}\label{prop_sol_quant_master_exp_preexp} The  solution to the quantum master equation to all orders in antifields  is 
\be\label{eq_saddle_point_quantum_bv_solution}
\mathcal{S} =  f(x) \psi\bar{\psi} +  ( \bar{\ve} \psi + \ve \bar{\psi}) x^\ast  + \sum_{k=0}^\infty  \frac{ ( \bar{\ve} \bar{\psi}^\ast -\ve\psi^\ast )^k}{k!}   \left(\frac{1}{ f} \frac{d}{dx} \right)^{k}  h(x). 
\ee
\end{Proposition}
\begin{proof} We will prove the proposition using induction in $k$. The classical action, supersymmetry vector fields (\ref{eq_susy_saddle_int}) and the BV bivector (\ref{eq_bivector_saddle_int}) represent  the 
$k=0,1,2$ terms in the sum.  We assume that for all $k\leq n$ the polyvectors are of the form
\be\label{eq_ind_assumpt_polyvect}
\Pi^{(k)}  = \frac{ ( \bar{\ve} \bar{\psi}^\ast-\ve\psi^\ast )^k}{k!}   \left(\frac{1}{ f} \frac{d}{dx} \right)^{k}  h(x) ,
\ee
 The BV bracket between any pair of polyvectors in the form (\ref{eq_ind_assumpt_polyvect}) is trivial, i.e.  $\{ \Pi^{(k)}, \Pi^{(l)}\} = 0$
  and  $\Delta_\mu \Pi^{(k)} = 0$.
 Hence the  equation (\ref{eq_polyvector_equation}) for  the polyvector  $\Pi^{(n+1)}$ becomes a linear problem
\be\label{eq_liner_iteration_polyvectors}
D_S \Pi^{(n+1)} =- \{ \mathcal{Q}, \Pi^{(n)}\}  =(\bar{\ve} \psi +\ve \bar{\psi}) \frac{d}{dx} \Pi^{(n)}.
\ee
The right hand side of the linear problem (\ref{eq_liner_iteration_polyvectors}) is independent on $x^\ast$ hence we can use  the homotopy (\ref{eq_d_s_homotopy})  to evaluate 
\be
\begin{split}
\Pi^{(n+1)} &= -K  \{ \mathcal{Q}, \Pi^{(n)}\}   = \frac{( \bar{\ve}  \bar{\psi}^\ast- \ve  \psi^\ast )^{n+1}}{(n+1)!}  \left(\frac{1 }{f}  \frac{d}{dx}  \right)^{n+1} h(x).
\end{split}
\ee
Our result for $\Pi^{(n+1)}$ matches with the assumptions of induction, hence the proof is complete.
\end{proof}

\subsection{Supersymmetric systems from polynomial class}\label{sec_polynomial_class_bv_system}

Proposition \ref{prop_sol_quant_master_exp_preexp} describes the BV action for the most general on-shell supersymmetric system with a single $d=0$ $\mathcal{N}=2$ on-shell supermultiplet. The structure of the antifield dependence for the 
BV action is determined by  pair of function $f(x)$  and $h(x)$. For a specific choice we can generate  BV action with  polynomial dependence on antifields. 

\begin{Lemma} For a  degree $n$ polynomial $P_n$  and functions  $f(x) = W'(x)$, $h(x) = P_n (W(x))$ the BV action (\ref{eq_saddle_point_quantum_bv_solution}) is a polynomial of degree $n$ in antifields.
\end{Lemma}

\begin{proof} Proposition \ref{prop_sol_quant_master_exp_preexp} gives us  an explicit form of the polyvectors, while the the BV action for $f(x) = W'(x)$ and $h(x) = P_n (W(x))$ simplifies into
\be\label{eq_BV_action_polyn_type}
\mathcal{S} =  W'(x) \psi\bar{\psi} + ( \bar{\ve} \psi + \ve \bar{\psi})  x^\ast + P_n(W(x) +\bar{\ve} \bar{\psi}^\ast-\ve\psi^\ast). 
\ee
Since $P_n(x)$ is a polynomial of degree $n$, then the BV action is also a  polynomial  of degree $n$ in antifields.
\end{proof}
The classical action 
\be\label{eq_polynomial_on_shell_susy_example}
S = P_n(W(x)) + W'(x) \psi\bar{\psi} 
\ee
is a generalization of the on-shell superpotential system from section \ref{sub_sect_onshel_superpotential}. Note that for our choice of $f(x)$  and $h(x)$ the supersymmetry  vector fields (\ref{eq_susy_saddle_int}) do not have poles. In particular 
\be
Q = \bar{\psi} \frac{\p}{\p x} -P_n'(W(x)) \frac{\p}{\p \psi},\;\;\; \bar{Q} = \psi \frac{\p}{\p x} + P_n'(W(x)) \frac{\p}{\p \bar{\psi}}.
\ee
The partition function for the supersymmetric model (\ref{eq_polynomial_on_shell_susy_example})  simplifies,  using the change of variables $y = W(x)$, i.e. 
\be\label{eq_part_fun_polynomial_bv}
Z = \frac{1}{\hbar} \int dx\; W'\; e^{ - \frac{1}{\hbar}P_n(W(x))}  = \frac{1}{\hbar} \int dy\; e^{- \frac{1}{\hbar}P_n (y)}.
\ee
Note that the partition function depends only on the relative homology class for the curve defined by the function $W(x)$. Also,  similarly to the superpotential case, it is  invariant under the deformations of the function $W(x)$.

The partition function (\ref{eq_part_fun_polynomial_bv})  emerges in the $t\to \infty$ limit of the BV localization  for the Lagrangian submanifold generated by the linear function $V =- t x\psi$.
\be
Z = \lim_{t\to \infty} \int_{\cl_{V}}\sqrt{\mu}\;\; e^{-\frac{1}{\hbar}\mathcal{S}}.
\ee
 
\subsection{Nicolai map}\label{sec_Nicolai_map}

Let us adopt the Nicolai map  construction \cite{Nicolai:1979nr,Nicolai:1980jc} to the case of $d=0$ supersymmetric theory with a single supermultiplet. The partition function for the most general action  can be simplified by integrating out the fermionic (odd) varibles 
\be
Z = \int_{\mathbb{R}^{1|2}} dxd\psi d\bar{\psi}\;\; e^{-\frac{1}{\hbar}( h(x) + f(x) \psi \bar{\psi} )}  =\frac{1}{\hbar} \int_{\mathbb{R}} dx\; f(x)\;e^{-\frac{1}{\hbar} h(x) } .
\ee
We can treat $f(x)$ as a Jacobian for the coordinate transformation to the new variable 
\be\label{eq_Nicolai_coord_transfrom}
 y = \int^x  f(x') dx'
\ee
and express the partition function in the form 
\be\label{eq_Nicolai_transfrom_general}
Z = \frac{1}{\hbar} \int_{\mathbb{R}} dy\;\;e^{-\frac{1}{\hbar} \tilde{h}(y) },\;\;\; \tilde{h}(y) = h (x(y)).
\ee
The partition function (\ref{eq_Nicolai_transfrom_general}) is an integral over just even  variables with the standard measure $dy$ for a classical theory with action $\tilde{h}(y)$. However, an explicit form for the action $\tilde{h}(y)$
requires us to invert the coordinate transform (\ref{eq_Nicolai_coord_transfrom}). In our discussion for BV localization we assumed $f(x) = W'(x)$, hence the coordinate transform simplifies to $y = W(x)$    while  the Nicolai map becomes
$\tilde{h}(W(x)) = h(x)$.

The Nicolai map is usually  constructed perturbatively in coupling constant, what can be illustrated by our example of quadratic deformation for the on-shell theory.
\begin{Example} \label{ex_nicolai_map_pert}
Let us perform the Nicolai transform for the  partition function (\ref{eq_part_funct_quadr_deform_exact}).  The integration over fermions produces a factor  $(1+2gx^2)$ which we use as  the Jacobian   of the coordinate transformation, i.e.  
\be
dy = (1+2gx^2)dx = d \left( x + \frac{2g}{3} x^3\right).
\ee
The perturbative  inverse to the leading order in $g$ is 
\be
x = y - \frac{2g}{3} y^3 +\cO(g^2).
\ee
The action for the purely even system 
\be
\begin{split}
\tilde{h}(y) & =h (x(y))= x^2 +gx^4  =  y^2 - \frac13 g y^4+\cO(g^2).
\end{split}
\ee
The leading order  partition function after the Nicolai transform matches with result of the BV localization (\ref{eq_qudr_model_large_t_part_function}) for the same system using the limiting Lagrangian submanifold  generated by the  linear function.
\end{Example}

An example \ref{ex_nicolai_map_pert} is a consequence of a more general relation.

\begin{Proposition} The BV integral for the quantum master equation solution (\ref{eq_saddle_point_quantum_bv_solution})  over the  linear  Lagrangian submanifold  generated by $V =-t \psi x$ at the $t\to \infty$ limit is identical to the Nicolai transform for  partition function (\ref{eq_Nicolai_transfrom_general}), i.e.
\be
\int_{\cl_V} \sqrt{\mu} \; e^{-\frac1\hbar \mathcal{S}}  =   \frac{1}{\hbar} \int_{\mathbb{R}} dy\;\;e^{-\frac{1}{\hbar} \tilde{h}(y) }.
\ee
\end{Proposition}
\begin{proof}
In section \ref{sec_exp_integral_bv} we observed that  the sum over $k$ in the solution to the quantum master equation 
\be
\mathcal{S} =  f(x) \psi\bar{\psi} + x^\ast ( \bar{\ve} \psi + \ve \bar{\psi})   + \sum_{k=0}^\infty  \frac{ (  \bar{\ve} \bar{\psi}^\ast-\ve\psi^\ast)^k}{k!}   \left(\frac{1}{ f} \frac{d}{dx} \right)^{k}  h(x). 
\ee
 when $f(x)=1$  becomes a Taylor expansion and leads to the simplification of the BV action. For $f(x)\neq 1$ we can perform a similar summation if we introduce a new function $W(x)$ such that  $f(x)  = W'(x)$.   The derivatives  simplify into
\be
\frac{1}{ f} \frac{d}{dx} = \frac{1}{ W'(x)} \frac{d}{dx} = \frac{d}{dW}.
\ee
The sum over $k$ is the exponential operator, describing the shift in $W$-variable 
\be
 \sum_{k=0}^\infty  \frac{ ( \bar{\ve} \bar{\psi}^\ast -\ve\psi^\ast )^k}{k!}  \left(\frac{1}{ f} \frac{d}{dx} \right)^{k} = e^{( \bar{\ve} \bar{\psi}^\ast -\ve\psi^\ast ) \frac{d}{dW}}.
\ee
Let us use the function  $\tilde{h}(x)$ from the Nicolai map construction (\ref{eq_Nicolai_transfrom_general}), then the BV action can be written in the following form
\be
\mathcal{S} =  W'(x) \psi\bar{\psi} +( \bar{\ve} \psi + \ve \bar{\psi})  x^\ast   +\tilde{h}(W(x) +  \bar{\ve} \bar{\psi}^\ast-\ve\psi^\ast ).
\ee
We can perform the BV integral for $\ve =1, \;\; \bar{\ve}=0$ over the  linear  Lagrangian submanifold  generated by 
\be
V =-t \psi x.
\ee
The restriction of BV action
\be
\mathcal{S}\Big|_{\mathcal{L}_V} =  W'(x) \psi\bar{\psi} + t \psi \bar{\psi}  +\tilde{h}(W(x) + tx).
\ee
The $t \to \infty$ limit of the BV integral
\be
Z = \int_{\mathbb{R}^{1|2}} dxd\psi d\bar{\psi} \;  e^{-\frac{1}{\hbar}( t \psi \bar{\psi}  +\tilde{h}(W(x) + tx))}  =\frac{1}{\hbar} \int_{\mathbb{R}} tdx \;  e^{-\frac{1}{\hbar}\tilde{h}(W(x) + tx)}  =\frac{1}{\hbar} \int_{\mathbb{R}} dy \;  e^{-\frac{1}{\hbar}\tilde{h}(y)}. 
\ee
\end{proof}
The BV integral for $t=0$  equals to the original partition function for on-shell supersymmetric theory, hence we have an equality 
\be\label{eq_Nicolai_saddle_point_itegral}
Z = \int_{\mathbb{R}^{1|2}} dxd\psi d\bar{\psi} \;\;e^{-\frac{1}{\hbar}( h(x) + f(x) \psi \bar{\psi} )}  = \frac{1}{\hbar} \int_{\mathbb{R}} dy \;  e^{-\frac{1}{\hbar}\tilde{h}(y)}.
\ee
An equality (\ref{eq_Nicolai_saddle_point_itegral}) is the the $d=0$  version  of Nicolai map construction.

\subsection{BV action from superspace}

In section \ref{sec_polynomial_class_bv_system} we described a generalization (\ref{eq_BV_action_polyn_type})  of the on-shell superpotential model (\ref{eq_on_susy_superpotential_bivector}), with the corresponding BV action being  polynomial type. There is a further generalization of the quantum master equation solution, parametrized by the 
pair of function $W$ and $G$.

\begin{Proposition}  For a pair of functions $W(x)$ and  $G(x)$ there is  a solution to the quantum master equation
\be\label{eq_quant_master_solution_general_single_multiplet}
\mathcal{S}=   W'(x) \psi\bar{\psi} +  (\ve \bar{\psi}+ \bar{\ve}\psi)x^\ast + G(W(x)+ \bar{\ve} \bar{\psi}^\ast- \ve \psi^\ast ) .
\ee
\end{Proposition}

\begin{proof} The non-trivial part of BV brackets
\be\nn
\begin{split}
\{  \mathcal{S},\mathcal{S} \}&=  \{  W'(x) \psi\bar{\psi}+   (\ve \bar{\psi}+ \bar{\ve}\psi)x^\ast , G(W(x)+ \bar{\ve} \bar{\psi}^\ast-  \ve \psi^\ast )  \}  \\
& =  \{  W'(x) \psi\bar{\psi}+   (\ve \bar{\psi}+ \bar{\ve}\psi)x^\ast , W(x)+\bar{\ve} \bar{\psi}^\ast - \ve \psi^\ast  \} G'(W(x)+\bar{\ve} \bar{\psi}^\ast- \ve \psi^\ast )  = 0.
\end{split}
\ee
The  last equality follows from the individual BV brackets
\be
\begin{split}
&\{    (\ve \bar{\psi}+ \bar{\ve}\psi)x^\ast , W(x) \}  =- W'(x) (\ve \bar{\psi}+ \bar{\ve}\psi),  \\
& \{  W'(x) \psi\bar{\psi} ,   \bar{\ve} \bar{\psi}^\ast -\ve \psi^\ast \}  = W'(x) (\ve \bar{\psi}+ \bar{\ve}\psi) , \\
& \{     (\ve \bar{\psi}+ \bar{\ve}\psi)x^\ast ,  \ve \psi^\ast - \bar{\ve} \bar{\psi}^\ast  \} =  (\ve \bar{\ve} - \bar{\ve}\ve)x^\ast  =0.
\end{split}
\ee
The BV action (\ref{eq_quant_master_solution_general_single_multiplet}) satisfies the $\Delta_\mu \mathcal{S} = 0$ for the standard Berezinian $\mu = dxdx^\ast d\psi d\psi^\ast d\bar{\psi} d\bar{\psi}^\ast$
in rather trivial way due to antifield derivatives of it being independent on the corresponding fields.

\end{proof}

In section \ref{sec_bv_induction_superpotential} we observed that the BV action of quadratic type (\ref{eq_on_susy_superpotential_bivector}) for the on-shell theory of superpotential is related to the BV action of the linear type for the off-shell theory of superpotential  (\ref{eq_off_susy_superpotential})  by the BV induction (\ref{eq_BV_induction}).   We can further generalize the BV induction to the  family of solutions (\ref{eq_quant_master_solution_general_single_multiplet}) parametrized by the 
pair of functions $W$  and $G$. 

The first step is to generalize the  off-shell supersymmetric actions constructed using the superfield formalism from section \ref{sec_superspace_formalism} by adding an additional function $K$ for an auxiliary field $F$. In particular  
\be\label{eq_superspace_action_higher_der}
S = -\int d\theta d\bar{\theta}  \left(\hat{x}  K(\mathfrak{\bar{D}}\mathfrak{D}\hat{x})  +  H(\hat{x}) \right) = F K(F) + H'(x) F + H''(x) \psi\bar{\psi}.
\ee 
\begin{Remark} Note that the superspace action (\ref{eq_superspace_action_higher_der}) is the  $d=0$ superspace version of the higher derivative theory. 
\end{Remark}
We replace $H'(x) = W(x)$ and use the proposition \ref{prop_quant_master_solution_off_shell}   to write down the corresponding   solution to the quantum master equation
\be\label{eq_general_superfield_master_solution}
\mathcal{S} =  F K(F) + W(x) F + W'(x) \psi\bar{\psi} + (\ve \bar{\psi}+ \bar{\ve}\psi) x^\ast + F (  \bar{\ve} \bar{\psi}^\ast-\ve \psi^\ast).
\ee
\begin{Proposition} The  two families of solutions to quantum master equation (\ref{eq_quant_master_solution_general_single_multiplet}) and (\ref{eq_general_superfield_master_solution}) are related by the BV induction (\ref{eq_BV_induction}), 
while the functions $G$ and $K$ are related by the Laplace-like transform
\be\label{eq_laplace_transfrom_relation}
\mathcal{C}(\hbar) e^{- \frac{1}{\hbar} G(y, \hbar)} = \int  dF\;  e^{- \frac{1}{\hbar} \left( FK(F) + y F \right)}.
\ee
\end{Proposition}
\begin{proof}
The BV induction (\ref{eq_BV_induction}) for the  Lagrangian submanifold $F^\ast=0$  gives us an integral 
\be
\int dF\; \exp-\frac1\hbar \left(F K(F) + W(x) F + W'(x) \psi\bar{\psi} + (\ve \bar{\psi}+ \bar{\ve}\psi)  x^\ast+ F (  \bar{\ve} \bar{\psi}^\ast-\ve \psi^\ast) \right)
\ee
The result of an integration gives us an induced BV action 
\be
\mathcal{S}_{ind} =   W'(x) \psi\bar{\psi} + (\ve \bar{\psi}+ \bar{\ve}\psi)  x^\ast+ G(W(x)+ \bar{\ve} \bar{\psi}^\ast - \ve \psi^\ast),
\ee
where  we introduced a function $G(y, \hbar)$ defined by
\be
\mathcal{C}(\hbar) \cdot e^{- \frac{1}{\hbar} G(y, \hbar)} = \int  dF\;  e^{- \frac{1}{\hbar} \left( FK(F) + y F \right)}.
\ee
Hence the proof is complete.
\end{proof}
The Laplace-like  transform (\ref{eq_laplace_transfrom_relation}) for the linear function $K$ gives a quadratic function $G$ which is $\hbar$-independent.  While for the higher degree polynomial function $F(K)$ the corresponding function $G(y, \hbar)$
is not a polynomial function with a complicated $\hbar$-dependence.  Similarly, the polynomial BV actions ($G$ is polynomial function) even though can be induced from the BV description of the  off-shell supersymmetric system, but the off-shell action for such system will have a very complicated auxiliary field dependence, the function   $K(F, \hbar)$ is not a polynomial and has a non-trivial $\hbar$-dependence.

\begin{Remark}  In the classical   ($\hbar \to 0$)  limit the Laplace-like  transform  (\ref{eq_laplace_transfrom_relation})  becomes  the Legendre transform. The Legendre transform in general will give us a non-trivial $\hbar$-dependence for the $G(y, \hbar)$. However in some cases the Legendre transform can be 1-loop exact, i.e. the $G(y, \hbar)$ is at most linear in $\hbar$. This phenomenon was studied and completely classified for the case of one variable by Kontsevich and Odesskii \cite{Kontsevich:2023glg}.
\end{Remark}

\section*{Acknowledgments}

We are grateful to  Yasha Neiman  and Pavel Mnev   for many discussions on the topics presented in this paper.  The work  A.L. is supported  by the Basic Research Program of the National Research University Higher School of Economics and by Shanghai Institute for Mathematics and Interdisciplinary Sciences. The work of V.L. is   supported by the Arnold Fellowship at the London Institute for Mathematical Sciences.

\bibliography{BV_loc_ref}{}
\bibliographystyle{utphys}

\end{document}